\newif\iffull\fulltrue
\newif\ifdraft\drafttrue
\newif\ifspace\spacefalse
\draftfalse 

\documentclass[acmsmall,10pt,authorversion,nonacm]{acmart}
\usepackage{mathtools}
\usepackage{enumitem}
\usepackage{tcolorbox}

\usepackage[inference]{semantic}
\setnamespace{0mm}
\usepackage{xspace}
\usepackage{commands}
\usepackage{adjustbox}
\usepackage{amsthm}
\usepackage{pdflscape}

\newtheorem{property}{Property}

\setcopyright{acmcopyright}
\copyrightyear{2020}
\acmYear{2020}
\acmDOI{10.1145/1122445.1122456}
\acmBooktitle{}

\usepackage{listings}

\usepackage{amssymb}

\def\withcolor{}

\ifdefined\withcolor
  \definecolor{fstarblue}{rgb}{0.0, 0.0, 1.0}
  \definecolor{haskellstr}{rgb}{0.2, 0.2, 0.6}
  \definecolor{haskellred}{rgb}{1.0, 0.0, 0.0}
  \definecolor{gray_ulisses}{gray}{0.55}
  \definecolor{castanho_ulisses}{rgb}{0.59,0.42,0.15}
  \definecolor{preto_ulisses}{rgb}{0.55,0.28,0.59}
  \definecolor{green_ulises}{rgb}{0.59,0.42,0.15}
\else
	\definecolor{fstarblue}{gray}{0.1}
	\definecolor{haskellstr}{gray}{0.1}
	\definecolor{haskellred}{gray}{0.1}
	\definecolor{gray_ulisses}{gray}{0.1}
	\definecolor{castanho_ulisses}{gray}{0.1}
	\definecolor{preto_ulisses}{gray}{0.1}
	\definecolor{green_ulisses}{gray}{0.1}
\fi

\def\codesize{\small}

\lstdefinelanguage{HaskellUlisses} {
	basicstyle=\ttfamily\codesize,
	sensitive=true,
	morestring=[b]",
	stringstyle=\color{haskellstr},
	basewidth={0.53em},
	showstringspaces=false,
	numberstyle=\codesize,
	numberblanklines=true,
	showspaces=false,
	breaklines=true,
	showtabs=false,
	tabsize=4,
    literate={ {/\\}{{$\land$}}2
             {->}{{$\rightarrow$}}2
			 {<=>}{{$\Leftrightarrow$}}1
             {forall}{{$\forall$}}1
			 {'a}{{$\alpha$}}1
			 {labelty}{{$l$}}1
             {True}{{$\top$}}1
             {False}{{$\bot$}}1
             {~int}{{$\mathbb{Z}$}}1
             {~nat}{{$\mathbb{N}$}}1
			 {==>}{{$\Longrightarrow$}}1
			 {=>}{{$\Rightarrow$}}1
			 {`feq`}{{$\eqinfix$}}1
			 {ka}{{k${}_a$}}1
			 {kb}{{k${}_b$}}1
			 {dollar}{{$\$$}}1
			 {dsl}{{d$_{sl}$}}2
			 {dfs}{{d$_{fs}$}}2
			 {rsl}{{r$_{sl}$}}2
			 {rfs}{{r$_{fs}$}}2
			 {dlm}{{d$_{lm}$}}2
           },
	emph=
	{[1] Set, Level, Axiom, Propositional, Extensionality, Tot, Type, bool, Lemma, ensures, requires, Ifc, IFC, IfcClearance, GlobalInt, GTot
	},
	emphstyle={[1]\color{fstarblue}},
	emph=
	{[2] class, match, with, if, then, else, let, rec, type, val, in, instance, data, measure, where, effect,noeq, private
	},
	emphstyle={[2]\color{castanho_ulisses}},
	emph=
	{[3]
        lattice, value, equals, canFlow, meet, join, bottom, top, 
        lawBot, lawFlowReflexivity, lawFlowAntisymetry, lawFlowTransitivity, 
        lawMeet, lawJoin, labels, 
        lt, lmeet, ljoin, lcanFlow, eq,
        labeled, labeledTCB
	},
	emphstyle={[3]\color{preto_ulisses}\textbf},
	emph=
	{[4]
        Low, Medium, High
	},
	emphstyle={[4]\color{green_ulises}\textbf},
	emph=
	{[5] assume, admit, admitP
	},
	emphstyle=[5]\color{red}\textbf,
	emph={[6] leq, equals, join', c_0, c_1
	},
	emphstyle=[6]\color{green}\textbf,
}

\lstnewenvironment{code}
{\lstset{language=HaskellUlisses}}
{}

\lstnewenvironment{scode}
{\lstset{language=HaskellUlisses,basicstyle=\ttfamily\footnotesize,keepspaces,mathescape}}
{}

\lstnewenvironment{mcode}
{\lstset{language=HaskellUlisses,columns=fullflexible,keepspaces,mathescape}}
{}

\lstnewenvironment{ccode}
{\lstset{language=C,columns=fullflexible,keepspaces,mathescape}}
{}

\lstMakeShortInline[language=HaskellUlisses,mathescape,keepspaces,mathescape,basicstyle=\ttfamily\codesize,breakatwhitespace]|

\newcommand\subst[2]{\ensuremath{[#2/#1]}}
\newcommand\envLookUp[3]{\ensuremath{#1(#2)\ = \ #3}}
\newcommand\fEq[2]{\ensuremath{ #1 \backsimeq #2}}
\newcommand\bEq[2]{\ensuremath{ #1\ \text{==}\ #2}}
\newcommand\propExpr[1]{\ensuremath{\{#1\}}}
\newcommand\extName{\ensuremath{\texttt{funext}}}
\newcommand\hName{\ensuremath{\texttt{h}}\xspace}
\newcommand\kName{\ensuremath{\texttt{k}}\xspace}
\newcommand\lemmaName{\ensuremath{\texttt{lemma}}\xspace}
\newcommand\sublemma{\ensuremath{\textsc{Sub-L}}\xspace}
\newcommand\subsub{\ensuremath{\textsc{Sub-Sub}}\xspace}
\newcommand\subd{\ensuremath{\textsc{Sub-D}}\xspace}
\newcommand\dName{\ensuremath{\texttt{d}}}
\newcommand\dTy{\ensuremath{\texttt{Eq } \alpha}}
\newcommand\subh{\ensuremath{\textsc{Sub-H}}\xspace}

\newcommand\sfun[3]{\ensuremath{{#1}\text{:}{#2}\!\rightarrow\!{#3}}}
\newcommand\sfuna[3]{\ensuremath{{#1}\text{:}({#2}\!\rightarrow\!{#3})}}

\newcommand\as{\ensuremath{\alpha}}
\newcommand\bs{\ensuremath{\beta}}

\newcommand\edom{\ensuremath{d_e}\xspace}

\newcommand\hdom{\ensuremath{d_\hName}\xspace}
\newcommand\hrng{\ensuremath{r_\hName}\xspace}
\newcommand\hTy{\ensuremath{x:\sref{v}{\alpha}{\hdom}\rightarrow\sref{v}{\beta}{\hrng}}}

\newcommand\subk{\ensuremath{\textsc{Sub-K}}\xspace}
\newcommand\kdom{\ensuremath{d_\kName}\xspace}
\newcommand\krng{\ensuremath{r_\kName}\xspace}
\newcommand\kTy{\ensuremath{x:\sref{v}{\alpha}{\kdom}\rightarrow\sref{v}{\beta}{\krng}}}

\newcommand\ldom{\ensuremath{d_p}}
\newcommand\lemmaTy{\ensuremath{x:\sref{v}{\alpha}{\ldom}\rightarrow\propExpr{p}}}

\newcommand\sref[3]{\ensuremath{\{\!\mid\!\!{#3}\!\!\mid\!\}}\xspace}
\newcommand\fref[3]{\ensuremath{\{{#1}:{#2} \mid {#3} \}}\xspace}

\newcommand\tyApp[1]{\ensuremath{@ #1}}
\newcommand\apred{\ensuremath{\kappa_{\alpha}}\xspace}
\newcommand\bpred{\ensuremath{\kappa_{\beta}}\xspace}
\newcommand\tya{\sref{v}{\alpha}{\apred}}
\newcommand\tyde{\ensuremath{\sref{v}{\alpha}{\edom}}}
\newcommand\tyre{\ensuremath{\sref{v}{\beta}{r_e}}}

\newcommand\tyb{\sref{v}{\beta}{\bpred}}
\newcommand\typExt{\ensuremath{\forall a\ b.\typExtDict{a}{b} \Rightarrow f:(\typExtF{a}{b})\rightarrow g:(\typExtG{a}{b}) \rightarrow (\typeExtProp{a}{b}{f}{g}) \rightarrow \typeExtRes{f}{g}}}
\newcommand\typExtA[1]{\ensuremath{\forall b.\typExtDict{#1}{b} \Rightarrow f:(\typExtF{#1}{b})\rightarrow g:(\typExtG{#1}{b}) \rightarrow (\typeExtProp{#1}{b}{f}{g}) \rightarrow \typeExtRes{f}{g}}}
\newcommand\typExtB[2]{\ensuremath{\typExtDict{#1}{#2} \Rightarrow f:(\typExtF{#1}{#2})\rightarrow g:(\typExtG{#1}{#2}) \rightarrow (\typeExtProp{#1}{#2}{f}{g}) \rightarrow \typeExtRes{f}{g}}}
\newcommand\typExtC[2]{\ensuremath{ f:(\typExtF{#1}{#2})\rightarrow g:(\typExtG{#1}{#2}) \rightarrow (\typeExtProp{#1}{#2}{f}{g}) \rightarrow \typeExtRes{f}{g}}}
\newcommand\typExtD[3]{\ensuremath{g:(\typExtG{#1}{#2}) \rightarrow (\typeExtProp{#1}{#2}{#3}{g}) \rightarrow \typeExtRes{#3}{g}}}
\newcommand\typExtE[4]{\ensuremath{ (\typeExtProp{#1}{#2}{#3}{#4}) \rightarrow \typeExtRes{#3}{#4}}}
\newcommand\typExtDict[2]{\ensuremath{\texttt{Eq } #2}}
\newcommand\typExtF[2]{\ensuremath{#1 \rightarrow #2}}
\newcommand\typExtG[2]{\ensuremath{#1 \rightarrow #2}}
\newcommand\typeExtProp[4]{\ensuremath{x:#1 \rightarrow\propExpr{\bEq{#3\ x}{#4\ x}}}}
\newcommand\typeExtRes[2]{\ensuremath{\propExpr{\fEq{#1}{#2}}}}

\newcommand\envBind[2]{\ensuremath{#1 : #2}}

\setcopyright{none}
\begin{document}

\title{How to Safely Use Extensionality in Liquid Haskell \\ (extended version)}

\author{Niki Vazou}
\email{niki.vazou@imdea.org}
\affiliation{%
  \institution{IMDEA Software Institute}
  \city{Madrid}
  \country{Spain}
}

\author{Michael Greenberg}
\email{michael.greenberg@pomona.edu}
\affiliation{%
  \institution{Pomona College}
  \city{Claremont}
  \state{CA}
  \country{USA}
}


\begin{abstract}
  Refinement type checkers are a powerful way to reason about functional programs.
  For example, one can prove properties of a slow, specification
  implementation, porting the proofs to an optimized implementation that behaves the same.
  Without functional extensionality, proofs must relate functions that
  are fully applied.
  When data itself has a higher-order
  representation, fully applied proofs face serious impediments!
  When working with first-order data, fully applied proofs lead to noisome duplication when using
  higher-order functions.

  While dependent type theories are typically consistent with
  functional extensionality axioms, refinement type systems with semantic subtyping
  treat naive phrasings of functional
  extensionality inconsistently, leading to \emph{unsoundness}.
  We demonstrate this unsoundness and develop a new approach to
  equality in Liquid Haskell: we define a propositional equality in a
  library we call \PEq.
  Using \PEq avoids the unsoundness while still proving useful
  equalities at higher types; we demonstrate its use in several case studies.
  We validate \PEq by building a small model and developing its
  metatheory.
  Additionally, we prove metaproperties of \PEq inside Liquid Haskell
  itself using an unnamed folklore technique, which we dub `classy
  induction'.
\end{abstract}

\maketitle
\section{Introduction}\label{sec:intro}

Refinement types have been extensively used to reason about functional programs~\cite{rushby1998subtypes,
xi1998eliminating, constable1987partial, LT2008, mumon}. 
Higher-order functions are a key ingredient of functional programming, 
so reasoning about function equality within refinement type systems is unavoidable. 
For example, \citet{TPE2018} prove function optimizations correct by specifying equalities between fully applied functions.
Do these equalities hold in the context of higher order functions (\eg maps and folds)
or do the proofs need to be redone for each fully applied context?  
Without functional extensionality (a/k/a |funext|), one must duplicate proofs for each higher-order function.
Worse still, all reasoning about higher-order representations of data
requires first-order observations.

Most verification systems allow for function equality by way of functional extensionality, either built-in (\eg Lean) or as an axiom (\eg Agda, Coq).
Liquid Haskell and \Fstar, two major, SMT-based verification systems built on refinement types, 
are no exception: function equalities come up regularly.
But, in both these systems, the first attempt to give an axiom for functional extensionality 
was wrong.\footnote{
See \url{https://github.com/FStarLang/FStar/issues/1542} for \Fstar's initial, 
wrong encoding and \S\ref{sec:related} for \Fstar's different solution.
We explain the situation in Liquid Haskell in \S\ref{sec:inconsistency}.
}
A naive |funext| axiom proves equalities between unequal functions.

Our first contribution is to expose why a naive encoding of |unfext| is \inadequate (\S\ref{sec:inconsistency}). 
At first sight, function equality can be encoded as 
a refinement type stating that for functions |f| and |g|, 
if we can prove that |f x| equals |g x| for all |x|, then the functions 
|f| and |g| are equal:
\begin{mcode}
  $\extName$ :: forall a b. f:(a -> b) -> g:(a -> b) -> (x:a -> {f x $\req$ g x}) -> {f $\req$ g}
\end{mcode}
(The `refinement proposition' 
|{e}| is equivalent to \texttt{\{\_:() $\mid$ e\}}.)
On closer inspection, |$\extName$| does not encode function equality, 
since it is not reasoning about equality on the domains of the functions. 
What if we instantiate the domain type parameter |a|'s refinement to an intersection of the domains of 
the input functions or, worse, to an uninhabited type? Would such an instantiation of |$\extName$| still prove 
equality of the two input functions? 
It turns out that this naive extensionality axiom 
is inconsistent with refinement types: in \S\ref{sec:inconsistency}
we assume this naive $\extName$ and prove |false|---disaster!
We work in Liquid Haskell, but the problem generalizes to any refinement 
type system that allows for semantic subtyping along with refinement polymorphism, i.e., refinements inferred from constraints~\cite{LT2008}.
To be sound, proofs of function equality must carry information about the domain type 
on which the compared functions are equal.

Our second contribution is to define a type-indexed propositional equality as a Liquid Haskell library (\S\ref{sec:eqrt-gadt}),
where the type indexing uses Haskell's GADTs 
and Liquid Haskell's refinement types.
We call the propositional equality |$\PEq$| and find that it adequately reasons about function equality: we can prove the theorems we want, and we can't prove the (non-)theorems we \emph{don't} want.
Further, we prove in Liquid Haskell \emph{itself} that the implementation of |$\PEq$| is an equivalence relation, 
\ie it is reflexive, symmetric, and transitive. 
To conduct these proofs---which go by induction on the structure of
the type index---we applied a heretofore-unnamed folklore proof methodology, which we dub
\textit{classy induction} (\S\ref{sec:gadt-metatheory}). 

Our third contribution is to use |$\PEq$| to prove equalities between functions (\S\ref{sec:eg}). 
As simple examples, we prove optimizations correct as equalities between functions (\ie |reverse|),
work carefully with functions that only agree on certain domains and dependent codomains,
lift equalities to higher-order contexts (\ie |map|), 
prove equivalences with multi-argument higher-order functions (\ie |fold|), 
and showcase how higher-order, propositional equalities can co-exist with and speedup executable code. 
We also provide a more substantial case study, proving the monad laws
for reader monads.

Our fourth and final contribution is to formalize \corelaneq, a core calculus modeling
|$\PEq$|'s two important features: 
type-indexed, functionally extensional propositional equality and refinement types with semantic subtyping (\S\ref{sec:rules}).
We prove that \corelaneq is sound and that propositional equality implies equality in a term model of equivalence (\S\ref{sec:eqrt}).

\section{Functional Extensionality is Inconsistent in Refinement Types}
\label{sec:inconsistency}

Functional extensionality states that 
two functions are equal, if their values are equal 
at every argument: $\forall f, g : A \rightarrow B,  \forall x \in A, f(x) = g(x) \Rightarrow f = g$.
Most theorem provers
consistently admit functional extensionality as an axiom, which we call |funext| throughout.
Admitting |funext| is a convenient way 
to generate equalities on functions and reuse higher order 
proofs. For example, Agda defines 
functional extensionality as below in the standard library:
\begin{mcode}
  Extensionality : (a b : Level) -> Set _ -- Axiom.Extensionality.Propositional
  Extensionality a b =
    {A : Set a} {B : A -> Set b} {f g : (x : A) -> B x} -> ($\forall$ x -> f x $\equiv$ g x) -> f $\equiv$ g
\end{mcode}
Having seen |funext|'s success in other dependently typed languages,
we naively admitted the |funext| axiom below in Liquid Haskell:  
\begin{mcode}
{-@ assume funext :: $\forall$ a b. f:(a->b) -> g:(a->b) -> (x:a -> {f x = g x}) -> {f = g} @-}
funext :: (a -> b) -> (a -> b) -> (a -> ()) -> ()
funext _f _g _pf = () 
\end{mcode}
The |assume| keyword introduces an axiom: Liquid Haskell will accept the refinement signature of |funext|
wholesale and ignore its definition.
Also, note that the |=| symbol in the refinements refers to SMT equality (see \S\ref{smt:equality}).
Our encoding certainly \textit{looks}
like Agda's |Extensionality| axiom.
But looks can be deceiving: in Liquid Haskell,
we can use |funext| to prove |false|. Why?

\medskip

\noindent
Consider two functions on |Integer|s:
the |incrInt| function increases all integers by one;
the |incrPos| function increases positive numbers by one, returning |0| otherwise:
\begin{mcode}
  incrInt, incrPos :: Integer -> Integer
  incrInt n = n + 1
  incrPos n = if 0 < n then n + 1 else 0
\end{mcode}
Liquid Haskell easily proves that 
these two functions behave the same on positive 
numbers: 
\begin{mcode}
  {-@ type Pos = {n:Integer | 0 < n } @-}
  {-@ incrSamePos :: n:Pos -> {incrPos n = incrInt n} @-}
  incrSamePos :: Integer -> ()
  incrSamePos _n = ()
\end{mcode}    
We can use |funext|
to show that |incrPos| and |incrInt| are equal, 
using our proof |incrSamePos| on the domain of positive numbers. 
\begin{mcode}
  {-@ incrExt :: {incrPos = incrInt} @-}
  incrExt :: ()
  incrExt = funext incrPos incrInt incrSamePos  
\end{mcode}
Having |incrExt| to hand, it's easy
to prove that every higher-order use of |incrPos| can be
replaced with |incrInt|, which is much more efficient---it saves us a conditional branch!
For example, |incrMap| shows that mapping over a list with |incrPos| is just the same as mapping over it with |incrInt|.
\begin{mcode}
  {-@ incrMap :: xs:[Pos] -> {map incrPos xs = map incrInt xs} @-}
  incrMap :: [Integer] -> ()
  incrMap xs = incrExt
\end{mcode}
We could prove |incrMap| without function equality, \ie if we only knew |incrSamePos|. To do so, we would write
an inductive proof---and we'd have to redo the proof for every context in which we would rewrite |incrPos| to |incrInt|.
So |funext| is in part about \emph{modularity} and \emph{reuse} in theorem proving.
We don't give a full example here, but |funext| is particularly critical when trying to equate structures that are themselves higher order, like difference lists or streams.

\medskip

\noindent
Unfortunately, |incrExt| makes it \emph{too} easy to prove
equivalences... our system is inconsistent! Here's a proof that 0 is
equal to -4:
\begin{mcode}
  {-@ inconsistencyI :: {incrPos (-5) = incrInt (-5)} @-} -- 0 = -4
  inconsistencyI :: ()
  inconsistencyI = incrExt 
\end{mcode}
What happened here? How can we have that equality... that |0 = -4|?
Liquid Haskell looked at |incrExt| and saw the two functions were equal... without any regard to the domain on which |incrExt| proved |incrPos| and |incrInt| equal!
We \emph{forgot} the domain, and so |incrExt| generates a proof in SMT that those two functions are equal 
on \textit{any} domain.

\medskip

\noindent
So \textbf{\texttt{funext} is inconsistent in Liquid Haskell}! The problem is that \textbf{Liquid Haskell
forgets the domain on which the two functions are proved equal}, remembering only the equality itself.

We can exploit |funext| to find equalities between \emph{any} two
functions that share the same Haskell type on the \emph{empty} domain,
and Liquid Haskell will treat these functions as \emph{universally}
equal. Ouch!

For example, |plus2| below defines a function that increases its 
input by |2| and is obviously not equal to |incrInt| on any nontrivial domain.
\begin{mcode}
  plus2 :: Integer -> Integer 
  plus2 x = x + 2
\end{mcode}
Even so, we can use |funext| to prove that |plus2| behaves the same as
|incrInt| on the empty domain, \ie for all inputs |n| that satisfy
|false|.
\begin{mcode}
  {-@ type Empty = {v:Integer | false } @-}
  {-@ incrSameEmpty :: n:Empty -> {incrInt n = plus2 n} @-}
  incrSameEmpty :: Integer -> ()
  incrSameEmpty _n = ()
\end{mcode}
Now |incrSameEmpty| provides enough evidence for |funext| to show that 
|incrInt| equals |plus2|, which we use to prove another 
egregious inconsistency.
\begin{mcode}
  {-@ incrPlus2Ext :: {incrInt = plus2} @-}
  incrPlus2Ext ::  ()
  incrPlus2Ext = funext incrInt plus2 incrSameEmpty

  {-@ inconsistencyII :: {incrInt 0 = plus2 0} @-} -- 1 = 2
  inconsistencyII :: ()
  inconsistencyII = incrPlus2Ext 
\end{mcode}
Liquid Haskell isn't like most other dependent type theories: we can't
just admit |funext| as phrased. But we still want to prove equalities between
higher-order values! What can we do?

\subsection{Refined, Type-Indexed, Extensional, Propositional Equality}

\mmg{abrupt transition...}

\newcommand\dom[1]{\ensuremath{\mathcal{D}_{#1}}}

If we're going to reason using functional extensionality in Liquid
Haskell, we'll need to be careful to remember the type at which we
show the functions produce equal results.
What domains are involved when we use functional extensionality?

To prove two functions $f$ and $g$ extensionally equal, we must
reason about \emph{four} domains.
Let \dom{f} and \dom{g} be
the domains on which the functions $f$ and $g$ are respectively defined.
Let \dom{p} be the domain on which the two functions are proved equal
and \dom{e} the domain on which the resulting equality between the two functions is found.
In our |incrExt| example above, 
the function domains are |Integer| (|$\dom{f}$ = $\dom{g}$ = Integer|), 
as specified by the function definitions, 
the domain of the proof is positive numbers (|$\dom{p}$ = Pos|),
as specified by |incrSamePos|, and, disastrously, 
the domain of the equality itself is unspecified in |funext|.
Liquid Haskell will implicitly set
the domain on which the functions are equal to the most general one where both functions can be called (|$\dom{e}$ = Integer|). \mmg{is this right?}

Our |funext| encoding naively imposes no real constraints between these domains.
In fact, |funext| only requires that 
\dom{f}, \dom{g}, and \dom{p} are supertypes of the empty 
domain (\S\ref{sec:rules}), which trivially holds for all types, leaving \dom{e} underconstrained.

To be consistent,
we need a functional extensionality axiom that 
(1) captures the domain of function equality \dom{e} explicitly,  
(2) requires that the domain of the equality, \dom{e}, is a subtype of the domain of the proof, \dom{p}, which should be a subtype of the functions domains, \dom{f} and \dom{g}, and 
(3) ensures that the resulting equality between functions is only used on subdomains of \dom{e}.

\medskip \noindent
Our solution is to define a refined, type-indexed, extensional
propositional equality. We do so in the Liquid Haskell library
|$\libname$|, which defines a \underline{p}ropositional \underline{eq}uality also called |PEq|.
We write 
|PEq a {e$_l$} {e$_r$}|
to mean that the expressions |e$_l$| and |e$_r$|
are propositionally equal and of type |a|.
We carefully crafted |PEq|'s definition as a refined GADT
(\S\ref{sec:eqrt-gadt}) to meet our three criteria.

\paragraph{1. \texttt{PEq} is Type-Indexed}
The type index |a| in |PEq a {e$_l$} {e$_r$}| makes it easy to track
types explicitly.
|PEq|'s constructor axiomatizing functional extensionality keeps careful track of types:
\begin{mcode}
  $\XEq$ :: f:(a -> b) -> g:(a -> b) -> (x:a -> $\PEq$ b {f x} {g x}) -> $\PEq$ (a -> b) {f} {g}
\end{mcode}
The result type of $\XEq$ explicitly captures the equality domain as
the domain of the return type (\ie |a|).  The standard variance and
type checking rules of Liquid Haskell ensure that the domains
\dom{f}, \dom{g}, and \dom{p} are supertypes of \dom{e}.
(See \S\ref{sec:rules} for more detail on type checking.)

\paragraph{2. Generating Function Equalities}

The |XEq| case of |PEq| generates equalities at function types
using functional extensionality.
Liquid Haskell will check the domains appropriately: it won't prove
equality between functions at an inappropriate domain.

Returning to our concrete example of |incrPos| and |incrInt|,
we can use |XEq| to find these functions equal
on the domain |Pos|:
\begin{mcode}
{-@ incrExtGood :: $\PEq$ (Pos -> Integer) {incrPos} {incrInt} @-} 
incrExtGood :: $\PEq$ (Integer -> Integer)
incrExtGood = $\XEq$ incrPos incrInt incrEq 
\end{mcode}
|XEq| checks that the domains of the functions |incrPos| and |incrInt|
are supertypes of |Pos|, \ie |Pos <: Integer|. 
Further it checks that the domain of the proof |incrEq| is supertype of |Pos|.

What might we define for |incrEq|? Here are three alternatives.
Each alternative is either accepted or rejected by |XEq| as appropriate for
the |Pos -> Integer| type index;
each alternative is also possible or impossible to prove. (See~\S\ref{sec:eqrt-gadt} for more on  how |incrEq| can be defined.)
\begin{mcode}
incrEq :: n:Pos     -> PEq Integer {incrPos n} {incrInt n} -- ACCEPTED and POSSIBLE
incrEq :: n:Integer -> PEq Integer {incrPos n} {incrInt n} -- ACCEPTED and IMPOSSIBLE  
incrEq :: n:Empty   -> PEq Integer {incrPos n} {incrInt n} -- REJECTED and POSSIBLE
\end{mcode}
The first two alternatives, |n:Pos| and |n:Integer|, will be accepted by |XEq|, 
since both |Pos| and |Integer| are supertypes of |Pos|... 
though it is impossible to actually construct a proof for the second alternative,
\ie a proof that |incrPos n| equals |incrInt n| for all integers |n|.
On the other hand, the last proof on |n:Empty| is trivial, but |XEq| 
rejects it, because |Empty| is not a supertype of |Pos|.
Liquid Haskell's checks on |XEq|'s type indices prevents inconsistencies like |inconsistencyII|.

\paragraph{3. Using Function Equalities}
Just as |PEq|'s |XEq| constructor ensures that the right domains are checked and tracked for functional extensionality, we have a constructor for ensuring these equalities are used appropriately.
The constructor $\CEq$ characterizes equality as valid in all contexts, \ie if |x| and |y| are equal, they can be substituted 
in any context |ctx| and the results
|ctx x| and |ctx y| will be equal:
\begin{mcode}
  $\CEq$ :: x:a -> y:a -> $\PEq$ a {x} {y} -> ctx:(a -> b) -> $\PEq$ b {ctx x} {ctx y}
\end{mcode}
It is easy to use |$\CEq$| to apply functional equalities in higher order
contexts. For example, we can prove that |map incrPos| equals |map incrInt|:
\begin{mcode}
{-@ incrMapProp :: PEq ([Pos] -> [Integer]) {map incrPos} {map incrInt} @-}
incrMapProp :: PEq ([Integer] -> [Integer])
incrMapProp = $\CEq$ incrPos incrInt incrExtGood (map)
\end{mcode}

We can more generally show that propositionally equal functions produce equal results on equal inputs.
The trick is to \emph{flip} the context, 
defining a function |app|
that takes as input two functions |f| and |g|, 
a proof these functions are equal, and an argument |x|, 
returning a proof that |f x = g x|: 
\begin{mcode}
  {-@ app :: f:(a -> b) -> g:(a -> b) -> PEq (a -> b) {f} {g} 
          -> x:a -> PEq b {f x} {g x} @-}
  app :: (a -> b) -> (a -> b) -> PEq (a -> b) -> a -> PEq b 
  app f g eq x = $\CEq$ f g eq (flip x)

  flip x f = f x 
\end{mcode}
The |app| lemma makes it easy to use function equalities while still checking the domain on which the function is applied.
These checks prevent inconsistencies 
like |inconsistencyI|. 
For instance, we can try to apply the functional equality |incrExtGood|
to a bad and a good input.
\begin{mcode}
  {-@ badFO ::PEq Integer {incrPos 0} {incrInt 0} @-}
  badFO = app incrPos incrInt incrExtGood 0  -- REJECTED

  {-@ goodFO :: x:{Integer | 42 < x } -> PEq Integer {incrPos x} {incrNat x} @-}
  goodFO x = app incrPos incrInt incrExtGood x -- ACCEPTED
\end{mcode}

Liquid Haskell rejects the bad input in |badFO|: the number $0$ isn't in the |Pos| domain on which |incrExtGood| was proved.
Liquid Haskell accepts the good input in |goodFO|, since any |x| greater than 42 is certainly positive.
The |goodFO| proof yields a first-order equality on any such |x|, here on |Integer|.
Such first order equalities correspond neatly with the notion of
equality used in the SMT solvers that buttress all of Liquid Haskell's
reasoning.
(For more information on how SMT equality relates to notions of equality in Liquid Haskell, see~\S\ref{sec:eqrt-gadt}. For an example of how these first-order equalities can lead to runtime optimizations, see~\S\ref{sec:eg:spec}.)
\nv{what is SMT comparison...} 
\mmg{this is vague but maybe okay? we can't explain everything at once... and we \emph{don't} want people constantly wondering what each equality is, etc.}

\subsection{Why Isn't \texttt{funext} Inconsistent in Agda?}
At the beginning of \S\ref{sec:inconsistency}, we present 
Agda's |Extensionality| axiom, whose return type is |f $\equiv$ g|.
Agda's equality appears to lack a type index.
Why doesn't Agda also suffer from inconsistency?

Agda's equality only seems to be unindexed. In fact, Agda's built-in equality is the standard, type-indexed Leibniz equality used in most dependent type theories (omitting |Level| polymorphism):
\begin{mcode}
data _$\equiv$_ {A : Set} (x : A) : A $\rightarrow$ Set a where
  refl : x $\equiv$ x
\end{mcode}
The curly braces around the type index |A| marks it as \emph{implicit}, i.e., to be
inferred.
If we were to explicitly give implicit arguments by wrapping them in
curly braces, Agda's extensionality axiom returns
|(_$\equiv$_) {a->b} f g|.

Our $\XEq$ axiom recovers the type indexing in Agda's equivalence that's missing in our original |funext| encoding. 
Of course, (Liquid) Haskell's lack of implicit type 
indices makes reasoning about function equalities verbose. 
On the other hand, Liquid Haskell's subtyping can reinterpret functions at many domains (see~\S\ref{sec:eg:refdom}).
In Agda, however, it is much more complex 
to reinterpret functions and to generate heterogeneous equality relating |incrInt| and |incrPos| 
only on positive inputs.

\section{\texttt{\libname}: a Library and GADT for Extensional Equality}
\label{sec:eqrt-gadt}

We define the |$\libname$|
library in Liquid Haskell,
implementing the type-indexed propositional equality, also called \PEq. 
First, we axiomatize equality for base types in the \AEq typeclass 
(\S\ref{sec:axiom-eq}).
Next, 
we define propositional equality for base and function types with the \PEq GADT~\cite{10.1145/604131.604150, cheneyhinze03gadt} (\S\ref{sec:gadt-eqt}).
Refinements on the GADT enforce the typing rules of our formal model
(\S\ref{sec:eqrt}), but we prove some of the metatheory in 
Liquid Haskell itself (\S\ref{sec:gadt-metatheory}).
Finally, 
we discuss how \AEq and \PEq interact 
with Haskell's and SMT's equalities (\S\ref{subsec:interaction}). 

\subsection{The \AEq typeclass, for axiomatized equality}
\label{sec:axiom-eq}

We begin with by axiomatizing equality that can be ported to SMT: such an equality should be an equivalence relation that implies SMT equality.
We use refinements on typeclasses~\citep{liu20typeclasses} 
to define a typeclass \AEq, which
contains the (operational) equality method $\aEqSym$,
three methods that encode the equality laws, 
and one method that encodes correspondence with SMT equality. 
\begin{mcode}
{-@ class $\AEq$ a where 
     ($\aEqSym$)    :: x:a -> y:a -> Bool
     reflP  :: x:a -> {$\bbEq{x}{x}$}
     symmP  :: x:a -> y:a -> { $\bbEq{x}{y}$ => $\bbEq{y}{x}$ }
     transP :: x:a -> y:a -> z:a -> { ($\bbEq{x}{y}$ && $\bbEq{y}{z}$) => $\bbEq{x}{z}$ } 
     smtP   :: x:a -> y:a -> { $\bbEq{x}{y}$ } -> { x = y } @-}
\end{mcode}
To define an instance of $\AEq$ one has to define the method 
($\aEqSym$) and provide explicit proofs that it is reflexive, 
symmetric, and transitive (|reflP|, |symmP|, and |transP| \resp); 
thus \aEqSym is, by construction, an equality. 
Finally, we require the proof |smtP| that captures that 
($\aEqSym$) implies equality provable by SMT (\eg structural equality).\footnote{
 The three axioms of equality alone are not enough to ensure SMT's structural 
 equality, \eg one can define an instance  \texttt{x} $\equiv$ \texttt{y = True} 
 which satisfies the equality laws, but does not correspond to SMT equality. 
}

\subsection{The \PBEq GADT and its \PEq Refinement}
\label{sec:gadt-eqt}

\begin{figure}[t]
\begin{mcode}
-- (1) Plain GADT
data $\PBEq$ :: * -> *  where 
     $\BEq$  :: $\AEq$ a => a -> a -> () -> $\PBEq$ a 
     $\XEq$  :: (a -> b) -> (a -> b) -> (a -> $\PEq$ b) -> $\PBEq$ (a -> b)
     $\CEq$  :: a -> a -> $\PBEq$ a -> (a -> b) -> $\PBEq$ b 

-- (2) Uninterpreted equality between terms e1 and e2 
{-@ type $\PEq$ a e1 e2 = {v:$\PBEq$ a | e1 `feq` e2} @-}
{-@ measure (`feq`) :: a -> a -> Bool @-}

-- (3) Type refinement of the GADT
{-@ data $\PBEq$  :: * -> *  where 
     $\BEq$ :: $\AEq$ a => x:a -> y:a -> {v:() | $\bbEq{x}{y}$} 
         -> $\PEq$ a {x} {y}
     $\XEq$ :: f:(a -> b) -> g:(a -> b) -> (x:a -> $\PEq$ b {f x} {g x}) 
         -> $\PEq$ (a -> b) {f} {g}
     $\CEq$ :: x:a -> y:a -> $\PEq$ a {x} {y} -> ctx:(a -> b) 
         -> $\PEq$ b {ctx x} {ctx y} @-}
\end{mcode}
  
\caption{Implementation of the propositional equality $\PEq$ as a refinement of Haskell's GADT $\PBEq$.}
\label{fig:eqt}
\end{figure}

We use \AEq to define our type-indexed propositional equality |$\PEq$ a {e1} {e2}| in three steps (Figure~\ref{fig:eqt}):
(1) structure as a GADT, 
(2) definition of the refined type |$\PEq$|, and
(3) axiomatization of equality by refining of the GADT.

First, we define the structure of our proofs of equality as |$\PBEq$|, 
an unrefined, \ie Haskell, GADT (Figure~\ref{fig:eqt}, (1)).
The plain GADT defines the structure of derivations in our
propositional equality (\ie which proofs are well formed), but none
of the constraints on derivations (\ie which proofs are valid).
There are three ways to prove our propositional equality, each
corresponding to a constructor of |$\PBEq$|:
using an |$\AEq$| instance (constructor |$\BEq$|);
using \extName\ (constructor |$\XEq$|);
and by congruence closure (constructor |$\CEq$|).

Next, we define the refinement type |$\PEq$| to be our propositional
equality (Figure~\ref{fig:eqt}, (2)).
Two terms |e1| and |e2| of type |a| are propositionally
equal when (a) there is a well formed and valid |$\PBEq$| proof and (b) we
have |e1 `feq` e2|, where |(`feq`)| is an \textit{uninterpreted} SMT function.
%
%
  Liquid Haskell uses curly braces for expression arguments
  in type applications, \eg in \texttt{\PEq a \{x\}
    \{y\}}, \texttt{x} and \texttt{y} are expressions, but
  \texttt{a} is a type.
%

Finally, we refine the type constructors of |$\PBEq$| to axiomatize the uninterpreted |(`feq`)|
and generate proofs of |$\PEq$| (Figure~\ref{fig:eqt}, (3)).
Each constructor of |$\PBEq$| is refined to return something of type
|$\PEq$|, where |$\PEq$ a {e1} {e2}| means that terms |e1| and |e2| are
considered equal at type |a|. 
|$\BEq$| constructs proofs that two terms, |x| and |y| of type |a|, are
equal when |$\bbEq{x}{y}$| according to the |$\AEq$| instance for |a|. 
%
|$\XEq$| is the (type-indexed) \extName\ axiom. Given functions |f| and
|g| of type |a -> b|, a proof of equality via extensionality also
needs a |$\PEq$|-proof that |f x| and |g x| are equal for all |x| of
type |a|. Such a proof has refined
type |x:a -> $\PEq$ b {f x} {g x}|. Critically, we don't lose any type information about |f| or |g|!
|$\CEq$| implements congruence closure: 
for |x| and |y| of type |a|
that are equal---\ie |$\PEq$ a {x} {y}|---and an arbitrary context
with an |a|-shaped hole (|ctx :: a -> b|), filling the context with
|x| and |y| yields equal results, \ie
|$\PEq$ b {ctx x} {ctx y}|.

\paragraph{Design Alternatives}
The first design choice we made was to define  
|$\PEq$| as a GADT and not an axiomatized opaque type. 
While there's no reason to pattern match on |$\PEq$| terms, there's also no harm in it.
A GADT provides a clean interface on how |$\PEq$| 
can be generated: it collects all the axioms as data contructors 
and prevents the user from arbitrarily adding new constructors. 
The second choice we made was to define the type $\PEq$
using a fresh uninterpreted equality symbol (Figure~\ref{fig:eqt}, (2)) 
instead of SMT equality. 
Again, we made this decision to ensure that all $\PEq$ terms are constructed 
via the constructors 
and not implicit SMT automation. 
The final choice we made was to define the base case using the |$\AEq$| constraints. 
We considered two alternatives:
\begin{mcode}
  $\BEq$ ::         x:a -> y:a -> {v:() | x = y } -> $\PEq$ a {x} {y} -- alternative I 
  $\BEq$ :: Eq a => x:a -> y:a -> {v:() | x = y } -> $\PEq$ a {x} {y} -- alternative II
\end{mcode}
We rejected the first to ensure that the base case does not include
functions (which don't generally have |Eq| instances) and to support our metatheory (\S\ref{sec:gadt-metatheory}). 
We rejected the second to exclude user-defined |Eq| instances that do not correspond to SMT 
equality (since in~\S\ref{subsec:interaction} we define a machanism to turn $\PEq$ to SMT equalities).

\paragraph{Example:}
Having seen |AEq| and the |BEq| case of $\PEq$, 
we can define the |incrEq| function from~\S\ref{sec:inconsistency}:
\begin{mcode}
  {-@ incrEq :: x:Pos -> PEq Integer {incrPos x} {incrInt x} @-}
  incrEq x = $\BEq$ (incrPos x) (incrInt x) (reflP (incrPos x))
\end{mcode}
We start from |reflP (incrPos x) :: {incrPos x$\bbEq{}{}$incrPos x}|, 
since |x| is positive, the SMT derives |incrPos x = incrInt x|,
generating the $\BEq$ proof term |{incrPos x$\bbEq{}{}$incrInt x}|.

\subsection{Equivalence Properties and Classy Induction}
\label{sec:gadt-metatheory}

We can  prove metaproperties of the actual implementation of 
$\PEq$---reflexivity, symmetry, and transitivity---within Liquid Haskell itself.

Our proofs in Liquid Haskell go by induction on types. 
But ``induction'' in Liquid Haskell means writing a recursive function,
which necessarily has a single, fixed type.
To express that $\PEq$ is reflexive, 
we want a Liquid Haskell theorem |refl :: x:a -> $\PEq$ a {x} {x}|,
but its proof goes by induction
on the type |a|, which is not possible in ordinary Haskell 
functions.\footnote{A variety of GHC extensions allow case
  analysis on types (\eg type families and generics), but, 
  unfortunately, Liquid Haskell doesn't support such fancy type-level programming.}

The essence of our proofs is a folklore method we call \emph{classy
  induction} (see \S\ref{sec:related} for the history).
To prove a theorem using classy induction on the |$\PEq$| GADT, one must:
(1) define a typeclass with a method whose refined type corresponds to
the theorem;
(2) prove the base case for types with |$\AEq$| instances;
and (3) prove the inductive case for function types, where typeclass constraints on smaller types generate inductive
hypotheses.
All three of our proofs follow this pattern.
Here's the proof 
for reflexivity.
%
\begin{mcode}
-- (1) Refined typeclass
{-@ class Reflexivity a where 
  refl :: x:a -> $\PEq$ a {x} {x} @-}

-- (2) Base case (AEq types)
instance $\AEq$ a => Reflexivity a where
  refl a = $\BEq$ a a (reflP a)
-- (3) Inductive case (function types)
instance Reflexivity b => Reflexivity (a -> b) where
  refl f = $\XEq$ f f (\a -> refl (f a))
\end{mcode}
%
%
For (1), the typeclass |Reflexivity| simply states the desired
theorem type,
|refl :: x:a -> $\PEq$ a {x} {x}|.
For (2), given an |$\AEq$ a| instance, |$\BEq$| and the |reflP| method are combined 
to define the |refl| method.
To define such a general instance, we
enabled the GHC extensions \texttt{FlexibleInstances} and
\texttt{UndecidableInstances}.
%
For (3), |$\XEq$| can show that |f| is equal to itself by using the
|refl| instance from the codomain constraint: the |Reflexivity b|
constraint generates a method |refl :: x:b -> $\PEq$ b {x} {x}|. The
codomain constraint |Reflexivity b| corresponds exactly to the inductive hypothesis 
on the codomain: we are doing induction!

At compile time, any use of |refl x| when |x| has type |a| asks the
compiler to find a |Reflexivity| instance for |a|.
If |a| has an |$\AEq$| instance, the proof of |refl x| will simply be
|$\BEq$ x x (reflP a)|.
If |a| is a function of type |b -> c|, then the compiler will try to
find a |Reflexivity| instance for the codomain |c|---and if it finds
one, generate a proof using |$\XEq$| and |c|'s proof.
The compiler's constraint resolver does the constructive proof for us,
assembling the `inductive tower' to give us a |refl| for our chosen
type.
That is, even though
Liquid Haskell can't mechanically check that our inductive proofs are in
general complete (\ie the base and inductive cases cover all types),
our |refl|
proofs will work for types where the codomain bottoms out with an |$\AEq$|
instance, i.e., any type consisting of functions and |AEq|-equable types.

Our proofs of symmetry and transitivity follow the same pattern, but both also make use congruence closure.
The full proofs can be found in~\citet{implementation}.
%
Here is the inductive case from symmetry:
%
\begin{mcode}
instance Symmetry b => Symmetry (a -> b) where
-- sym :: l:(a->b) -> r:(a->b) -> PEq (a->b) {l} {r} -> PEq (a->b) {r} {l}
   sym l r pf = $\XEq$ r l dollar \a -> sym (l a) (r a) ($\CEq$ l r pf (dollar a) ? (dollar a l) ? (dollar a r)))
\end{mcode}
Here |l| and |r| are functions of type |a -> b| and we know that
|l `feq` r|; we must prove that |r `feq` l|.
We do so using:
(a) |$\XEq$| for extensionality, letting |a| of type |a| be given;
(b) |sym (l a) (r a)| as the IH on the codomain |b| on 
(c) |$\CEq$| for congruence closure on |l `feq` r| in the context
|(dollar a)|.
The last step is the most interesting: if |l| is equal
to |r|, then plugging them into the same context yields
equal results; as our context,
we pick |(dollar a)|, \ie |\f -> f a|, showing that |l a `feq` r a|;
the IH on the codomain |b| yields |r a `feq` l a|,
and extensionality shows that |r `feq` l|, as desired.
The operator |?|, defined as |x ? p = x|, asks Liquid Haskell to encode `p` into the SMT solver to help prove `x`.
Our use of |?| unfolds the definitions |dollar a l| and |dollar a r| to help |CEq|.

\subsection{Interaction of the different equalities.}
\label{subsec:interaction}
\label{smt:equality}

\nv{we need to bring back the SMT equality discussion from POPL submission}

We have four equalities in our system (Figure~\ref{fig:equalities}):
SMT equality ($=$), the ($\aEqSym$) method of the \AEq typeclass(\S\ref{sec:axiom-eq}), 
the refined GADT \PEq (\S\ref{sec:gadt-eqt}), 
and the (|==|) method of Haskell's |Eq| typeclass.

\paragraph{SMT Equality}
The single equal sign ($=$) represents SMT equality, which
satisfies the three equality axioms and is syntactically 
defined for data types. 
The SMT-LIB standard~\cite{BarST-RR-10}
permits the equality symbol on functions
but does not specify its behavior.
Implementations vary.
CVC4 allows for functional extensionality and higher-order
reasoning~\cite{Barbosa19}.  When Z3 compares functions for equality,
it treats them as arrays, using the extensional array theory to
incompletely perform the comparison. When asked if two functions are
equal, Z3 typically answers |unknown|.
To avoid this unpredictability, our system avoids SMT equality 
on functions. 

\paragraph{Interactions of Equalities}
SMT equalities are internally generated by Liquid Haskell
using the reflection and PLE tactic of~\citet{VazouTCSNWJ18} 
(see also \S\ref{sec:eg:reverse}). 
An $\bbEq{e_1}{e_2}$ equality can be generated one of three ways: 
(1) If SMT can prove an SMT equality $e_1 = e_2$, then the reflexivity 
|reflP| method can generate that equality, \ie |reflP $e_1$| proves $\bbEq{e_1}{e_1}$, which is enough to show $\bbEq{e_1}{e_2}$.  
(2) Our system provides $\AEq$ instances for the primitive Haskell types 
using the Haskell equality that we \textit{assume} satisfies the four laws, 
\eg the |instance $\AEq$ Int| is provided. 
(3) Using refinements in typeclasses~\cite{liu20typeclasses}
one can explicitly define instances of $\AEq$, which may or may not coincide with Haskell |Eq| instances. 

Constructors generate $\PEq$ proofs, bottoming out at |AEq|:
$\BEq$ combined with an $\AEq$ term and 
$\XEq$ or $\CEq$ combined with other $\PEq$ terms. 

Finally, we define a mechanism to convert $\PEq$ into an 
SMT equality. 
This conversion is useful when
we want to derive an SMT equality $f\ e = g\ e$
from a function equality |PEq (a -> b) {f} {g}| (see~\S\ref{sec:eg:spec}). 
The derivation requires that the domain |b| admits 
the axiomatized equality, |AEq|.
To capture this requirement we define
|toSMT| that converts |PEq| to SMT equality 
as a method of a class that requires an |AEq|
constraint: 
\begin{mcode}
  class AEq a => SMTEq a where 
    toSMT :: x:a -> y:a -> PEq a {x} {y} -> {x = y} 
\end{mcode}

\begin{figure}
  \includegraphics[width=0.8\columnwidth]{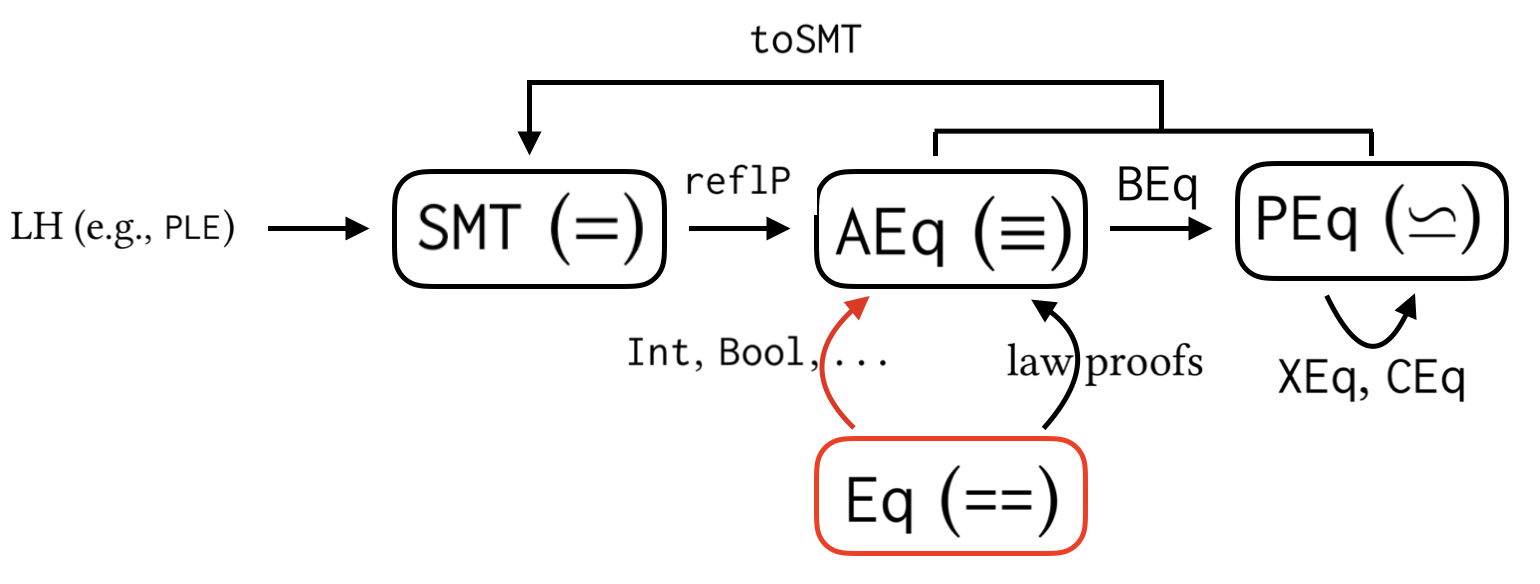}
\caption{The four different equalities and their interactions. Haskell equality is in \textcolor{red}{red} to highlight its potential for unsoundness.}
\label{fig:equalities}
\end{figure}

\paragraph{Non-interaction}
Liquid Haskell
maps Haskell's |(==)| to SMT equality by default.
It is surely unsound to do so, as users can define their own |Eq| instances with |(==)| methods that do arbitrarily strange things.
To avoid this built-in unsoundness, our implementation 
and case studies don't directly  
use Haskell's equality. 

\paragraph{Equivalence Relation Axioms}
Each of the four equalities has a different relationship to th equivalence relation axioms (reflexivity, symmetry, transitivity).
|AEq| comes with explicit proof methods that capture the axioms.
For |PEq|, we prove the equality axioms using classy induction (\S\ref{sec:gadt-metatheory}).
For SMT equality, we simply trust implementation of the 
underlying solver.
For Haskell's equality, there's no general way to enforce the equality axioms, though users can choose to prove them.

\paragraph{Computability}
Finally, the |Eq| and |AEq| classes 
define the computable equalities used in programs, |(==)| and |($\aEqSym$)| respectively.
The |PEq| equality only contains 
proof terms, while the SMT equality lives entirely
inside the refinements; neither can be meaningfully used in programs.

\section{Case Studies}
\label{sec:eg}

We demonstrate our propositional equality in seven case studies.
We start by moving from first-order equalities to
equalities between functions (|reverse|,
\S\ref{sec:eg:reverse}).
Next, we show how |$\PEq$|'s type indices reason about 
refined domains and dependent codomains of functions (|succ|, \S\ref{sec:eg:refdom}).
Proofs about higher-order functions demonstrate the contextual
equivalence axiom (|map|, \S\ref{sec:eg:map}).
Then, we see that |$\PEq$| plays well with multi-argument
functions (|foldl|, \S\ref{sec:eg:fold}).
Next, we present how a |$\PEq$| proof can speedup code (|spec|, \S\ref{sec:eg:spec}).
Finally, we present two bigger case studies that prove
the monoid laws for endofunctions (\S\ref{sec:endofunction}) and
the monad laws for reader monads (\S\ref{sec:reader}).
Complete code is available in the~\cite{implementation}.


\subsection{Reverse: from First- to Higher-Order Equality}
\label{sec:eg:reverse}

\begin{figure}
\flushleft
\noindent\textit{Two correct and one wrong implementations of reverse}\hfill
\vspace{0.1mm}

\begin{minipage}{.50\textwidth}
\begin{mcode}
slow, bad, fast :: [a] -> [a]
slow []     = [] 
slow (x:xs) = slow xs ++ [x]
bad xs = xs 
\end{mcode}
\end{minipage}%
\begin{minipage}{.50\textwidth}
\begin{mcode}
fast xs = fastGo [] xs 
fastGo :: [a] -> [a] -> [a]
fastGo acc []     = acc 
fastGo acc (x:xs) = fastGo (x:acc) xs 
\end{mcode}
\end{minipage}%

\vspace{0.2mm}
\noindent\textit{First-Order Theorems relating \texttt{fast} and \texttt{slow}}\hfill
\vspace{0.1mm}
\begin{mcode}
reverseEq :: xs:[a] -> { fast xs $\req$ slow xs }
lemma     :: xs:[a] -> ys:[a] -> {fastGo ys xs $\req$ slow xs ++ ys}
assoc     :: xs:[a] -> ys:[a] -> zs:[a] -> { (xs ++ ys) ++ zs $\req$ xs ++ (ys ++ zs) }
rightId   :: xs:[a] -> { xs ++ [] $\req$ xs }
\end{mcode}
\vspace{0.2mm}
\noindent\textit{Proofs of the First-Order Theorems}\hfill
\vspace{0.1mm}
\vspace*{-1em}
\begin{minipage}[t]{.63\textwidth}
\begin{mcode}
reverseEq x    = lemma x [] ? rightId (slow x)
lemma []    _  = ()
lemma (a:x) y  = lemma x (a:y) ? assoc (slow x) [a] y
x ? _pf        = x
\end{mcode}
\end{minipage}%
\hspace*{1.25em}%
\begin{minipage}[t]{.33\textwidth}
\begin{mcode}                
rightId []      = ()
rightId (_:x)   = rightId x
assoc []    _ _ = ()
assoc (_:x) y z = assoc x y z 
\end{mcode}
\end{minipage}%
\caption{Reasoning about list reversal.}
\label{fig:reverse}
\end{figure}

Consider three candidate definitions of the list-reverse function
(Figure~\ref{fig:reverse}, top): a `fast' one in accumulator-passing
style, a `slow' one in direct style,
and a `bad' one that is the identity. 

\paragraph{First-Order Proofs}
The |reverseEq| theorem neatly relates the two list reversals (Figure~\ref{fig:reverse}).
The final theorem |reverseEq| is a corollary of a |lemma| and
|rightId|, which shows that |[]| is a right identity for list append, |(++)|. The |lemma| is
the core induction, relating the accumulating |fastGo| and the
direct |slow|. The |lemma| itself uses the inductive lemma
|assoc| to show associativity of |(++)|.
All the equalities in the first order statements use the SMT equality, 
since they are automatically proved by Liquid Haskell's reflection 
and PLE tactic~\cite{VazouTCSNWJ18}. 

\paragraph{Higher-Order Proofs}
Plain SMT equality isn't enough to prove that
|fast| and |slow| are themselves equal.
We need functional extensionality: 
the |$\XEq$| constructor of the |$\PEq$| GADT.
\begin{mcode}
reverseHO :: PEq ([a] -> [a]) {fast} {slow}
reverseHO = $\XEq$ fast slow reversePf
\end{mcode}
The job of the |reversePf| lemma is to prove |fast xs|
propositionally equal to |slow xs| for all |xs|:
\begin{mcode}
reversePf :: xs:[a] -> PEq [a] {fast xs} {slow xs} 
\end{mcode}
There are several different ways to construct such a proof. 

\paragraph{Style 1: Lifting First-Order Proofs}
The first order equality proof |reverseEq|
lifts directly into propositional equality, 
using the |$\BEq$| constructor and the reflexivity property of $\AEq$. 
\begin{mcode}
reversePf1 :: AEq [a] => xs:[a] -> PEq [a] {fast xs} {slow xs}
reversePf1 xs = BEq (fast xs) (slow xs) (reverseEq xs ? reflP (fast xs))
\end{mcode}
Such proofs rely on SMT equality, which the |reflP| call turns into axiomatized
equality (|AEq|).

\paragraph{Style 2: Inductive Proofs}
Alternatively, inductive proofs can be directly performed 
in the propositional setting, eliminating the |AEq| constraint. 
To give a sense of what these proofs are like, 
we translate |lemma| into |lemmaP|:
\begin{mcode}
lemmaP :: (Reflexivity [a], Transitivity [a]) 
       => rest:[a] -> xs:[a] -> PEq [a] {fastGo rest xs} {slow xs ++ rest} 
lemmaP rest [] = refl rest
lemmaP rest (x:xs) = trans (fastGo rest (x:xs)) (slow xs ++ (x:rest)) 
                           (slow (x:xs) ++ rest)
                           (lemmaP (x:rest) xs) (assocP (slow xs) [x] rest)  
\end{mcode}
The proof goes by induction and uses the |Reflexivity| and
|Transitivity| properties of |$\PEq$| encoded as typeclasses (\S\ref{sec:gadt-metatheory})
along with |assocP| and |rightIdP|, the propositional versions of |assoc| and |rightId| (not shown).
These typeclass constraints  propagate to the |reverseHO| proof, 
via |reversePf2|. 
\begin{mcode}
reversePf2 :: (Reflexivity [a], Transitivity [a]) 
           => xs:[a] -> PEq [a] {fast xs} {slow xs}
reversePf2 xs = trans (fast xs) (slow xs ++ []) 
                  (slow xs) 
                  (lemmaP [] xs) (rightIdP (slow xs)) 
\end{mcode}

\paragraph{Style 3: Combinations}
One can combine the easy first order inductive proofs 
with the typeclass-encoded properties.
Here |refl| sets up the propositional context; 
|lemma| and |rightId| complete the proof.

\begin{mcode}
reversePf3 :: (Reflexivity [a]) => xs:[a] -> PEq [a] {fast xs} {slow xs} 
reversePf3 xs = refl (fast xs) ? lemma xs [] ? rightId (slow xs)
\end{mcode}

\paragraph{Bad Proofs}
We could not use any of these styles to generate a bad
(non-)proof: 
neither |$\PEq$ ([a] -> [a]) {fast} {bad}|
nor |$\PEq$ ([a] -> [a]) {slow} {bad}| are provable.

\subsection{Succ: Refined Domains and Dependent Codomains}
\label{sec:eg:refdom}

Our propositional equality |$\PEq$| 
naturally reasons about functions with refined domains and dependent codomains. 
For example, recall the functions |incrInt| and |incrPos| from~\S\ref{sec:inconsistency}:
\begin{mcode}
  incrInt, incrPos :: Integer -> Integer
  incrInt n = n + 1
  incrPos n = if 0 < n then n + 1 else 0
\end{mcode}
In~\S\ref{sec:inconsistency} we proved that the two functions are equal 
on the domain of positive numbers: 
\begin{mcode}
type Pos = {x:Integer | 0 < x }
posDom :: $\PEq$ (Pos -> Integer) {incrInt} {incrPos} 
posDom = $\XEq$ incrInt incrPos $\$$ \x -> $\BEq$ (incrInt x) (incrPos x) (reflP (incrInt x))
\end{mcode}
We can also reason about how each function's domain affects its codomain. 
For example, we can prove that these functions are equal \emph{and} they take |Pos| inputs 
to natural numbers.
\begin{mcode}
posRng :: $\PEq$ (Pos -> {v:Integer | 0 <= v}) {incrInt} {incrPos}
posRng = $\XEq$ incrInt incrPos $\$$ \x -> $\BEq$ (incrInt x) (incrPos x) (reflP (incrInt x)) 
\end{mcode}

Finally, we can prove properties of the function's codomain that depend 
on the inputs. Below we show that on positive arguments,
the result is always increased by one. 
\begin{mcode}
type SPos x = {v:Pos | v $\req$ x + 1}
depRng :: $\PEq$ (x:Pos -> SPos {x}) {incrInt} {incrPos}
depRng = $\XEq$ incrInt incrPos $\$$ \x -> $\BEq$ (incrInt x) (incrPos x) (reflP (incrInt x)) 
\end{mcode}

\paragraph{Equalities Rejected by Our System}
Liquid Haskell correctly rejects various wrong, (non-)proofs of equality between the functions 
|incrInt| and |incrPos|. We highlight three:
\mmg{We don't ever talk about \emph{negating} equality. Interesting.}
\begin{mcode}
badDom  :: $\PEq$ (  Integer -> Integer)               {incrInt} {incrPos}
badCod  :: $\PEq$ (  Pos     -> {v:Integer | v < 0})   {incrInt} {incrPos}
badDCod :: $\PEq$ (x:Pos     -> {v:Integer | v $\req$ x+2}) {incrInt} {incrPos}
\end{mcode}
|badDom| expresses that |incrInt| and |incrPos| are equal for any 
|Integer| input, which is wrong, \eg |incrInt (-2)| yields
|-1|, but |incrPos (-2)| yields 0.
Correctly constrained to positive domains, 
|badCod| specifies a negative codomain (wrong) while |badDCod| specifies 
that the result is increased by 2 (also wrong). 
Our system rejects all three with a refinement type error.  

\subsection{Map: Putting Equality in Context}
\label{sec:eg:map}
Our propositional equality can be used 
in higher order settings: we prove that if |f| and |g| are propositionally equal, then |map f| and |map g| are also equal.
Our proofs use the congruence closure equality constructor/axiom |$\CEq$|. 

\paragraph{Equivalence on the Last Argument}
Direct application of |$\CEq$| ports 
a proof of equality to the last argument of the context (a function). 
For example, |mapEqP| below states that if two functions |f|
and |g| are equal, then so are the partially applied functions 
|map f| and |map g|. 
\begin{mcode}
mapEqP :: f:(a -> b) -> g:(a -> b) -> $\PEq$ (a -> b) {f} {g} 
       -> $\PEq$ ([a] -> [b]) {map f} {map g}
mapEqP f g pf = $\CEq$ f g pf map 
\end{mcode}

\paragraph{Equivalence on an Arbitrary Argument}
To show that |map f xs| and |map g xs| are equal for all |xs|, 
we use |$\CEq$| with |flipMap|, \ie a context that puts |f| and |g| in a `flipped' context. 
\begin{mcode}
mapEq :: f:(a -> b) -> g:(a -> b) -> $\PEq$ (a -> b) {f} {g} 
      -> xs:[a] -> $\PEq$ [b] {map f xs} {map g xs}
mapEq f g pf xs = $\CEq$ f g pf (flipMap xs) ? fMapEq f xs ? fMapEq g xs

fMapEq :: f:_ -> xs:[a] -> {map f xs $\req$ flipMap xs f}
fMapEq f xs  = ()
flipMap xs f = map f xs  
\end{mcode}
The |mapEq| proof relies on |$\CEq$| with the flipped context 
and needs to know that |map f xs $\req$ flipMap xs f|.
Liquid Haskell won't infer this fact on its own
in the higher order setting of this proof; we explicitly
provide this evidence with the |fMapEq| calls.  

\paragraph{Proof Reuse in Context}
Finally, we use the |posDom| proof (\S\ref{sec:eg:refdom})
to show how existing proofs can be reused with |map|. 
\begin{mcode}
client :: xs:[Pos] -> $\PEq$ [Integer] {map incrInt xs} {map incrPos xs} 
client = mapEq incrInt incrPos posDom 

clientP :: $\PEq$ ([Pos] -> [Integer]) {map incrInt} {map incrPos}
clientP = mapEqP incrInt incrPos posDom 
\end{mcode}
|client| proves that 
|map incrInt xs| is equivalent to |map incrPos xs| for each list |xs| of positive numbers, while 
|clientP| proves that the partially applied functions |map incrInt| and |map incrPos| are 
equivalent on the domain of lists of positive numbers. 

\subsection{Fold: Equality of Multi-Argument Functions}
\label{sec:eg:fold}
As an example of equality proofs on multi-argument functions, we show that
the directly tail-recursive |foldl| is equal to |foldl'|, a 
|foldr| encoding of a left-fold via CPS. 
The first-order equivalence theorem is expressed as follows: 
\begin{mcode}
thm :: f:(b -> a -> b) -> b:b -> xs:[a] -> { foldl f b xs $\req$ foldl' f b xs } 
\end{mcode}
%
We lifted the first-order property into a multi-argument function equality
by using |$\XEq$| for all but the last arguments and |$\BEq$| for the last, as below: 
\begin{mcode}
foldEq :: AEq b => $\PEq$ ((b -> a -> b) -> b -> [a] -> b) {foldl} {foldl'}
foldEq = $\XEq$ foldl foldl' dollar \f ->
           $\XEq$ (foldl f) (foldl' f) dollar \b -> 
             $\XEq$ (foldl f b) (foldl' f b) dollar \xs ->
               $\BEq$ (foldl f b xs) (foldl' f b xs) 
                 (thm f b xs ? reflP (foldl f b xs))  
\end{mcode}
One can avoid the first-order proof and the |AEq| constraint, by using
the second, typeclass-oriented style of~\S\ref{sec:eg:reverse},
(see~\citet{implementation} for details).

\subsection{Spec: Function Equality for Program Efficiency}
\label{sec:eg:spec}

Function equality 
can be used to prove optimizations sound.
For example, consider a |critical| function that, for safety, can only run on inputs 
that satisfy a specification |spec|,
and |fastSpec|, a fast implementation to check |spec|.
\begin{mcode}
  spec, fastSpec :: a -> Bool 
  critical :: x:{ a | spec x } -> a 
\end{mcode}
A |client| function can soundly call |critical| for any input |x|
by performing the runtime |fastSpec x| check, 
given a |PEq| proof that the functions |fastSpec| and |spec| are equal.
\begin{mcode}
  client :: PEq (a -> Bool) {fastSpec} {spec} -> a -> Maybe a
  client pf x = 
    if fastSpec x ? toSMT (fastSpec x) (spec x) (CEq fastSpec spec pf (\x f -> f x))
      then Just (critical x)
      else Nothing
\end{mcode}
The |toSMT| call generates the SMT equality that |fastSpec x $\req$ spec x|, 
which, combined with the branch condition check |fastSpec x|, 
lets the path-sensitive refinement type checker decide that 
the call to |critical x| is safe in the |then| branch.

%

Our propositional equality 
(1) co-exists with practical features of refinement types, \eg path sensitivity, 
and (2) can help optimize executable code. 

\subsection{Monoid Laws for Endofunctions}
\label{sec:endofunction}

Endofunctions form a law-abiding monoid.
A function $f$ is an \emph{endofunction} when
its domain and codomain types are the same\iffull.
\noindent
\fi
A \emph{monoid} is an algebraic structure comprising an identity
element (|mempty|) and an associative operation (|mappend|).
For the monoid of endofunctions, |mempty| is
the identity function and |mappend| is
function composition (Figure~\ref{fig:endo}; top).

\begin{figure*}
\noindent\textit{Monoid Instance for Endofunctions}\hfill
\vspace{0.01mm}

\begin{minipage}{.5\linewidth}
\begin{mcode}
type Endo a = a -> a
  
mempty :: Endo a
mempty a = a
\end{mcode}
\end{minipage}%
\begin{minipage}{.5\linewidth}
\begin{mcode}
_ =~= y = y 

mappend :: Endo a -> Endo a -> Endo a
mappend f g a = f (g a) -- a/k/a (<>)   
\end{mcode}
\end{minipage}

\noindent\textit{Endofunction Monoid Laws}\hfill
\vspace{0.01mm}
\begin{mcode}
mLeftIdentity  :: (Reflexivity a, Transitivity a) 
               => x:Endo a -> $\PEq$ (Endo a) {mappend mempty x} {x} 
mRightIdentity :: (Reflexivity a, Transitivity a) 
               => x:Endo a -> $\PEq$ (Endo a) {x} {mappend x mempty} 
mAssociativity :: (Reflexivity a, Transitivity a) 
               => x:(Endo a) -> y:(Endo a) -> z:(Endo a)
               -> $\PEq$ (Endo a) {mappend (mappend x y) z} {mappend x (mappend y z)}
\end{mcode}

\vspace{0.1mm}
\noindent\textit{Proofs By Reflexivity and Transitivity}\hfill
\vspace{0.01mm}
\begin{mcode}
mLeftIdentity x = $\XEq$ (mappend mempty x) x dollar \a ->
      refl (mappend mempty x a) ? (mappend mempty x a =~= mempty (x a) =~= x a *** QED)  

mRightIdentity x = $\XEq$ x (mappend x mempty) dollar \a ->
      refl (x a) ? (x a =~= x (mempty a) =~= mappend x mempty a *** QED)

mAssociativity x y z = 
    $\XEq$ (mappend (mappend x y) z) (mappend x (mappend y z)) dollar \a ->
      refl (mappend (mappend x y) z a) ?
           (mappend (mappend x y) z a =~= (mappend x y) (z a) =~= x (y (z a)) 
                                      =~= x (mappend y z a) 
                                      =~= mappend x (mappend y z) a *** QED)
\end{mcode}

\caption{Case study: Endofunction Monoid Proofs.}
\label{fig:endo}
\end{figure*}

To be a monoid,
|mempty| must really be an identity with respect to |mappend|
(|mLeftIdentity| and |mRightIdentity|\!) and  
|mappend| must really be associative (|mAssociativity|\!) (Figure~\ref{fig:endo}; middle).

Proving the monoid laws for endofunctions demands functional extensionality 
(Figure~\ref{fig:endo}; bottom).
For example, consider the proof that |mempty| is a left identity for
|mappend|, \ie |mappend mempty x $\req$ x|. 
To prove this equation between \emph{functions}, 
we can't use  SMT equality.
With functional extensionality, each proof reduces to three parts:
|$\XEq$| to take an input of type |a|;
|refl|\iffull on the left-hand side of the equation\fi, to generate an equality proof;
and |(=~=)| to give unfolding hints to the SMT solver. The |(=~=)| operator is defined 
as |_ =~= y = y|, and it is unrefined, \ie it is not checking equality of its arguments.
%
%

The |Reflexivity| constraints on the theorems make our proofs general
in the underlying type |a|: endofunctions on the type |a| form a monoid whether
|a| admits SMT equality or if it's a complex
higher-order type (whose ultimate result admits equality).
Haskell's typeclass resolution ensures that an appropriate |refl|
method will be constructed whatever type |a| happens to be.

\subsection{Monad Laws for Reader Monads}
\label{sec:reader}
A \emph{reader} is a function with a fixed domain |r|, \ie the
partially applied type |Reader r| (Figure~\ref{fig:readers}, top left).
Readers form a monad and their composition is a useful way of
defining and composing functions that take some fixed information, like command-line
arguments or configuration files\mmg{is there a cite for this?}.
Our propositional equality can prove 
the monad laws for readers. 

\begin{figure*}
\noindent\textit{Monad Instance for Readers}\hfill
\vspace{0.01mm}
\begin{minipage}{.45\textwidth}
\begin{mcode}
  type Reader r a = r -> a

  kleisli :: (a -> Reader r b) 
          -> (b -> Reader r c)  
          -> a -> Reader r c
  kleisli f g x = bind (f x) g
\end{mcode}
\end{minipage}%
\begin{minipage}{.55\textwidth}
\begin{mcode}
pure :: a -> Reader r a
pure a _r = a

bind :: Reader r a -> (a -> Reader r b) 
     -> Reader r b
bind fra farb = \r -> farb (fra r) r
\end{mcode}   
\end{minipage}%

\noindent\textit{Reader Monad Laws}\hfill
\vspace{0.01mm}

\begin{mcode}
monadLeftIdentity  :: Reflexivity b => a:a 
                   -> f:(a -> Reader r b) -> $\PEq$ (Reader r b) {bind (pure a) f} {f a}
monadRightIdentity :: Reflexivity a 
                   => m:(Reader r a) -> $\PEq$ (Reader r a) {bind m pure} {m}
monadAssociativity :: (Reflexivity c, Transitivity c) 
                   => m:(Reader r a) -> f:(a -> Reader r b) -> g:(b -> Reader r c)
                   -> $\PEq$ (Reader r c) {bind (bind m f) g} {bind m (kleisli f g)}
\end{mcode}    

\vspace{0.1mm}
\noindent\textit{Identity Proofs By Reflexivity}\hfill
\vspace{0.01mm}
\begin{minipage}{.50\textwidth}
\begin{mcode}
  monadLeftIdentity a f =
    $\XEq$ (bind (pure a) f) (f a) dollar \r ->
     refl (bind (pure a) f r) ?
      (bind (pure a) f r =~= f (pure a r) r 
               =~= f a r *** QED)  
\end{mcode}
\end{minipage}%
\begin{minipage}{.50\textwidth}
\begin{mcode}
monadRightIdentity m =
  $\XEq$ (bind m pure) m dollar \r -> 
    refl (bind m pure r) ?
     (bind m pure r =~= pure (m r) r 
                    =~= m r *** QED)
\end{mcode}
\end{minipage}
\vspace{1mm}
\noindent\textit{Associativity Proof By Transitivity and Reflexivity}\hfill
\vspace{0.01mm}
\begin{mcode}
  monadAssociativity m f g = $\XEq$ (bind (bind m f) g) (bind m (kleisli f g)) dollar \r ->
  let { el  = bind (bind m f) g r ; eml = g (bind m f r) r ; em  = (bind (f (m r)) g) r
      ; emr = kleisli f g (m r) r ; er  = bind m (kleisli f g) r   }
  in trans el em er (trans el eml em (refl el) (refl eml)) 
                    (trans em emr er (refl em) (refl emr))
\end{mcode}
  \caption{Case study: Reader Monad Proofs.}
  \label{fig:readers}
\end{figure*}
The monad instance for the reader type is defined using function composition (Figure~\ref{fig:readers}, top).
We also define
Kleisli composition of monads as a convenience for specifying the monad.
We prove that readers are in fact monads,
\ie their operations satisfy the monad laws (Figure~\ref{fig:readers}, bottom). We also prove that they satisfy the functor and applicative
laws in~\citet{implementation}.
The reader monad laws are expressed 
as refinement type specifications using |$\PEq$|. 
We prove the left and right identities 
following the pattern of~\S\ref{sec:endofunction}, \ie 
|$\XEq$|, followed by reflexivity with |(=~=)| for function unfolding (Figure~\ref{fig:readers}, middle). 
We use transitivity to 
conduct the more complicated proof of associativity (Figure~\ref{fig:readers}, bottom). 

\paragraph{Proof by Associativity and Error Locality}
As noted earlier, the use of |(=~=)| in proofs by reflexivity is not 
checking intermediate equational steps. 
So, the proof either succeeds or fails without explanation. 
To address this problem, during proof construction, we employed transitivity. 
For instance, in the |monadAssociativity| proof, our goal is to construct 
the proof |$\PEq$ _ {el} {er}|. 
To do so, we pick an intermediate term |em|; we might attempt an equivalence proof 
as follows:
\begin{mcode}
trans el em er
  (refl el)        -- proof of el  = em; local error 
  (trans em emr er -- proof of em  = er 
    (refl em)      -- proof of em  = emr 
    (refl emr))    -- proof of emr = er 
\end{mcode}
The |refl el| proof will produce a type error; replacing that proof
with an appropriate |trans| to connect |el| and |em| via |eml| completes the |monadAssociativity| proof
(Figure~\ref{fig:readers}, bottom).
Writing proofs in this |trans|/|refl| style works well: start
with |refl| and where the SMT solver can't figure things out, a local
refinement type error tells you to expand with |trans| (or look
for a counterexample).

Our reader proofs use the |Reflexivity| and
|Transitivity| typeclasses to ensure that readers are monads whatever
the return type |a| may be (with the type of `read' values fixed to |r|).
Having generic monad laws is critical: readers are typically used to
compose functions that take configuration information, but such
functions usually have other arguments, too!
For example, an interpreter might run |readFile >>= parse >>= eval|,
where |readFile :: Config -> String| and 
|parse :: String -> Config -> Expr| and
|eval :: Expr -> Config -> Value|.
With our generic proof of associativity, we can rewrite the above to
|readFile >>= (kleisli parse eval)| even though |parse| and |eval| are
higher-order terms without |Eq| instances. Doing so could, in theory,
trigger inlining/fusion rules that would combine the parser and the
interpreter.

\section{Type Checking \texttt{XEq}: Did We Get It Right?}
\label{sec:rules}

We've seen that |XEq| is effective at proving equalities between
functions (\S\ref{sec:eg}) and we've argued that we avoid the inconsistency with
|funext|.
Things \emph{seem} to work in Liquid Haskell. But: Why do things go so
wrong with |funext|? Does |XEq| really avoid |funext|'s issues?
We give a schematic example showing why Liquid Haskell works with
|XEq| consistently but works with |funext| inconsistently.
(We give a detailed, formal model of our propositional equality in \S\ref{sec:eqrt}.)


Suppose we have
two functions \hName and \kName,
defined on
domains $\hdom$ and $\kdom$ 
and codomains $\hrng$ and $\krng$, respectively.
Let's also assume we have some $\lemmaName$
that proves, for all $x$ in some domain 
$\ldom$, we have an equality \fEq{e_l}{e_r}, 
where $e_l$ and $e_r$ are arbitrary expressions of 
type |{v:$\beta$ $\mid$ $r_p$}|.
\begin{mcode}
  $\hName$ :: x:{$\alpha$ | $\hdom$} -> {v:$\beta$ | $\hrng$}
  $\kName$ :: x:{$\alpha$ | $\kdom$} -> {v:$\beta$ | $\krng$}
  $\lemmaName$ :: x:{$\alpha$ | $\ldom$} -> $\PEq$ {v:$\beta$ | $r_p$} {$e_l$} {$e_r$}
\end{mcode}
We can pass these along
to our $\XEq$ constructor (of~\S\ref{sec:eqrt-gadt})
to form a proof that $\hName$ equals $\kName$ on some domain $\edom$:
\begin{mcode}
  $\XEq$ $\hName$ $\kName$ $\lemmaName$ :: $\PEq$ ({v:$\alpha$ | $\edom$} -> {v:$\beta$ | $r_e$}) {$\hName$} {$\kName$}
\end{mcode}
When type checking this use of |XEq|, we need to
check that the lemma equates the right expressions (\ie forall $x$. \fEq{e_l}{e_r} implies \fEq{\hName\ x}{\kName\ x}).
Critically, type checking must also ensure that the final equality domain 
($\edom$) is stronger than the domains for the functions ($\hdom$, $\kdom$) and for the lemma ($\ldom$). 

\begin{figure}
\small

\noindent \textit{Typing Environmennt}\hfill
\vspace*{0.1em}
$$
\begin{array}{rcll}
    \env & \doteq &  \{ & 
                  \envBind{\XEq}{\forall \alpha \beta. f:(\typExtG{\as}{\bs}) \rightarrow g:(\typExtG{\as}{\bs}) 
                  \rightarrow(x:\as \rightarrow \PEq\ \bs \{f \ x \} \{g \ x\})
                  \rightarrow \PEq\ (\as \rightarrow \bs) \{f \} \{g \}
                  }\\
      &&          , & \envBind{\hName}{\hTy}, \envBind{\kName}{\kTy}
                  , \envBind{\lemmaName}{x:\sref{v}{\alpha}{\ldom} \rightarrow \PEq\ \sref{v}{\beta}{r_p} \{e_l \} \{e_r\}}  \quad \}\\
\end{array}
$$
\vspace*{0.5em}
\noindent \textit{Type Checking}\hfill
\vspace*{0.1em}

\begin{adjustbox}{max size={\textwidth}{\textheight}}
$
\begin{array}{l}
    1. \hastype{\env}{\XEq}{\forall \alpha \beta. 
                         \sfuna{f}{\as}{\bs} \rightarrow 
                         \sfuna{g}{\as}{\bs} 
                         \rightarrow(\sfun{x}{\as}{\PEq\ \bs\ \{f \ x \} \{g \ x\})}
                         \rightarrow \PEq\ (\as \rightarrow \bs) \{f \} \{g \}
                         } 
     \\ \hline
    2. \hastype{\env}{\XEq\ \tyApp{\tya} }{\forall \beta. 
                         \sfuna{f}{\tya}{\bs}
                         \rightarrow \sfuna{g}{\tya}{\bs} 
                         \rightarrow(\sfun{x}{\tya}{\PEq\ \bs\ \{f \ x \} \{g \ x\})}
                         \rightarrow \PEq\ (\tya \rightarrow \bs) \{f\} \{g\}
                         } 
     \\ \hline
    3. \hastype{\env}{\XEq\ \tyApp{\tya}\ \tyApp{\tyb} }{
                         \sfuna{f}{\tya}{\tyb}
                         \rightarrow \sfuna{g}{\tya}{\tyb} 
                         \rightarrow(\sfun{x}{\tya}{\PEq\ \tyb\ \{f \ x \} \{g \ x\})}
                         \rightarrow \PEq\ (\tya \rightarrow \tyb) \{f\} \{g\}
                         } 
     \\ \hline
    4. \hastype{\env}{\XEq\ \tyApp{\tya}\ \tyApp{\tyb} \ \hName}{
                          g:(\typExtG{\tya}{\tyb}) 
                         \rightarrow(x\colon\tya \rightarrow \PEq\ \tyb \{\hName\ x \} \{g\ x\})
                         \rightarrow \PEq\ (\tya \rightarrow \tyb) \{\hName\} \{g\}
                         } 
    \quad {\fbox{\subh}} \\ \hline
    5. \hastype{\env}{\XEq\ \tyApp{\tya}\ \tyApp{\tyb} \ \hName \ \kName}{
                         (x:\tya \rightarrow \PEq\ \tyb \{\hName\ x \} \{\kName\ x\})
                         \rightarrow \PEq\ (\tya \rightarrow \tyb) \{\hName\} \{\kName\}
                         } 
    \quad {\fbox{\subk}} \\ \hline
    6. \hastype{\env}{\XEq\ \tyApp{\tya}\ \tyApp{\tyb} \ \hName \ \kName\ \lemmaName{}}{
                         \PEq\ (\tya \rightarrow \tyb) \{\hName\} \{\kName\}
                         } 
    \quad {\fbox{\sublemma}} \\ \hline
    7. \hastype{\env}{\XEq\ \tyApp{\tya}\ \tyApp{\tyb} \ \hName \ \kName\ \lemmaName{}}{
                         \PEq\ (\tya \rightarrow \tyb) \{\hName\} \{\kName\}
                         } 
    \quad {\fbox{\subsub}} \\ \hline
    8. \hastype{\env}{\XEq\ \tyApp{\tya}\ \tyApp{\tyb} \ \hName \ \kName\ \lemmaName}{
                         \PEq\ (\tyde \rightarrow \tyre) \{\hName\} \{\kName\}
                         } 
                           \\ 
\end{array}
$
\end{adjustbox}

\vspace*{0.5em}
\noindent\textit{Subtyping Derivation Leaves}\hfill
\vspace*{0.1em}

\begin{adjustbox}{max size={\textwidth}{\textheight}}
$$
\inference{
    \inference{
        \textcolor{purple}{\romannumeral 1.\ \apred \Rightarrow \hdom}
    }{
    \issubtype{\env}{\tya}{\sref{v}{\alpha}{\hdom}}
     } && 
     \inference{
        \apred \Rightarrow \hrng \Rightarrow \bpred
     }{
     \issubtype{\env,\envBind{x}{\tya}}{\sref{v}{\beta}{\hrng}}{\tyb}
     }
}{
  \issubtype{\env}{\hTy}{\typExtF{\tya}{\tyb}}
}[\subh]
\quad
\inference{
    \inference{
        \textcolor{purple}{\romannumeral 2.\ \apred \Rightarrow \kdom}
    }{
    \issubtype{\env}{\tya}{\sref{v}{\alpha}{\kdom}}
     } && 
     \inference{
        \apred \Rightarrow \krng \Rightarrow \bpred
     }{
     \issubtype{\env,\envBind{x}{\tya}}{\sref{v}{\beta}{\krng}}{\tyb}
     }
}{
  \issubtype{\env}{\kTy}{\typExtF{\tya}{\tyb}}
}[\subk]
$$
\end{adjustbox}

\vspace*{0.5em}

\begin{adjustbox}{max size={\textwidth}{\textheight}}
$$
\inference{
    \inference{
        \textcolor{purple}{\romannumeral 3.\ \apred \Rightarrow \ldom}    
    }{
    \issubtype{\env}{\tya}{\sref{v}{\alpha}{\ldom}}
     } && 
     \inference{
         \inference{
            \apred \Rightarrow r_p \Rightarrow \bpred
         }{
            \textcolor{red}{\issubtype{\env,\envBind{x}{\tya}}
        {\sref{v}{\beta}{r_p}}
        {\tyb}}}
        && 
        \inference{
            \apred \Rightarrow \bpred \Rightarrow r_p 
        }{
            \textcolor{blue}{\issubtype{\env,\envBind{x}{\tya}}
        {\tyb}
        {\sref{v}{\beta}{r_p}}
        }}
        && 
        \textcolor{purple}{\romannumeral 4.\ \apred \Rightarrow \fEq{e_l}{e_r} \Rightarrow \fEq{\hName\ x}{\kName\ x}}
     }{
     \issubtype{\env,\envBind{x}{\tya}}
        {\PEq\ \sref{v}{\beta}{r_p} \{e_l \} \{e_r\}}
        {\PEq\ \tyb \{\hName\ x \} \{\kName\ x\}}
     }
}{
  \issubtype{\env}
     {x:\sref{v}{\alpha}{\ldom} \rightarrow \PEq\ \sref{v}{\beta}{r_p} \{e_l \} \{e_r\}}
     {x:\tya \rightarrow \PEq\ \tyb \{\hName\ x \} \{\kName\ x\}}
}[\sublemma]
$$
\end{adjustbox}
\begin{adjustbox}{max size={\textwidth}{\textheight}}
$$
\inference{
    \inference{
        \inference{
            \textcolor{purple}{\romannumeral 6.\ \edom \Rightarrow \apred}
        }{
            \issubtype{\env}{\tyde}{\tya}
        } &&
        \inference{
            \apred \Rightarrow \bpred \Rightarrow r_e 
        }{ 
            \issubtype{\env,\envBind{x}{\tyde}}{\tyb}{\tyre}
        }     
    }{
        \textcolor{red}{\issubtype{\env}{\tya \rightarrow \tyb}{\tyde \rightarrow \tyre}}
     } && 
    \inference{
        \inference{
            \textcolor{purple}{\romannumeral 5.\ \apred \Rightarrow \edom}
        }{
            \issubtype{\env}{\tya}{\tyde}
        } &&
        \inference{
            \apred \Rightarrow r_e \Rightarrow \bpred
        }{ 
            \issubtype{\env,\envBind{x}{\tya}}{\tyre}{\tyb}
        }
    }{
        \textcolor{blue}{\issubtype{\env}{\tyde \rightarrow \tyre}{\tya \rightarrow \tyb}}
    } && 
      \fEq{\hName}{\kName} \Rightarrow \fEq{\hName}{\kName}
}{
  \issubtype{\env}{\PEq\ (\tya \rightarrow \tyb) \{\hName\} \{\kName\}}
           {\PEq\ (\tyde \rightarrow \tyre) \{\hName\} \{\kName\}}
}[\subsub]
$$
\end{adjustbox}

\vspace*{0.5em}

\caption{Type checking $\XEq\ \hName \ \kName\ \lemmaName$.
For space, we write \sref{v}{t}{d} to mean the refined type \fref{v}{t}{d}. 
}
\label{fig:extensionality-checking}
\end{figure}

Liquid Haskell goes through a complex series of steps to enforce 
both required checks (Figure~\ref{fig:extensionality-checking}).
We haven't modified Liquid Haskell's typing rules or implementation \emph{at all};
we merely defined $\PEq$ in such a way 
that the existing type checking rules in Liquid Haskell implement the right checks to soundly show extensional equality between functions.

It's easiest to understand how type checking works from top to bottom (``Type Checking'', Figure~\ref{fig:extensionality-checking}).
First, we look up $\XEq$'s type in the environment (1). 
Since the $\XEq$ is polymorphic, 
we instantiate the type arguments with the types, \fref{v}{\alpha}{\apred} (2) and \fref{v}{\beta}{\bpred} (3). 
(We write \sref{v}{\alpha}{\apred} as a short for \fref{v}{\alpha}{\apred}, 
since we focus on the refinements assuming the Haskell types match.)
Here \apred and \bpred are refinement type variables; type checking will generate constraints on them that liquid type inference will try to resolve~\cite{LT2008}.
Next we apply each of the arguments:
$\hName$ (4), $\kName$ (5), and $\lemmaName$ (6).
Each application applies standard
dependent function application, with consideration for subtyping.
That is, each application
(a) substitutes the applied argument in the codomain type 
and (b) checks
that the type of the argument is a subtype of the function's domain type.
Application leads to the subtyping constraints \subh, \subk, and \sublemma set off in boxes, resolved below.
Now Liquid Haskell has \textit{inferred} 
a type for the checked expression (7). To conclude the check, 
it introduces the final subtype constraint \subsub:
the inferred type should be a subtype of the
type the user specified (8). 

The four instances of subtyping during type checking
generate 13 logical implications to resolve
for the original expression to type check (``Subtyping Derivation Leaves'', Figure~\ref{fig:extensionality-checking}).
The six \textcolor{purple}{purple implications with Roman numerals} place requirements on the domain; we'll ignore the others, which
impose less interesting constraints on the functions' codomains.
The \subh and \subk derivations
require (via contravariance) that the refinement variable \apred implies 
the refinements on the functions' domains, \hdom and \kdom. 
Similarly, the derivation \sublemma requires that \apred implies 
the proof domain \ldom.
Since \PEq is defined as refined type alias (\S\ref{sec:eqrt-gadt}),
\sublemma also checks that the refinements given imply the top level refinements of
\PEq,
\ie that the result of the lemma is sufficient to show
$\XEq$'s precondition.
The \subsub derivation checks subtyping of two \PEq types, by treating the type arguments 
\emph{invariantly}. (We mark \textcolor{red}{covariant implications in red} and \textcolor{blue}{contravariant implications in blue}.)
Liquid Haskell treats checks invariantly because 
\PEq's definition uses its type parameter 
in both positive and negative positions.
\subsub will ultimately require that the refinement variable \apred is \emph{equivalent} to 
the equality domain \edom. 

To sum up, type checking imposes the following 
six implications as constraints: 
\[
\begin{array}{r@{\quad}lr@{\quad}lr@{\quad}l}
  \romannumeral 1.& \apred \Rightarrow \hdom &
  \romannumeral 2.& \apred \Rightarrow \kdom &
  \romannumeral 3.& \apred \Rightarrow \ldom \\
  \romannumeral 4.& \apred \Rightarrow \fEq{e_l}{e_r} \Rightarrow \fEq{\hName\ x}{\kName\ x} &
  \romannumeral 5.& \apred \Rightarrow \edom &
  \romannumeral 6.& \edom \Rightarrow \apred \\
\end{array}
\]

Implications \emph{\romannumeral 5}\ and \emph{\romannumeral 6}\
require the refinement variable  \apred to be equivalent 
to the equality domain \edom.
Given that equality,
implications \emph{\romannumeral 1}--\emph{\romannumeral 3}\
state that the equality domain \edom should imply 
the domains of the functions (\emph{\romannumeral 1}\ and \emph{\romannumeral 2}) and lemma (\emph{\romannumeral 3}). Implication
\emph{\romannumeral 4}\ requires that the lemma's domain
implies equality of the two functions 
for each argument |x| that satisfies the domain \edom. 
All together, these constraints exactly capture the requirements 
of functional extensionality.

\paragraph{Naive Functional Extensionality with \texttt{funext}}

When, in~\S\ref{sec:inconsistency}, we use the non-type-indexed |funext| in Liquid Haskell, the typing derivation
looks almost exactly the same, 
but one critical thing changes: 
the type-indexed |$\PEq$ t {$e_l$} {$e_r$}|
is replaced by a refined unit |{v:() $\mid$ $e_l$ = $e_r$}|.
This only affects the \sublemma and \subsub derivations,
which lose the red and blue parts and become:
$$
\inference{
    \inference{
        \textcolor{purple}{{\romannumeral 3}'.\ \apred \Rightarrow \ldom}
    }{
    \issubtype{\env}{\sref{v}{\alpha}{\apred}}{\sref{v}{\alpha}{\ldom}}
     } && \hspace*{-1em}
     \inference{
        \textcolor{purple}{{\romannumeral 4}'.\ \apred \Rightarrow e_l = e_r \Rightarrow h\ x = k\ x}
     }{
     \issubtype{\env,\envBind{x}{\tya}}{\fref{v}{\texttt{()}}{e_l = e_r}}{
      \fref{v}{\texttt{()}}{h\ x = k\ x}
       }
     }
}{
  \issubtype{\env}{
    \sfun{x}{\sref{v}{\alpha}{\ldom}}{\fref{v}{\texttt{()}}{e_l = e_r}}
    }{
    \sfun{x}{\sref{v}{\alpha}{\apred}}{\fref{v}{\texttt{()}}{h\ x = k\ x}}
    }
}[\sublemma-\textsc{Naive}]
$$
$$
\inference{
      h\ x = k\ x \Rightarrow h\ x = k\ x
}{
  \issubtype{\env}{
    \fref{v}{\texttt{()}}{h\ x = k\ x}
    }{
    \fref{v}{\texttt{()}}{h\ x = k\ x}
    }
}[\subsub-\textsc{Naive}]
$$
\sublemma-\textsc{Naive}
generates the implications ${\romannumeral 3}'$ and 
${\romannumeral 4}'$ that 
are essentially the same as before.
But, \subsub-\textsc{Naive}
won't generate any meaningful checks,
because equality is just a unit type. 
We lost implications ${\romannumeral 5}$\ and ${\romannumeral 6}$! 
We are now left with an implication system in which the 
refinement variable \apred only appears in the assumptions. Since
Liquid Haskell always tries to infer the most specific refinement possible,
it will find a very specific refinement for \apred: |false|!
Having inferred |false| for \apred, the entire use of |funext|
trivially holds and can be used on other, nontrivial domains---with
inconsistent results. 

\section{A Refinement Calculus with Built-in Type-Indexed Equality}
\label{sec:eqrt}

Because |funext| is inconsistent in Liquid Haskell (\S\ref{sec:inconsistency}),
we developed |$\libname$| to reason consistently about extensional
equality, using the GADT |PBEq| and the uninterpreted equality |PEq|
(\S\ref{sec:eqrt-gadt}).
We're able to prove some interesting equalities (\S\ref{sec:eg}) and Liquid Haskell's type checking seems to be doing the right thing (\S\ref{sec:rules}).
But how do we know that our definitions suffice? Formalizing
\emph{all} of Liquid Haskell is a challenge, but we can build a model
to check the features we use.
We formalize a core calculus \corelaneq with $R$efinement
types, semantic subtyping, and type-indexed propositional $E$quality.

\corelaneq contains just enough to check the core interactions between
refinement types and a type-indexed propositional equality resembling
our |PBEq| definition (\S\ref{subsec:theory:core}).
%
We omit plenty of important features from Liquid Haskell (\eg algebraic data types): 
our purpose here is not to develop a
complete formal model, but to check that our implementation holds together.

Using \corelaneq's static semantics (\S\ref{subsec:static-semantics}),
we prove several metatheorems (\S\ref{subsec:metaproperties}).
Most importantly, we define a logical relation that characterizes
\corelaneq equivalence and reflects \corelaneq's propositional
equality.
Propositional equivalence in \corelaneq implies equivalence in the
logical relation (Theorem~\ref{thm:eq-relation}); both are reflexive, symmetric, and transitive
(Theorems~\ref{thm:theory:equality-logical-relation}
and~\ref{thm:theory:properties}).

\subsection{Syntax and Semantics of \corelaneq}
\label{subsec:theory:core}

\begin{figure}
\small
$$
\begin{array}{rrcl}
  \textit{Constants} & 
  \con & ::= &  \etrue \mid \efalse \mid \eunit \mid (\bbbeqsym{\tbase}) \mid \bbbeq{(}{)}{(\con,\tbase)} \\[1mm]
  \textit{Expressions} & 
  \expr & ::= & \con \mid x \mid \expr\ \expr \mid \elam{x}{\typ}{\expr} \eqcolor{\mid} \eqcolor{\ebeq{\tbase}\ \expr\ \expr\ \expr} 
                                 \ \eqcolor{\mid} \ \eqcolor{\exeq{x}{\typ}{\typ}\ \expr\ \expr\ \expr}
                                 \\[1mm]
  \textit{Values} & 
  \val & ::= & \con \mid \elam{x}{\typ}{\expr} \eqcolor{\mid} \eqcolor{\ebeq{\tbase}\ \expr\ \expr\ \val}
                           \ \eqcolor{\mid} \ \eqcolor{\exeq{x}{\typ}{\typ}\ \expr\ \expr\ \val}
                           \\[1mm] 
    \textit{Refinements} & 
    \refa & ::= & \expr \\[1mm]
    \textit{Basic Types} & 
    \tbase & ::= &  \tbool \mid \tunit 
                         \\[1mm]
    \textit{Types} & 
    \typ & ::= & \tref{\rbind}{\tbase}{\refa} 
                            \mid \tfun{x}{\typ}{\typ}  
                            \mid \eqrt{\typ}{\expr}{\expr} \\[1mm]
    \textit{Typing Environment} & 
    \env & ::= & \emptyset \mid \env, \envBind{x}{\typ}  \\[1mm]
    \textit{Closing Substitutions} & 
    \model & ::= & \emptyset \mid \model, \modelBind{x}{\val}  \\[1mm]
    \textit{Equivalence Environments} & 
    \rmodel & ::= & \emptyset \mid \rmodel, \rmodelBind{\val}{\val}{x}  \\[1mm]
    \textit{Evaluation Contexts} & 
    \ectx & ::= & \bullet \mid \ectx \ \expr \mid \val\ \ectx \eqcolor{\mid} \eqcolor{\ebeq{\tbase}\ \expr\ \expr\ \ectx} \eqcolor{\mid} \eqcolor{\exeq{x}{\typ}{\typ}\ \expr\ \expr\ \ectx}
                                         \\[1mm] 
\end{array}
$$
\textit{Reduction}\hfill\fbox{\evals{\expr}{\expr}}
$$
\begin{array}{rcll}
  \evals{\ctxapp{\ectx}{\expr}&}{&\ctxapp{\ectx}{\expr'}}, & \text{if}\ \evals{\expr}{\expr'}\\ 
  \evals{(\elam{x}{\typ}{\expr})\ \val&}{&\expr\subst{x}{\val}} & \\ 
  \evals{\bbbeq{(}{)\ \con_1}{\tbase}&}{&\bbbeq{(}{)}{(\con_1,\tbase)}} & \\ 
  \evals{\bbbeq{(}{)\ \con_2}{(\con_1,\tbase)}&}{&\con_1 = \con_2}, & \textit{syntactic equality on constants} \\ 
\end{array}
$$

\caption{Syntax and Dynamic Semantics of \corelaneq.}
\label{fig:coresyntax}
\end{figure}

We present \corelaneq, 
a core calculus with $R$efinement types and type-indexed $E$quality (Figure~\ref{fig:coresyntax}).

\paragraph{Expressions}
\corelaneq expressions include constants (booleans, unit, and equality operations 
on base types), variables, lambda abstraction, and application. 
There are also two primitives to prove propositional equality: 
\ebeq{\tbase} and \exeq{x}{\typ_x}{\typ}  construct proofs of equality at 
base and function types, respectively.
Equality proofs take three arguments: the two expressions equated and a proof of their equality; proofs at base type are trivial, of type \tunit, but higher types use functional extensionality.
These two primitives correspond to $\BEq$ and $\XEq$ constructors of \S\ref{sec:eqrt-gadt}; we did not encode congruence closure since it can be proved 
by induction on expressions, which is impossible in Haskell. 

\paragraph{Values} 
The values of \corelaneq are constants, functions, 
and equality proofs with converged proofs. 

\paragraph{Types}
\corelaneq's \emph{basic types} are booleans and unit. 
Basic types are refined with boolean expressions \refa in
\emph{refinement types} \tref{\rbind}{\tbase}{\refa}, 
which denote all expressions of base type \tbase that satisfy the 
refinement \refa. 
In addition to refinements, \corelaneq's types also include \emph{dependent function types}
\tfun{x}{\typ_x}{\typ} with arguments of type $\typ_x$ 
and result type \typ, where $\typ$ can refer back to the argument $x$. 
Finally, types include our \emph{propositional equality}
\eqrt{\typ}{\expr_1}{\expr_2}, which denotes a proof of equality
between the two expressions $\expr_1$ and $\expr_2$ of type $\typ$.
We write \tbase to mean the trivial refinement type \tref{\rbind}{\tbase}{\etrue}.
We omit polymorphic types to avoid known and resolved metatheoretical problems~\cite{toplas17}. Yet, $\exeqName$ equality primitive is defined as a family of 
operators, one for each refinement function type, capturing the essence of polymorphic 
function equality.

\paragraph{Environments}
The typing environment \env binds variables to types, 
the (semantic typing) closing substitution \model binds variables to values, and 
the (logical relation) pending substitution \rmodel binds variables to pairs of equivalent values.

\paragraph{Runtime Semantics}
The relation \evals{\cdot}{\cdot} evaluates \corelan expressions using contextual, 
small step, call-by-value semantics (Figure~\ref{fig:coresyntax}, bottom). 
The semantics are standard with \ebeq{\tbase} and \exeq{x}{\typ_x}{\typ}
evaluating proofs but not the equated terms.
Let \goesto{\cdot}{\cdot} be the reflexive, transitive closure of 
\evals{\cdot}{\cdot}.

\paragraph{Type Interpretations}
\begin{figure*}
\small
\[
  \begin{array}{rcr@{~}c@{~}l}
    \interp{\tref{\rbind}{\tbase}{\refa}} & \doteq & \{\expr & \mid & \ \goesto{\expr}{\val} \wedge \hasbtype{}{\expr}{\tbase} \wedge \goesto{\refa\subst{\rbind}{\expr}}{\etrue} \}\\
  \interp{\tfun{x}{\typ_x}{\typ}} & \doteq & \{\expr & \mid &  \forall \expr_x \in \interp{\typ_x}\!.\ \expr\ \expr_x \in \interp{\typ\subst{x}{\expr_x}} \}\\
  \interp{\eqrt{\tbase}{\expr_l}{\expr_r}} 
    & \doteq 
    & \{\expr & \mid &  \ \hasbtype{}{\expr}{\eqt{\tbase}} 
                \wedge \goesto{\expr}{\ebeq{\tbase}\ {\expr_l}\ {\expr_r}\ {\expr_{pf}}}  \wedge
           \goesto{\bbbeq{\expr_l}{\expr_r}{\tbase}}{\etrue}
            \}\\
  \interp{\eqrt{\tfun{x}{\typ_x}{\typ}}{\expr_l}{\expr_r}}
    & \doteq & 
    \{\expr & \mid &  \ \hasbtype{}{\expr}{\eqt{\unrefine{\tfun{x}{\typ_x}{\typ}}}} 
                    \wedge \goesto{\expr}{\exeqName_{\_}\ {\expr_l}\ {\expr_r}\ {\expr_{pf}}} \\
    &&&  \wedge & 
            \expr_l, \expr_r \in \interp{\tfun{x}{\typ_x}{\typ}} \wedge 
            \forall \expr_x\in\interp{\typ_x}. \expr_{pf}\ \expr_x \in \interp{\eqrt{\typ\subst{x}{\expr_x}}{\expr_l\ \expr_x}{\expr_r\ \expr_x}}
            \}\\
\end{array}
\]
\nv{Changed eq definitions to align with operational semantics change}
\caption{Semantic typing: a unary syntactic logical relation interprets types.}
\label{fig:semantic-typing}
\end{figure*}
Semantic typing uses a unary logical relation to
interpret types in a syntactic term model
(closed terms, Figure~\ref{fig:semantic-typing}; open terms, Figure~\ref{fig:coretyping}).

The interpretation of the base type \tref{\rbind}{\tbase}{\refa}
includes all expressions which yield \tbase-value \val
that satisfy the refinement, \ie \refa evaluates to true on \val.
To decide the unrefined type of an expression we use
\hasbtype{}{\expr}{\tbase} (defined in \iffull\S\ref{subsec:appendix:base-type-checking}\fi).
The interpretation of function types \tfun{x}{\typ_x}{\typ} is \emph{logical}: it includes 
all expressions that yield \typ-results when applied to $\typ_x$ arguments. 
The interpretation of base-type equalities \eqrt{\tbase}{\expr_l}{\expr_r}
includes all expressions that satisfy the basic typing 
(\eqt{\typ} is the unrefined version of \eqrt{\typ}{\expr_l}{\expr_r}) 
and reduce to a basic equality proof whose first arguments reduce to equal \tbase-constants.
Finally, the interpretation of the function equality type 
\eqrt{\tfun{x}{\typ_x}{\typ}}{\expr_l}{\expr_r} includes all expressions 
that satisfy the basic typing (based on the \unrefine{\cdot} operator; \iffull\S\ref{subsec:appendix:base-type-checking}\fi).
These expressions reduce to a proof
whose first two arguments 
are functions of type \tfun{x}{\typ_x}{\typ}
and the third, proof argument
takes $\typ_x$ arguments to equality proofs of type \eqrt{\typ\subst{x}{\expr_x}}{\expr_l\ \expr_x}{\expr_r\ \expr_x}.
We write these proofs as $\exeqName_{\_}$, since the type index does not need to be syntactically equal to the index of the type.

\paragraph{Constants}
For simplicity, \corelan constants are only the two boolean values,
unit, and equality operators for basic types.
For each basic type $\tbase$,
we define the type indexed 
``computational'' equality $\bbbeqsym{\tbase}$. 
For two constants $\con_1$ and $\con_2$ of basic type $\tbase$, 
$\con_1\ \bbbeqsym{\tbase}\ \con_2$ evaluates in one step to 
$\bbbeq{(}{)}{(\con_1,\tbase)}\ \con_2$, which then steps to \etrue when $\con_1$ and $\con_2$ are the same and \efalse otherwise.

Each constant \con has the type \tycon{\con}\iffull, as defined below.
$$
\begin{array}{rcl}
  \tycon{\etrue} & \doteq & \tref{\rbind}{\tbool}{\bbbeq{\rbind}{\etrue}{\tbool}}   \\
  \tycon{\efalse} & \doteq & \tref{\rbind}{\tbool}{\bbbeq{\rbind}{\efalse}{\tbool}} \\
  \tycon{\eunit} & \doteq & \tref{\rbind}{\tunit}{\bbbeq{\rbind}{\eunit}{\tunit}}   \\
  \tycon{\bbbeqsym{\tbase}} &\doteq& \tfun{x}{\tbase}{\tfun{y}{\tbase}{\tref{\vv}{\tbool}{\bbbeq{\vv}{(\bbbeq{x}{y}{\tbase})}{\tbool}}}}
\end{array}
$$

\fi
\iffull
Our system could of course be extended with further constants, as long as 
they belong in the interpretation of their type. 
This requirement is formally defined by the Property~\ref{property:constants} 
which, for the four constants of our system is proved in~Theorem~\ref{theorem:constant-property}
\begin{property}[Constants]\label{property:constants}
$\con \in \interp{\tycon{\con}}$
\end{property}
\fi

\subsection{Static Semantics of \corelan}
\label{subsec:static-semantics}
\corelan's static semantics comes in two parts:
as typing judgments (\S\ref{subsec:typing})
and as a binary logical relation characterizing equivalence (\S\ref{subsec:logical-relation}).

\subsubsection{Typing of \corelan}
\label{subsec:typing}
Type checking in \corelaneq uses three mutually recursive judgments (Figure~\ref{fig:coretyping}):
\emph{type checking}, \hastype{\env}{\expr}{\typ}, for when \expr has type \typ in \env;
\emph{well formedness}, \iswellformed{\env}{\typ}, for when when \typ is well formed in \env; and
\emph{subtyping}, \issubtype{\env}{\typ_l}{\typ_r} , for when when  $\typ_l$ is a subtype of $\typ_r$ in \env.
%
\begin{figure*}
\small
\textit{Type checking}\hfill\fbox{\hastype{\env}{\expr}{\typ}}
\begin{center}
$$
\inference{
    \hastype{\env}{\expr}{\typ} & 
  \issubtype{\env}{\typ}{\typ'}
}{
    \hastype{\env}{\expr}{\typ'}
}[\tsub]
\quad
\inference{
    \hastype{\env}{\expr}{\tref{\vv}{\tbase}{\refa}}  
}{
    \hastype{\env}{\expr}{\tref{\vv}{\tbase}{\bbbeq{\vv}{\expr}{\tbase}}}
}[\tself]
\quad
\inference{}{
    \hastype{\env}{\con}{\tycon{\con}}
}[\tcon]
$$

$$
\inference{
    \envBind{x}{\typ} \in \env
}{
    \hastype{\env}{x}{\typ}
}[\tvar]
\quad
\inference{
  \iswellformed{\env}{\typ_x} & 
  \hastype{\env, \envBind{x}{\typ_x}}{\expr}{\typ} 
}{
    \hastype{\env}{\elam{x}{\typ_x}{\expr}}{\tfun{x}{\typ_x}{\typ}}
}[\tlam]
\quad
\inference{    
    \hastype{\env}{\expr_x}{\typ_x} & 
    \hastype{\env}{\expr}{\tfun{x}{\typ_x}{\typ}} 
}{
    \hastype{\env}{\expr\ \expr_x}{\typ\subst{x}{\expr_x}}
}[\tapp]
$$

$$
\inference{
 \hastype{\env}{\expr_l}{\typ_l} && 
 \issubtype{\env}{\typ_l}{\tref{x}{\tbase}{\etrue}} \\
 \hastype{\env}{\expr_r}{\typ_r} &&
 \issubtype{\env}{\typ_r}{\tref{x}{\tbase}{\etrue}} \\ 
 \hastype{\env, \envBind{l}{\typ_l}, \envBind{r}{\typ_r}}{\expr}{\tref{\rbind}{\tunit}{\bbbeq{l}{r}{\tbase}}}
}{
  \hastype{\env}{\ebeq{\tbase}\ \expr_l\ \expr_r\ \expr}{\eqrt{\tbase}{\expr_l}{\expr_r}}
}[\teqbase]
\quad
\eqcolor{
\inference{
 \hastype{\env}{\expr_l}{\typ_l} && 
 \issubtype{\env}{\typ_l}{\tfun{x}{\typ_x}{\typ}} \\
 \hastype{\env}{\expr_r}{\typ_r} &&
 \issubtype{\env}{\typ_r}{\tfun{x}{\typ_x}{\typ}} &&
 \iswellformed{\env}{\tfun{x}{\typ_x}{\typ}} \\  
 \hastype{\env, \envBind{l}{\typ_l}, \envBind{r}{\typ_r}}{\expr}{(\tfun{x}{\typ_x}{\eqrt{\typ}{l\ x}{r\ x}})} && 
}{
  \hastype{\env}{\exeq{x}{\typ_x}{\typ}\ \expr_l\ \expr_r\ \expr}{\eqrt{\tfun{x}{\typ_x}{\typ}}{\expr_l}{\expr_r}}
}[\teqfun]
}
$$
\end{center}
\vspace{0.1mm}

\textit{Well-formedness}\hfill\fbox{\iswellformed{\env}{\typ}}\qquad\fbox{\envwellformed{\env}} 
$$
\inference{
   \hasbtype{\unrefine{\env},\envBind{\rbind}{\tbase}}{\refa}{\tbool}
}{
    \iswellformed{\env}{\tref{\rbind}{\tbase}{\refa}} 
}[\wfBase]
\quad
\inference{    
    \iswellformed{\env}{\typ_x} & 
    \iswellformed{\env, \envBind{x}{\typ_x}}{\typ}
}{
    \iswellformed{\env}{\tfun{x}{\typ_x}{\typ}}
}[\wfFun]
$$
$$\inference{
 \iswellformed{\env}{\typ} & 
 \hastype{\env}{\expr_l}{\typ} & 
 \hastype{\env}{\expr_r}{\typ}
}{
  \iswellformed{\env}{\eqrt{\typ}{\expr_l}{\expr_r}}
}[\wfEq]
\quad
\inference{}{\envwellformed{\emptyEnv}}[\wfEmpty]
\quad
\inference{\envwellformed{\env} && \iswellformed{\env}{\typ}}{\envwellformed{\env, \envBind{x}{\typ}}}[\wfBind]
$$

\vspace{0.1mm}

\textit{Subtyping}\hfill\fbox{$\issubtype{\env}{\typ}{\typ}$}
$$
\inference{
    \forall \model\in\interp{\env},~ \interp{\modelapp{\model}{\tref{\rbind}{\tbase}{\refa}}}
                          \subseteq \interp{\modelapp{\model}{\tref{\rbind'}{\tbase}{\refa'}}}
}{
    \issubtype{\env}{\tref{\rbind}{\tbase}{\refa}}{\tref{\rbind'}{\tbase}{\refa'}} 
}[\subBase]
\quad
\inference{    
    \issubtype{\env}{\typ'_x}{\typ_x} & 
    \issubtype{\env, \envBind{x}{\typ'_x}}{\typ}{\typ'}
}{
    \issubtype{\env}{\tfun{x}{\typ_x}{\typ}}{\tfun{x}{\typ'_x}{\typ'}}
}[\subFun]
$$
$$
\eqcolor{
\inference{
 \issubtype{\env}{\typ}{\typ'} & \issubtype{\env}{\typ'}{\typ}
}{
  \issubtype{\env}{\eqrt{\typ}{\expr_l}{\expr_r}}{\eqrt{\typ'}{\expr_l}{\expr_r}}
}[\subEq]
}
$$
\vspace{0.1mm}

\textit{Semantic typing and closing substitutions}\hfill\fbox{$\model \in \interp{\env}$}\qquad\fbox{\hassemtype{\env}{\expr}{\typ}}

\[
\inference{}{\emptyset \in \interp{\emptyEnv}}[\cbase]
\quad
\inference{
  \val \in \interp{\typ} && 
  \model \in \interp{\env\subst{x}{\val}}
}{
  \modelBind{x}{\val}, \model  \in \interp{\envBind{x}{\typ}, \env}
}[\csub]
\quad
\begin{array}{c}
\hassemtype{\env}{\expr}{\typ} \Leftrightarrow 
    \forall \model\in\interp{\env},~ \hasdtype{\model}{\expr}{\typ} 
\end{array}
\]
\caption{Typing of \corelaneq.}
\label{fig:coretyping}
\end{figure*}

\paragraph{Type Checking}
Beyond the conventional rules for refinement type systems~\cite{DBLP:conf/ifipTCS/OuTMW04,Hybrid,LT2008}, the interesting rules are concerned with equality (\teqbase, \teqfun).

The rule \teqbase assigns to the expression $\ebeq{\tbase}\ \expr_l\ \expr_r\ \expr$
the type \eqrt{\tbase}{\expr_l}{\expr_r}. 
To do so, we guess types $\typ_l$ and $\typ_r$ that fit
$\expr_l$ and $\expr_r$, respectively. 
Both these types should be subtypes of $\tbase$ that are \emph{strong} enough
to derive that if
\envBind{l}{\typ_l} and \envBind{r}{\typ_r}, then the proof argument $\expr$
has type $\tref{\_}{\tunit}{\bbbeq{l}{r}{\tbase}}$.
%
Our formal model allows checking of strong, selfified
types (rule \tself), but does not define an algorithmic procedure 
to generate them. 
In Liquid Haskell,  
type inference~\cite{LT2008} 
automatically and algorithmically derives such strong types.
We don't encumber \corelan with inference, since, formally speaking, we
can always guess any type that inference can derive. 

The rule \teqfun gives the expression $\exeq{x}{\typ_x}{\typ}\ \expr_l\ \expr_r\ \expr$
type \eqrt{\tfun{x}{\typ_x}{\typ}}{\expr_l}{\expr_r}. 
As in \teqbase, we guess strong types $\typ_l$ and $\typ_r$ to stand for
$\expr_l$ and $\expr_r$ such that with \envBind{l}{\typ_l} and \envBind{r}{\typ_r}, the proof argument $\expr$
should have type $\tfun{x}{\typ_x}{\eqrt{\typ}{l\ x}{r\ x}}$, \ie it should prove that $l$ and $r$ are extensionally equal.
We require that the index $\tfun{x}{\typ_x}{\typ}$ is well formed as technical bookkeeping.

\paragraph{Well Formedness}
Refinements should be booleans
(\wfBase); functions are treated in the usual way (\wfFun); and
the propositional equality 
\eqrt{\typ}{\expr_l}{\expr_r} is well formed when
the  expressions $\expr_l$ and $\expr_r$ are typed at the index \typ, which is also well formed (\wfEq). 

\paragraph{Subtyping}
Basic types are related by
set inclusion on the interpretation of those types (\subBase, and Figure~\ref{fig:semantic-typing}). 
Concretely, for all closing substitutions 
(\cbase, \csub)
the interpretation of the left-hand side type should be a subset 
of the right-hand side type. 
The rule \subFun implements the usual (dependent) function subtyping.
Finally, \subEq reduces subtyping of equality types to 
subtyping of the type indices, while the expressions to be equated 
remain unchanged. 
Even though covariant treatment of the type index would suffice 
for our metatheory, we treat the type index invariantly
to be consistent with the implementation (\S\ref{sec:rules})
where the GADT encoding of |$\PEq$| is invariant.
Our subtyping rule allows equality proofs
between functions with 
convertible types
(\S\ref{sec:eg:refdom}).


\subsubsection{Equivalence Logical Relation for \corelan}
\label{subsec:logical-relation}

\begin{figure}
\small
\arraycolsep=0pt
\[
\begin{array}{r@{~}c@{~}l@{~}lcl}
\multicolumn{6}{l}{\textit{Value equivalence relation}\hfill\fbox{\relates{\rmodel}{\val}{\val}{\typ}}}\\[2mm]
  \relates{\rmodel}{\con&}{&\con&}{\tref{\rbind}{\tbase}{\refa}} 
  & \quad\doteq\quad & 
  \hasbtype{}{\con}{\tbase} \wedge 
  \goesto{\modelapp{\rmodel_1}{\refa\subst{\rbind}{\con}}}{\etrue} \wedge
  \goesto{\modelapp{\rmodel_2}{\refa\subst{\rbind}{\con}}}{\etrue} \\ 
  \relates{\rmodel}{\val_1&}{&\val_2&}{\tfun{x}{\typ_x}{\typ}}  
  & \quad\doteq\quad & 
  \forall \relates{\rmodel}{\val_3}{\val_4}{\typ_x}. ~
  \relates{\rmodel,\rmodelBind{\val_3}{\val4_4}{x}}{\val_1\ \val_3}{\val_2\ \val_4}{\typ}\\
  \eqcrelates{\eqcolor{\rmodel}}{\eqcolor{\val_1}&}{&\eqcolor{\val_2}&}{\eqcolor{\eqrt{\typ}{\expr_l}{\expr_r}}}
  & \eqcolor{\quad\doteq\quad} & 
  \eqcolor{\relates{\rmodel}{\modelapp{\rmodel_1}{\expr_l}}{\modelapp{\rmodel_2}{\expr_r}}{\typ}} \\[3mm]
  \multicolumn{6}{l}{\textit{Expression equivalence relation}\hfill\fbox{\relates{\rmodel}{\expr}{\expr}{\typ}}}\\[2mm]
  \relates{\rmodel}{\expr_1&}{&\expr_2&}{\typ}  & \quad\doteq\quad & 
  \goesto{\expr_1}{\val_1},\quad  
  \goesto{\expr_2}{\val_2},\quad
  \relates{\rmodel}{\val_1}{\val_2}{\typ}\\[3mm]
  \multicolumn{6}{l}{\textit{Open expression equivalence relation}\hfill\fbox{$\rmodel \in \env$} \qquad \fbox{\relatesEnv{\env}{\expr}{\expr}{\typ}}}\\[2mm]
\multicolumn{6}{c}{
  {\rmodel \in \env}  ~~\doteq~~
  \forall \envBind{x}{\typ} \in \env,\ \relates{\rmodel}{\rmodel_1(x)}{\rmodel_2(x)}{\typ}

  \qquad\quad

    { \relatesEnv{\env}{\expr_1}{\expr_2}{\typ}}   ~~\doteq~~
  \forall \rmodel \in \env,\ \relates{\rmodel}{\modelapp{\rmodel_1}{\expr_1}}{\modelapp{\rmodel_2}{\expr_2}}{\typ}

} \\
\end{array}
\]
\caption{Definition of equivalence logical relation.}
\label{fig:equivalence-def}
\end{figure}

We characterize equivalence with a term-model binary logical relation. We lift a relation on closed values to closed and then open expressions (Figure~\ref{fig:equivalence-def}).
Instead of directly substituting in type indices, all three relations use \emph{pending substitutions} $\rmodel$, which map variables to pairs of equivalent values.

\paragraph{Closed Values and Expressions}
We read \relates{\rmodel}{\val_1}{\val_2}{\typ} as saying that
values $\val_1$ and $\val_2$ are related under the type $\typ$ with
pending substitutions \rmodel. The relation is defined as a fixpoint on types, 
noting that the propositional equality on a type, \eqrt{\typ}{\expr_1}{\expr_2}, is structurally larger than the type \typ.

For refinement types \tref{\rbind}{\tbase}{\refa}, related values must 
be the same constant $\con$. Further, this constant should 
actually be a \tbase-constant and it should actually satisfy the refinement \refa,
\ie substituting \con for $x$ in \refa
should evaluate to \etrue
under either pending substitution ($\rmodel_1$ or $\rmodel_2$). 
Two values of function type are equivalent when applying them to
equivalent arguments yield equivalent results. Since we have dependent types, we record the arguments in the pending substitution for later substitution in the codomain.
Two proofs of equality are equivalent
when the two equated expressions are equivalent in the logical relation at type-index $\typ$---equality proofs `reflect' the logical relation.
Since the equated expressions appear in the type itself, they may be open,
referring to variables in the pending substitution \rmodel. 
Thus we use \rmodel to close these expressions, using the logical relation on
\modelapp{\rmodel_1}{\expr_l} and \modelapp{\rmodel_2}{\expr_r}.
Following the proof irrelevance notion of refinement typing, 
the equivalence of equality proofs does not relate the proof terms---in fact, 
it doesn't even \emph{inspect} the proofs $\val_1$ and $\val_2$.
\mmg{this may pose a problem for any future completeness proof}

Two closed expressions $\expr_1$ and $\expr_2$ are equivalent on type \typ 
with pending substitions \rmodel, written \relates{\rmodel}{\expr_1}{\expr_2}{\typ},
\textit{iff} they respectively evaluate to equivalent values $\val_1$ and $\val_2$. 

\paragraph{Open Expressions}
A pending substitution \rmodel satisfies a typing environment \env
when its bindings are relates pairs of values at the type in \env.
Two open expressions, with variables from  \env 
are equivalent on type \typ, written 
\relatesEnv{\env}{\expr_1}{\expr_2}{\typ}, \textit{iff} 
for each \rmodel that satisfies \env, we have
\relates{\rmodel}{\modelapp{\delta_1}{\expr_1}}{\modelapp{\delta_2}{\expr_2}}{\typ}. 
The expressions $\expr_1$ and $\expr_2$ and the type \typ 
might refer to variables in the environment \env. 
We use \rmodel to close the expressions eagerly, while we close the type lazily:
we apply \rmodel in the refinement and equality cases of the closed value
equivalence relation.

\subsection{Metaproperties: \PEq is an Equivalence Relation}
\label{subsec:metaproperties}

Finally, we show various metaproperties of \corelan. 
Theorem~\ref{thm:soundness} proves soundness of 
syntactic typing with respect to semantic typing. 
Theorem~\ref{thm:eq-relation} proves that propositional equality 
implies equivalence in the term model. 
Theorems~\ref{thm:theory:equality-logical-relation} and~\ref{thm:theory:properties}
prove that both the equivalence relation and propositional equality define equivalences,
\ie satisfy the three equality axioms.%
{}
All the proofs are in \iffull Appendix~\ref{sec:proofs-metatheory}\fi. 

\corelaneq is semantically sound: syntactically well typed programs are also semantically well typed.
\begin{theorem}[Typing is Sound]
  \label{thm:soundness}
  If \hastype{\env}{\expr}{\typ}, then \hassemtype{\env}{\expr}{\typ}. 
\end{theorem}
\noindent
The proof goes by induction on the derivation tree. 
\iffull
Our system could not be proved sound using purely syntactic
techniques, like progress and preservation~\cite{Wright94syntactic}, for two reasons.
First, and most essentially, \subBase needs to quantify over all
closing substitutions and purely syntactic approaches flirt with
non-monotonicity (though others have attempted syntactic approaches in similar systems~\cite{Zalewski20WGT}).
Second, and merely coincidentally, our system does not enjoy subject
reduction. In particular, \subEq allows us to change the type index of
propositional equality, but not the term index.
Why? Consider the term: \[ (\elam{x}{\tref{x}{\tbool}{\etrue}}{\ebeq{\tbool}\ x\ x\ ()})\ e \] such that $\evals{e}{e'}$ for some $e'$.
The whole application has type \eqrt{\tbool}{e}{e}; after we take a
step, it will have type \eqrt{\tbool}{e'}{e'}. Subject reduction demands
that the latter is a subtype of the former. We have
\[ \parreds{\eqrt{\tbool}{e}{e}}{\eqrt{\tbool}{e'}{e'}} \] so we could
recover subject reduction by allowing a supertype's terms to
parallel reduce (or otherwise convert) to a subtype's terms.
Adding this condition would not be hard:
the logical relations' metatheory already demands a
variety of lemmas about parallel reduction, relegated to
supplementary material\iffull (Appendix~\ref{app:parred})\fi\ to avoid distraction
and preserve space for our main contributions. We haven't made this change because subject reduction isn't necessary for our purposes.
\fi

\begin{theorem}[\PEq is Sound]
  \label{thm:eq-relation}
  If \hastype{\env}{\expr}{\eqrt{\typ}{\expr_1}{\expr_2}}, 
  then \relatesEnv{\env}{\expr_1}{\expr_2}{\typ}.
\end{theorem}
\noindent
The proof \iffull(see Theorem~\ref{proofs-thm:eq-relation})\fi
is a corollary of the fundamental property of the logical relation\iffull (Theorem~\ref{thm:fundamental})\fi, \ie 
if \hastype{\env}{\expr}{\typ} then \relatesEnv{\env}{\expr}{\expr}{\typ}, 
which is proved in turn by induction on the typing derivation.  

\begin{theorem}[The logical relation is an Equivalence]
  \label{thm:theory:equality-logical-relation}
  $\relatesEnv{\env}{\expr_1}{\expr_2}{\typ}$ is reflexive, symmetric, and transitive\iffull:
  \begin{itemize}[leftmargin=*]
    \item \textit{Reflexivity:} If \hastype{\env}{\expr}{\typ}, then $\relatesEnv{\env}{\expr}{\expr}{\typ}$.
    \item \textit{Symmetry:} If $\relatesEnv{\env}{\expr_1}{\expr_2}{\typ}$, then $\relatesEnv{\env}{\expr_2}{\expr_1}{\typ}$. 
    \item \textit{Transitivity:} If $\hastype{\env}{\expr_2}{\typ}$, $\relatesEnv{\env}{\expr_1}{\expr_2}{\typ}$, and $\relatesEnv{\env}{\expr_2}{\expr_3}{\typ}$, 
                                 then $\relatesEnv{\env}{\expr_1}{\expr_3}{\typ}$. 
\end{itemize}\fi
\end{theorem}
\noindent
Reflexivity is also called the \emph{fundamental property} of the logical relation.
The other proofs go by structural 
induction on \typ\iffull (Theorem~\ref{proofs:equality-logical-relation})\fi.
Transitivity requires reflexivity on $\expr_2$, 
so we also assume that $\hastype{\env}{\expr_2}{\typ}$. 
%
\begin{theorem}[\PEq is an Equivalence]
  \label{thm:theory:properties}
$\eqrt{\typ}{\expr_1}{\expr_2}$ is reflexive, symmetric, and
  transitive on equable types. That is, for all $\typ$ that do not contain
  equalities themselves:
  \begin{itemize}[leftmargin=*]
    \item \textit{Reflexivity:} If $\hastype{\env}{\expr}{\typ}$, then
      there exists $\val$ such that
      $\hastype{\env}{\val}{\eqrt{\typ}{\expr}{\expr}}$.

    \item \textit{Symmetry:} 
      If $\hastype{\env}{\val_{12}}{\eqrt{\typ}{\expr_1}{\expr_2}}$, then
      there exists $\val_{21}$ such that
      $\hastype{\env}{\val_{21}}{\eqrt{\typ}{\expr_2}{\expr_1}}$.

      \todo{we can change $\val_{12}$ to an expression when we correct issues with term indices}

    \item \textit{Transitivity:} 
      If  \hastype{\env}{\val_{12}}{\eqrt{\typ}{\expr_1}{\expr_2}}
      and \hastype{\env}{\val_{23}}{\eqrt{\typ}{\expr_2}{\expr_3}}, 
      then there exists $\val_{13}$ such that \hastype{\env}{\val_{13}}{\eqrt{\typ}{\expr_1}{\expr_3}}.

  \end{itemize}
\end{theorem}
\noindent
The proofs go by induction on $\typ$\iffull (Theorem~\ref{proofs:eqrt-equivalence})\fi. Reflexivity requires us to
generalize the inductive hypothesis to generate appropriate $\typ_l$ and $\typ_r$ for
the |PEq| proofs.

\section{Related Work}\label{sec:related}

\paragraph{Functional Extensionality and Subtyping with an SMT Solver}

\Fstar also uses a type-indexed \extName\ axiom
after having run into similar unsoundness issues~\cite{FstarUnsound}.
%
Their extensionality axiom makes a more roundabout connection with
SMT: function equality uses |==|\ifspace
\else, a proof-irrelevant, propositional Leibniz equality. \fi
They assume that their Leibniz equality coincides
with SMT equality. 
Liquid Haskell can't just copy \Fstar:
there are no dependent, inductive type definitions, nor a dedicated notion of propositions.
Our |$\PEq$| GADT approximates \Fstar's approach, with different
compromises.

Dafny's SMT encoding axiomatizes extensionality for data,
but not for functions~\cite{DBLP:conf/sigada/Leino12}. Function equality is
utterable but neither provable nor disprovable\ifspace
\else\ in their encoding into Z3\fi.
{}

\newcommand\selfty{\ensuremath{\mathsf{self}}}
\citet{DBLP:conf/ifipTCS/OuTMW04} introduce \emph{selfification},
which assigns singleton types using equality (as in our \tself rule).
{}
%
%
SAGE assigns selfified types to \emph{all} variables, implying equality
on functions~\cite{Knowles06sage}.
Dminor avoids the question: it lacks first-class functions~\cite{DBLP:journals/jfp/BiermanGHL12}.

\paragraph{Extensionality in Dependent Type Theories}

Functional extensionality (\extName) has a rich history of study.
Martin-L\"of type theory comes in a decidable, intensional flavor (ITT)~\cite{MARTINLOF197573} as well as an undecidable, extensional one (ETT)~\cite{martin1984intuitionistic}.
NuPRL implements ETT~\cite{DBLP:books/daglib/0068834}, while Coq and Agda implement ITT~\citeyearpar{the_coq_development_team_2020_3744225,10.5555/1813347.1813352}.
Lean's quotient-based reasoning can \emph{prove} \extName~\cite{DBLP:conf/cade/MouraKADR15}.
%
%
\todo{Isabelle? Not a dependent type theory.}
Extensionality axioms are independent of the rules of ITT; \extName\ is a common axiom, but is not consistent in every model of type theory\ifspace
\else\ \cite{vonGlehn}\fi.
\citet{10.1007/3-540-61780-9_68} shows that ETT is a conservative but less computational
extension of ITT with \extName\ and UIP.
\ifspace
\else
\citet{932499} and \citet{Altenkirch06towardsobservational} try to reconcile ITT and ETT.
\fi

Dependent type theories often care about equalities between
equalities, with axioms like UIP (all identity proofs are the same), K
(all identity proofs are |refl|), and univalence (identity proofs are
isomorphisms, and so not the same).
If we allowed equalities between equalities, we could add UIP.
Our propositional equality isn't exactly Leibniz equality, so axiom
K would be harder to encode.

Zombie's type theory uses an adaptation of a congruence closure
algorithm to automatically reason about equality~\cite{10.1145/2676726.2676974}.
Zombie can do some reasoning about equalities on functions but cannot show equalities based on bound variables.
Zombie is careful to omit a $\lambda$-congruence rule, which could be
used to prove \extName, ``which is not
compatible with [their] `very heterogeneous' treatment of equality''
[Ibid., \S{9}].

Cubical type theory
offers alternatives to our propositional
equality~\cite{DBLP:journals/corr/abs-1904-08562}. Such approaches may play better
with \Fstar's approach using dependent, inductive types than the
`flatter' approach we used for Liquid Haskell.
Univalent systems like cubical type theory get \extName\ `for free'---that is, for the price of the univalence
axiom or of cubical foundations.

\paragraph{Classy Induction: Inductive Proofs Using Typeclasses}
We used `classy induction'
to prove metaproperties of $\PEq$ inside Liquid Haskell
(\S\ref{sec:gadt-metatheory}), 
using ad-hoc polymorphism and general instances
to generate proofs that `cover' some class of types.
%
%
We did not \emph{invent} classy induction---it is a folklore technique \ifspace
\else that we named\fi.
We have seen five independent uses of ``classy induction'' in the literature~\cite{10.1145/1411204.1411218,10.1145/3009837.3009923,boulier:hal-01445835,DBLP:journals/jfp/DagandTT18,2019arXiv190905027T}.

Any typeclass system that accommodates ad-hoc polymorphism and a
notion of proof can use classy induction.
\citet{sozeau2008environnement} generates proofs of nonzeroness using
something akin to classy induction, though it goes by induction on the
operations used to build up arithmetic expressions in the (dependent!)
host language (\S{6.3.2}); he calls this the `programmation logique'
aspect of typeclasses.
Instance resolution is characterized as proof search over lemmas
(\S{7.1.3}).
\citet{10.1007/978-3-540-71067-7_23} introduce typeclasses to
Coq; their system can do induction by typeclasses, but they do not demonstrate
the idea in the paper.
{\ifspace
\else
Earlier work on typeclasses focused on
overloading~\cite{10.1145/75277.75283,10.5555/645420.652540,10.1145/158511.158698},
with no notion of classy induction even in settings with proofs~\cite{10.1007/BFb0028402}.
\fi}

\section{Conclusion}\label{sec:conclusion}

In a refinement type system with subtyping
a naive encoding of \extName\ is inconsistent. 
We explained the inconsistency by examples (that proved |false|)
and by standard type checking (where the equality domain is inferred as |false|). 
We implemented a type-indexed propositional equality that avoids this inconsistency
and validated it with a model calculus.
Several case studies demonstrate the range, effectiveness, and power
of our work.

\begin{acks}
We thank Conal Elliott for his help in exposing the inadequacy of the
naive functional extensionality encoding.
Stephanie Weirich, \'Eric Tanter, and Nicolas Tabareau offered
valuable insights into the folklore of classy induction.
\end{acks}

\bibliographystyle{ACM-Reference-Format}
\bibliography{references}


\begin{thebibliography}{46}


\ifx \showCODEN    \undefined \def \showCODEN     #1{\unskip}     \fi
\ifx \showDOI      \undefined \def \showDOI       #1{#1}\fi
\ifx \showISBNx    \undefined \def \showISBNx     #1{\unskip}     \fi
\ifx \showISBNxiii \undefined \def \showISBNxiii  #1{\unskip}     \fi
\ifx \showISSN     \undefined \def \showISSN      #1{\unskip}     \fi
\ifx \showLCCN     \undefined \def \showLCCN      #1{\unskip}     \fi
\ifx \shownote     \undefined \def \shownote      #1{#1}          \fi
\ifx \showarticletitle \undefined \def \showarticletitle #1{#1}   \fi
\ifx \showURL      \undefined \def \showURL       {\relax}        \fi
\providecommand\bibfield[2]{#2}
\providecommand\bibinfo[2]{#2}
\providecommand\natexlab[1]{#1}
\providecommand\showeprint[2][]{arXiv:#2}

\bibitem[\protect\citeauthoryear{Altenkirch and McBride}{Altenkirch and
  McBride}{2006}]%
        {Altenkirch06towardsobservational}
\bibfield{author}{\bibinfo{person}{Thorsten Altenkirch} {and}
  \bibinfo{person}{Conor McBride}.} \bibinfo{year}{2006}\natexlab{}.
\newblock \bibinfo{title}{Towards observational type theory}.
\newblock
\newblock
\newblock
\shownote{Unpublished manuscript.}


\bibitem[\protect\citeauthoryear{Barbosa, Reynolds, El~Ouraoui, Tinelli, and
  Barrett}{Barbosa et~al\mbox{.}}{2019}]%
        {Barbosa19}
\bibfield{author}{\bibinfo{person}{Haniel Barbosa}, \bibinfo{person}{Andrew
  Reynolds}, \bibinfo{person}{Daniel El~Ouraoui}, \bibinfo{person}{Cesare
  Tinelli}, {and} \bibinfo{person}{Clark Barrett}.}
  \bibinfo{year}{2019}\natexlab{}.
\newblock \showarticletitle{Extending SMT Solvers to Higher-Order Logic}. In
  \bibinfo{booktitle}{\emph{Automated Deduction -- CADE 27}},
  \bibfield{editor}{\bibinfo{person}{Pascal Fontaine}} (Ed.).
  \bibinfo{publisher}{Springer International Publishing},
  \bibinfo{address}{Cham}, \bibinfo{pages}{35--54}.
\newblock
\showISBNx{978-3-030-29436-6}


\bibitem[\protect\citeauthoryear{Barrett, Stump, and Tinelli}{Barrett
  et~al\mbox{.}}{2010}]%
        {BarST-RR-10}
\bibfield{author}{\bibinfo{person}{Clark Barrett}, \bibinfo{person}{Aaron
  Stump}, {and} \bibinfo{person}{Cesare Tinelli}.}
  \bibinfo{year}{2010}\natexlab{}.
\newblock \bibinfo{booktitle}{\emph{{The SMT-LIB Standard: Version 2.0}}}.
\newblock \bibinfo{type}{{T}echnical {R}eport}.
  \bibinfo{institution}{Department of Computer Science, The University of
  Iowa}.
\newblock
\newblock
\shownote{Available at {\tt www.SMT-LIB.org}.}


\bibitem[\protect\citeauthoryear{Bierman, Gordon, Hritcu, and
  Langworthy}{Bierman et~al\mbox{.}}{2012}]%
        {DBLP:journals/jfp/BiermanGHL12}
\bibfield{author}{\bibinfo{person}{Gavin~M. Bierman},
  \bibinfo{person}{Andrew~D. Gordon}, \bibinfo{person}{Catalin Hritcu}, {and}
  \bibinfo{person}{David~E. Langworthy}.} \bibinfo{year}{2012}\natexlab{}.
\newblock \showarticletitle{Semantic subtyping with an {SMT} solver}.
\newblock \bibinfo{journal}{\emph{J. Funct. Program.}} \bibinfo{volume}{22},
  \bibinfo{number}{1} (\bibinfo{year}{2012}), \bibinfo{pages}{31--105}.
\newblock
\urldef\tempurl%
\url{https://doi.org/10.1017/S0956796812000032}
\showDOI{\tempurl}


\bibitem[\protect\citeauthoryear{Boulier, P{\'e}drot, and Tabareau}{Boulier
  et~al\mbox{.}}{2017}]%
        {boulier:hal-01445835}
\bibfield{author}{\bibinfo{person}{Simon Boulier},
  \bibinfo{person}{Pierre-Marie P{\'e}drot}, {and} \bibinfo{person}{Nicolas
  Tabareau}.} \bibinfo{year}{2017}\natexlab{}.
\newblock \showarticletitle{{The next 700 syntactical models of type theory}}.
  In \bibinfo{booktitle}{\emph{{Certified Programs and Proofs (CPP 2017)}}}.
  \bibinfo{address}{Paris, France}, \bibinfo{pages}{182 -- 194}.
\newblock
\urldef\tempurl%
\url{https://doi.org/10.1145/3018610.3018620}
\showDOI{\tempurl}


\bibitem[\protect\citeauthoryear{Cheney and Hinze}{Cheney and Hinze}{2003}]%
        {cheneyhinze03gadt}
\bibfield{author}{\bibinfo{person}{James Cheney} {and} \bibinfo{person}{Ralf
  Hinze}.} \bibinfo{year}{2003}\natexlab{}.
\newblock \bibinfo{booktitle}{\emph{First-Class Phantom Types}}.
\newblock \bibinfo{type}{{T}echnical {R}eport}. \bibinfo{institution}{Cornell
  University}.
\newblock


\bibitem[\protect\citeauthoryear{Constable, Allen, Bromley, Cleaveland, Cremer,
  Harper, Howe, Knoblock, Mendler, Panangaden, Sasaki, and Smith}{Constable
  et~al\mbox{.}}{1986}]%
        {DBLP:books/daglib/0068834}
\bibfield{author}{\bibinfo{person}{Robert~L. Constable},
  \bibinfo{person}{Stuart~F. Allen}, \bibinfo{person}{Mark Bromley},
  \bibinfo{person}{Rance Cleaveland}, \bibinfo{person}{J.~F. Cremer},
  \bibinfo{person}{R.~W. Harper}, \bibinfo{person}{Douglas~J. Howe},
  \bibinfo{person}{Todd~B. Knoblock}, \bibinfo{person}{N.~P. Mendler},
  \bibinfo{person}{Prakash Panangaden}, \bibinfo{person}{James~T. Sasaki},
  {and} \bibinfo{person}{Scott~F. Smith}.} \bibinfo{year}{1986}\natexlab{}.
\newblock \bibinfo{booktitle}{\emph{Implementing mathematics with the Nuprl
  proof development system}}.
\newblock \bibinfo{publisher}{Prentice Hall}.
\newblock
\showISBNx{978-0-13-451832-9}
\urldef\tempurl%
\url{http://dl.acm.org/citation.cfm?id=10510}
\showURL{%
\tempurl}


\bibitem[\protect\citeauthoryear{Constable and Smith}{Constable and
  Smith}{1987}]%
        {constable1987partial}
\bibfield{author}{\bibinfo{person}{Robert~L Constable} {and}
  \bibinfo{person}{Scott~Fraser Smith}.} \bibinfo{year}{1987}\natexlab{}.
\newblock \bibinfo{booktitle}{\emph{Partial objects in constructive type
  theory}}.
\newblock \bibinfo{type}{{T}echnical {R}eport}. \bibinfo{institution}{Cornell
  University}.
\newblock


\bibitem[\protect\citeauthoryear{Dagand, Tabareau, and Tanter}{Dagand
  et~al\mbox{.}}{2018}]%
        {DBLP:journals/jfp/DagandTT18}
\bibfield{author}{\bibinfo{person}{Pierre{-}{\'{E}}variste Dagand},
  \bibinfo{person}{Nicolas Tabareau}, {and} \bibinfo{person}{{\'{E}}ric
  Tanter}.} \bibinfo{year}{2018}\natexlab{}.
\newblock \showarticletitle{Foundations of dependent interoperability}.
\newblock \bibinfo{journal}{\emph{J. Funct. Program.}}  \bibinfo{volume}{28}
  (\bibinfo{year}{2018}), \bibinfo{pages}{e9}.
\newblock
\urldef\tempurl%
\url{https://doi.org/10.1017/S0956796818000011}
\showDOI{\tempurl}


\bibitem[\protect\citeauthoryear{de~Moura, Kong, Avigad, van Doorn, and von
  Raumer}{de~Moura et~al\mbox{.}}{2015}]%
        {DBLP:conf/cade/MouraKADR15}
\bibfield{author}{\bibinfo{person}{Leonardo~Mendon{\c{c}}a de Moura},
  \bibinfo{person}{Soonho Kong}, \bibinfo{person}{Jeremy Avigad},
  \bibinfo{person}{Floris van Doorn}, {and} \bibinfo{person}{Jakob von
  Raumer}.} \bibinfo{year}{2015}\natexlab{}.
\newblock \showarticletitle{The Lean Theorem Prover (System Description)}. In
  \bibinfo{booktitle}{\emph{Automated Deduction - {CADE-25} - 25th
  International Conference on Automated Deduction, Berlin, Germany, August 1-7,
  2015, Proceedings}} \emph{(\bibinfo{series}{Lecture Notes in Computer
  Science}, Vol.~\bibinfo{volume}{9195})},
  \bibfield{editor}{\bibinfo{person}{Amy~P. Felty} {and} \bibinfo{person}{Aart
  Middeldorp}} (Eds.). \bibinfo{publisher}{Springer},
  \bibinfo{pages}{378--388}.
\newblock
\urldef\tempurl%
\url{https://doi.org/10.1007/978-3-319-21401-6\_26}
\showDOI{\tempurl}


\bibitem[\protect\citeauthoryear{FStarLang}{FStarLang}{2018}]%
        {FstarUnsound}
\bibfield{author}{\bibinfo{person}{Github FStarLang}.}
  \bibinfo{year}{2018}\natexlab{}.
\newblock \bibinfo{title}{Functional Equality Discussions in F*}.
\newblock
\newblock
\newblock
\shownote{\url{https://github.com/FStarLang/FStar/blob/cba5383bd0e84140a00422875de21a8a77bae116/ulib/FStar.FunctionalExtensionality.fsti\#L133-L134}
  and \url{https://github.com/FStarLang/FStar/issues/1542} and
  \url{https://github.com/FStarLang/FStar/wiki/SMT-Equality-and-Extensionality-in-F\%2A}.}


\bibitem[\protect\citeauthoryear{Guillemette and Monnier}{Guillemette and
  Monnier}{2008}]%
        {10.1145/1411204.1411218}
\bibfield{author}{\bibinfo{person}{Louis-Julien Guillemette} {and}
  \bibinfo{person}{Stefan Monnier}.} \bibinfo{year}{2008}\natexlab{}.
\newblock \showarticletitle{A Type-Preserving Compiler in Haskell}. In
  \bibinfo{booktitle}{\emph{Proceedings of the 13th ACM SIGPLAN International
  Conference on Functional Programming}} (Victoria, BC, Canada)
  \emph{(\bibinfo{series}{ICFP ’08})}. \bibinfo{publisher}{Association for
  Computing Machinery}, \bibinfo{address}{New York, NY, USA},
  \bibinfo{pages}{75–86}.
\newblock
\showISBNx{9781595939197}
\urldef\tempurl%
\url{https://doi.org/10.1145/1411204.1411218}
\showDOI{\tempurl}


\bibitem[\protect\citeauthoryear{Hofmann}{Hofmann}{1996}]%
        {10.1007/3-540-61780-9_68}
\bibfield{author}{\bibinfo{person}{Martin Hofmann}.}
  \bibinfo{year}{1996}\natexlab{}.
\newblock \showarticletitle{Conservativity of equality reflection over
  intensional type theory}. In \bibinfo{booktitle}{\emph{Types for Proofs and
  Programs}}, \bibfield{editor}{\bibinfo{person}{Stefano Berardi} {and}
  \bibinfo{person}{Mario Coppo}} (Eds.). \bibinfo{publisher}{Springer Berlin
  Heidelberg}, \bibinfo{address}{Berlin, Heidelberg},
  \bibinfo{pages}{153--164}.
\newblock
\showISBNx{978-3-540-70722-6}


\bibitem[\protect\citeauthoryear{Knowles and Flanagan}{Knowles and
  Flanagan}{2010}]%
        {Hybrid}
\bibfield{author}{\bibinfo{person}{Kenneth Knowles} {and}
  \bibinfo{person}{Cormac Flanagan}.} \bibinfo{year}{2010}\natexlab{}.
\newblock \showarticletitle{Hybrid Type Checking}.
\newblock \bibinfo{journal}{\emph{ACM Trans. Program. Lang. Syst.}}
  \bibinfo{volume}{32}, \bibinfo{number}{2}, Article \bibinfo{articleno}{6}
  (\bibinfo{date}{Feb.} \bibinfo{year}{2010}), \bibinfo{numpages}{34}~pages.
\newblock
\showISSN{0164-0925}
\urldef\tempurl%
\url{https://doi.org/10.1145/1667048.1667051}
\showDOI{\tempurl}


\bibitem[\protect\citeauthoryear{Knowles, Tomb, Gronski, Freund, and
  Flanagan}{Knowles et~al\mbox{.}}{2006}]%
        {Knowles06sage}
\bibfield{author}{\bibinfo{person}{Kenneth Knowles}, \bibinfo{person}{Aaron
  Tomb}, \bibinfo{person}{Jessica Gronski}, \bibinfo{person}{Stephen~N.
  Freund}, {and} \bibinfo{person}{Cormac Flanagan}.}
  \bibinfo{year}{2006}\natexlab{}.
\newblock \showarticletitle{Sage: Hybrid checking for flexible specifications}.
  In \bibinfo{booktitle}{\emph{Scheme and {F}unctional {P}rogramming
  {W}orkshop}}.
\newblock


\bibitem[\protect\citeauthoryear{Leino}{Leino}{2012}]%
        {DBLP:conf/sigada/Leino12}
\bibfield{author}{\bibinfo{person}{K.~Rustan~M. Leino}.}
  \bibinfo{year}{2012}\natexlab{}.
\newblock \showarticletitle{Developing verified programs with Dafny}. In
  \bibinfo{booktitle}{\emph{Proceedings of the 2012 {ACM} Conference on High
  Integrity Language Technology, {HILT} '12, December 2-6, 2012, Boston,
  Massachusetts, {USA}}}, \bibfield{editor}{\bibinfo{person}{Ben Brosgol},
  \bibinfo{person}{Jeff Boleng}, {and} \bibinfo{person}{S.~Tucker Taft}}
  (Eds.). \bibinfo{publisher}{{ACM}}, \bibinfo{pages}{9--10}.
\newblock
\urldef\tempurl%
\url{https://doi.org/10.1145/2402676.2402682}
\showDOI{\tempurl}


\bibitem[\protect\citeauthoryear{Liu, Parker, Redmond, Kuper, Hicks, and
  Vazou}{Liu et~al\mbox{.}}{2020}]%
        {liu20typeclasses}
\bibfield{author}{\bibinfo{person}{Yiyun Liu}, \bibinfo{person}{James Parker},
  \bibinfo{person}{Patrick Redmond}, \bibinfo{person}{Lindsey Kuper},
  \bibinfo{person}{Michael Hicks}, {and} \bibinfo{person}{Niki Vazou}.}
  \bibinfo{year}{2020}\natexlab{}.
\newblock \showarticletitle{Verifying Replicated Data Types with Typeclass
  Refinements in Liquid Haskell}. In \bibinfo{booktitle}{\emph{Proceedings of
  the {ACM} Conference on Object-Oriented Programming Languages, Systems, and
  Applications (OOPSLA)}}.
\newblock


\bibitem[\protect\citeauthoryear{Martin-L{\"o}f}{Martin-L{\"o}f}{1975}]%
        {MARTINLOF197573}
\bibfield{author}{\bibinfo{person}{Per Martin-L{\"o}f}.}
  \bibinfo{year}{1975}\natexlab{}.
\newblock \showarticletitle{An Intuitionistic Theory of Types: Predicative
  Part}.
\newblock In \bibinfo{booktitle}{\emph{Logic Colloquium '73}},
  \bibfield{editor}{\bibinfo{person}{H.E. Rose} {and} \bibinfo{person}{J.C.
  Shepherdson}} (Eds.). \bibinfo{series}{Studies in Logic and the Foundations
  of Mathematics}, Vol.~\bibinfo{volume}{80}. \bibinfo{publisher}{Elsevier},
  \bibinfo{pages}{73 -- 118}.
\newblock
\showISSN{0049-237X}
\urldef\tempurl%
\url{https://doi.org/10.1016/S0049-237X(08)71945-1}
\showDOI{\tempurl}


\bibitem[\protect\citeauthoryear{Martin-L{\"o}f}{Martin-L{\"o}f}{1984}]%
        {martin1984intuitionistic}
\bibfield{author}{\bibinfo{person}{Per Martin-L{\"o}f}.}
  \bibinfo{year}{1984}\natexlab{}.
\newblock \bibinfo{booktitle}{\emph{Intuitionistic Type Theory}}.
\newblock \bibinfo{publisher}{Bibliopolis}.
\newblock
\urldef\tempurl%
\url{https://books.google.com/books?id=\_D0ZAQAAIAAJ}
\showURL{%
\tempurl}
\newblock
\shownote{As recorded by Giovanni Sambin.}


\bibitem[\protect\citeauthoryear{Nipkow and Prehofer}{Nipkow and
  Prehofer}{1993}]%
        {10.1145/158511.158698}
\bibfield{author}{\bibinfo{person}{Tobias Nipkow} {and}
  \bibinfo{person}{Christian Prehofer}.} \bibinfo{year}{1993}\natexlab{}.
\newblock \showarticletitle{Type Checking Type Classes}. In
  \bibinfo{booktitle}{\emph{Proceedings of the 20th ACM SIGPLAN-SIGACT
  Symposium on Principles of Programming Languages}} (Charleston, South
  Carolina, USA) \emph{(\bibinfo{series}{POPL ’93})}.
  \bibinfo{publisher}{Association for Computing Machinery},
  \bibinfo{address}{New York, NY, USA}, \bibinfo{pages}{409–418}.
\newblock
\showISBNx{0897915607}
\urldef\tempurl%
\url{https://doi.org/10.1145/158511.158698}
\showDOI{\tempurl}


\bibitem[\protect\citeauthoryear{Nipkow and Snelting}{Nipkow and
  Snelting}{1991}]%
        {10.5555/645420.652540}
\bibfield{author}{\bibinfo{person}{Tobias Nipkow} {and} \bibinfo{person}{Gregor
  Snelting}.} \bibinfo{year}{1991}\natexlab{}.
\newblock \showarticletitle{Type Classes and Overloading Resolution via
  Order-Sorted Unification}. In \bibinfo{booktitle}{\emph{Proceedings of the
  5th ACM Conference on Functional Programming Languages and Computer
  Architecture}}. \bibinfo{publisher}{Springer-Verlag},
  \bibinfo{address}{Berlin, Heidelberg}, \bibinfo{pages}{1–14}.
\newblock
\showISBNx{3540543961}


\bibitem[\protect\citeauthoryear{Norell}{Norell}{2008}]%
        {10.5555/1813347.1813352}
\bibfield{author}{\bibinfo{person}{Ulf Norell}.}
  \bibinfo{year}{2008}\natexlab{}.
\newblock \showarticletitle{Dependently Typed Programming in Agda}. In
  \bibinfo{booktitle}{\emph{Proceedings of the 6th International Conference on
  Advanced Functional Programming}} (Heijen, The Netherlands)
  \emph{(\bibinfo{series}{AFP’08})}. \bibinfo{publisher}{Springer-Verlag},
  \bibinfo{address}{Berlin, Heidelberg}, \bibinfo{pages}{230–266}.
\newblock
\showISBNx{3642046517}


\bibitem[\protect\citeauthoryear{Ou, Tan, Mandelbaum, and Walker}{Ou
  et~al\mbox{.}}{2004}]%
        {DBLP:conf/ifipTCS/OuTMW04}
\bibfield{author}{\bibinfo{person}{Xinming Ou}, \bibinfo{person}{Gang Tan},
  \bibinfo{person}{Yitzhak Mandelbaum}, {and} \bibinfo{person}{David Walker}.}
  \bibinfo{year}{2004}\natexlab{}.
\newblock \showarticletitle{Dynamic Typing with Dependent Types}. In
  \bibinfo{booktitle}{\emph{Exploring New Frontiers of Theoretical Informatics,
  {IFIP} 18th World Computer Congress, {TC1} 3rd International Conference on
  Theoretical Computer Science (TCS2004), 22-27 August 2004, Toulouse, France}}
  \emph{(\bibinfo{series}{{IFIP}}, Vol.~\bibinfo{volume}{155})},
  \bibfield{editor}{\bibinfo{person}{Jean{-}Jacques L{\'{e}}vy},
  \bibinfo{person}{Ernst~W. Mayr}, {and} \bibinfo{person}{John~C. Mitchell}}
  (Eds.). \bibinfo{publisher}{Kluwer/Springer}, \bibinfo{pages}{437--450}.
\newblock
\urldef\tempurl%
\url{https://doi.org/10.1007/1-4020-8141-3\_34}
\showDOI{\tempurl}


\bibitem[\protect\citeauthoryear{{Pfenning}}{{Pfenning}}{2001}]%
        {932499}
\bibfield{author}{\bibinfo{person}{F. {Pfenning}}.}
  \bibinfo{year}{2001}\natexlab{}.
\newblock \showarticletitle{Intensionality, extensionality, and proof
  irrelevance in modal type theory}. In \bibinfo{booktitle}{\emph{Proceedings
  16th Annual IEEE Symposium on Logic in Computer Science}}.
  \bibinfo{pages}{221--230}.
\newblock


\bibitem[\protect\citeauthoryear{Rondon, Kawaguci, and Jhala}{Rondon
  et~al\mbox{.}}{2008}]%
        {LT2008}
\bibfield{author}{\bibinfo{person}{Patrick~M. Rondon}, \bibinfo{person}{Ming
  Kawaguci}, {and} \bibinfo{person}{Ranjit Jhala}.}
  \bibinfo{year}{2008}\natexlab{}.
\newblock \showarticletitle{Liquid Types}. In
  \bibinfo{booktitle}{\emph{Proceedings of the 29th ACM SIGPLAN Conference on
  Programming Language Design and Implementation}} (Tucson, AZ, USA)
  \emph{(\bibinfo{series}{PLDI '08})}. \bibinfo{publisher}{ACM},
  \bibinfo{address}{New York, NY, USA}, \bibinfo{pages}{159--169}.
\newblock
\showISBNx{978-1-59593-860-2}
\urldef\tempurl%
\url{https://doi.org/10.1145/1375581.1375602}
\showDOI{\tempurl}


\bibitem[\protect\citeauthoryear{Rushby, Owre, and Shankar}{Rushby
  et~al\mbox{.}}{1998}]%
        {rushby1998subtypes}
\bibfield{author}{\bibinfo{person}{John Rushby}, \bibinfo{person}{Sam Owre},
  {and} \bibinfo{person}{Natarajan Shankar}.} \bibinfo{year}{1998}\natexlab{}.
\newblock \showarticletitle{Subtypes for specifications: Predicate subtyping in
  PVS}.
\newblock \bibinfo{journal}{\emph{IEEE Transactions on Software Engineering}}
  \bibinfo{volume}{24}, \bibinfo{number}{9} (\bibinfo{year}{1998}),
  \bibinfo{pages}{709--720}.
\newblock


\bibitem[\protect\citeauthoryear{Sekiyama, Igarashi, and Greenberg}{Sekiyama
  et~al\mbox{.}}{2017}]%
        {toplas17}
\bibfield{author}{\bibinfo{person}{Taro Sekiyama}, \bibinfo{person}{Atsushi
  Igarashi}, {and} \bibinfo{person}{Michael Greenberg}.}
  \bibinfo{year}{2017}\natexlab{}.
\newblock \showarticletitle{Polymorphic Manifest Contracts, Revised and
  Resolved}.
\newblock \bibinfo{journal}{\emph{ACM Trans. Program. Lang. Syst.}}
  \bibinfo{volume}{39}, \bibinfo{number}{1}, Article \bibinfo{articleno}{3}
  (\bibinfo{date}{Feb.} \bibinfo{year}{2017}), \bibinfo{numpages}{36}~pages.
\newblock
\showISSN{0164-0925}
\urldef\tempurl%
\url{https://doi.org/10.1145/2994594}
\showDOI{\tempurl}


\bibitem[\protect\citeauthoryear{Sj\"{o}berg and Weirich}{Sj\"{o}berg and
  Weirich}{2015}]%
        {10.1145/2676726.2676974}
\bibfield{author}{\bibinfo{person}{Vilhelm Sj\"{o}berg} {and}
  \bibinfo{person}{Stephanie Weirich}.} \bibinfo{year}{2015}\natexlab{}.
\newblock \showarticletitle{Programming up to Congruence}. In
  \bibinfo{booktitle}{\emph{Proceedings of the 42nd Annual ACM SIGPLAN-SIGACT
  Symposium on Principles of Programming Languages}} (Mumbai, India)
  \emph{(\bibinfo{series}{POPL ’15})}. \bibinfo{publisher}{Association for
  Computing Machinery}, \bibinfo{address}{New York, NY, USA},
  \bibinfo{pages}{369–382}.
\newblock
\showISBNx{9781450333009}
\urldef\tempurl%
\url{https://doi.org/10.1145/2676726.2676974}
\showDOI{\tempurl}


\bibitem[\protect\citeauthoryear{Sozeau}{Sozeau}{2008}]%
        {sozeau2008environnement}
\bibfield{author}{\bibinfo{person}{Matthieu Sozeau}.}
  \bibinfo{year}{2008}\natexlab{}.
\newblock \emph{\bibinfo{title}{Un environnement pour la programmation avec
  types d{\'e}pendants}}.
\newblock \bibinfo{thesistype}{Ph.D. Dissertation}.
  \bibinfo{school}{Universit{\'e} Paris 11}, \bibinfo{address}{Orsay, France}.
\newblock


\bibitem[\protect\citeauthoryear{Sozeau and Oury}{Sozeau and Oury}{2008}]%
        {10.1007/978-3-540-71067-7_23}
\bibfield{author}{\bibinfo{person}{Matthieu Sozeau} {and}
  \bibinfo{person}{Nicolas Oury}.} \bibinfo{year}{2008}\natexlab{}.
\newblock \showarticletitle{First-Class Type Classes}. In
  \bibinfo{booktitle}{\emph{Theorem Proving in Higher Order Logics}},
  \bibfield{editor}{\bibinfo{person}{Otmane~Ait Mohamed},
  \bibinfo{person}{C{\'e}sar Mu{\~{n}}oz}, {and} \bibinfo{person}{Sofi{\`e}ne
  Tahar}} (Eds.). \bibinfo{publisher}{Springer Berlin Heidelberg},
  \bibinfo{address}{Berlin, Heidelberg}, \bibinfo{pages}{278--293}.
\newblock
\showISBNx{978-3-540-71067-7}


\bibitem[\protect\citeauthoryear{Sterling, Angiuli, and Gratzer}{Sterling
  et~al\mbox{.}}{2019}]%
        {DBLP:journals/corr/abs-1904-08562}
\bibfield{author}{\bibinfo{person}{Jonathan Sterling}, \bibinfo{person}{Carlo
  Angiuli}, {and} \bibinfo{person}{Daniel Gratzer}.}
  \bibinfo{year}{2019}\natexlab{}.
\newblock \showarticletitle{Cubical Syntax for Reflection-Free Extensional
  Equality}.
\newblock \bibinfo{journal}{\emph{CoRR}}  \bibinfo{volume}{abs/1904.08562}
  (\bibinfo{year}{2019}).
\newblock
\showeprint[arxiv]{1904.08562}
\urldef\tempurl%
\url{http://arxiv.org/abs/1904.08562}
\showURL{%
\tempurl}


\bibitem[\protect\citeauthoryear{{supplementary material}}{{supplementary
  material}}{2021}]%
        {implementation}
\bibfield{author}{\bibinfo{person}{{supplementary material}}.}
  \bibinfo{year}{2021}\natexlab{}.
\newblock \bibinfo{title}{Supplementary\ Material for Functional Extensionality
  for Refinement Types}.
\newblock
\newblock


\bibitem[\protect\citeauthoryear{Swamy, Hritcu, Keller, Rastogi,
  Delignat-Lavaud, Forest, Bhargavan, Fournet, Strub, Kohlweiss,
  Zinzindohou\'e, and {Zanella-B\'eguelin}}{Swamy et~al\mbox{.}}{2016}]%
        {mumon}
\bibfield{author}{\bibinfo{person}{Nikhil Swamy}, \bibinfo{person}{Catalin
  Hritcu}, \bibinfo{person}{Chantal Keller}, \bibinfo{person}{Aseem Rastogi},
  \bibinfo{person}{Antoine Delignat-Lavaud}, \bibinfo{person}{Simon Forest},
  \bibinfo{person}{Karthikeyan Bhargavan}, \bibinfo{person}{C\'{e}dric
  Fournet}, \bibinfo{person}{Pierre-Yves Strub}, \bibinfo{person}{Markulf
  Kohlweiss}, \bibinfo{person}{Jean-Karim Zinzindohou\'e}, {and}
  \bibinfo{person}{Santiago {Zanella-B\'eguelin}}.}
  \bibinfo{year}{2016}\natexlab{}.
\newblock \showarticletitle{Dependent Types and Multi-Monadic Effects in {F*}}.
  In \bibinfo{booktitle}{\emph{43rd ACM SIGPLAN-SIGACT Symposium on Principles
  of Programming Languages (POPL)}}. \bibinfo{publisher}{ACM},
  \bibinfo{pages}{256--270}.
\newblock
\showISBNx{978-1-4503-3549-2}
\urldef\tempurl%
\url{https://www.fstar-lang.org/papers/mumon/}
\showURL{%
\tempurl}


\bibitem[\protect\citeauthoryear{{Tabareau}, {Tanter}, and {Sozeau}}{{Tabareau}
  et~al\mbox{.}}{2019}]%
        {2019arXiv190905027T}
\bibfield{author}{\bibinfo{person}{Nicolas {Tabareau}},
  \bibinfo{person}{{\'E}ric {Tanter}}, {and} \bibinfo{person}{Matthieu
  {Sozeau}}.} \bibinfo{year}{2019}\natexlab{}.
\newblock \showarticletitle{{The Marriage of Univalence and Parametricity}}.
\newblock \bibinfo{journal}{\emph{arXiv e-prints}}, Article
  \bibinfo{articleno}{arXiv:1909.05027} (\bibinfo{date}{Sept.}
  \bibinfo{year}{2019}), \bibinfo{numpages}{arXiv:1909.05027}~pages.
\newblock
\showeprint[arxiv]{1909.05027}~[cs.PL]


\bibitem[\protect\citeauthoryear{Takahashi}{Takahashi}{1989}]%
        {DBLP:journals/jsc/Takahashi89}
\bibfield{author}{\bibinfo{person}{Masako Takahashi}.}
  \bibinfo{year}{1989}\natexlab{}.
\newblock \showarticletitle{Parallel Reductions in lambda-Calculus}.
\newblock \bibinfo{journal}{\emph{J. Symb. Comput.}} \bibinfo{volume}{7},
  \bibinfo{number}{2} (\bibinfo{year}{1989}), \bibinfo{pages}{113--123}.
\newblock
\urldef\tempurl%
\url{https://doi.org/10.1016/S0747-7171(89)80045-8}
\showDOI{\tempurl}


\bibitem[\protect\citeauthoryear{Team}{Team}{2020}]%
        {the_coq_development_team_2020_3744225}
\bibfield{author}{\bibinfo{person}{The Coq~Development Team}.}
  \bibinfo{year}{2020}\natexlab{}.
\newblock \bibinfo{booktitle}{\emph{The Coq Proof Assistant, version 8.11.0}}.
\newblock
\urldef\tempurl%
\url{https://doi.org/10.5281/zenodo.3744225}
\showDOI{\tempurl}


\bibitem[\protect\citeauthoryear{Vazou, Breitner, Kunkel, Van~Horn, and
  Hutton}{Vazou et~al\mbox{.}}{2018a}]%
        {TPE2018}
\bibfield{author}{\bibinfo{person}{Niki Vazou}, \bibinfo{person}{Joachim
  Breitner}, \bibinfo{person}{Rose Kunkel}, \bibinfo{person}{David Van~Horn},
  {and} \bibinfo{person}{Graham Hutton}.} \bibinfo{year}{2018}\natexlab{a}.
\newblock \showarticletitle{Theorem Proving for All: Equational Reasoning in
  Liquid Haskell (Functional Pearl)}. In \bibinfo{booktitle}{\emph{Proceedings
  of the 11th ACM SIGPLAN International Symposium on Haskell}} (St. Louis, MO,
  USA) \emph{(\bibinfo{series}{Haskell 2018})}. \bibinfo{publisher}{ACM},
  \bibinfo{address}{New York, NY, USA}, \bibinfo{pages}{132--144}.
\newblock
\showISBNx{978-1-4503-5835-4}
\urldef\tempurl%
\url{https://doi.org/10.1145/3242744.3242756}
\showDOI{\tempurl}


\bibitem[\protect\citeauthoryear{Vazou, Tondwalkar, Choudhury, Scott, Newton,
  Wadler, and Jhala}{Vazou et~al\mbox{.}}{2018b}]%
        {VazouTCSNWJ18}
\bibfield{author}{\bibinfo{person}{Niki Vazou}, \bibinfo{person}{Anish
  Tondwalkar}, \bibinfo{person}{Vikraman Choudhury}, \bibinfo{person}{Ryan~G.
  Scott}, \bibinfo{person}{Ryan~R. Newton}, \bibinfo{person}{Philip Wadler},
  {and} \bibinfo{person}{Ranjit Jhala}.} \bibinfo{year}{2018}\natexlab{b}.
\newblock \showarticletitle{Refinement reflection: complete verification with
  {SMT}}.
\newblock \bibinfo{journal}{\emph{{PACMPL}}} \bibinfo{volume}{2},
  \bibinfo{number}{{POPL}} (\bibinfo{year}{2018}),
  \bibinfo{pages}{53:1--53:31}.
\newblock
\urldef\tempurl%
\url{https://doi.org/10.1145/3158141}
\showDOI{\tempurl}


\bibitem[\protect\citeauthoryear{von Glehn}{von Glehn}{2014}]%
        {vonGlehn}
\bibfield{author}{\bibinfo{person}{Tamara von Glehn}.}
  \bibinfo{year}{2014}\natexlab{}.
\newblock \emph{\bibinfo{title}{Polynomials and Models of Type Theory}}.
\newblock \bibinfo{thesistype}{Ph.D. Dissertation}. \bibinfo{school}{Magdalene
  College, University of Cambridge}.
\newblock


\bibitem[\protect\citeauthoryear{Wadler and Blott}{Wadler and Blott}{1989}]%
        {10.1145/75277.75283}
\bibfield{author}{\bibinfo{person}{P. Wadler} {and} \bibinfo{person}{S.
  Blott}.} \bibinfo{year}{1989}\natexlab{}.
\newblock \showarticletitle{How to Make Ad-Hoc Polymorphism Less Ad Hoc}. In
  \bibinfo{booktitle}{\emph{Proceedings of the 16th ACM SIGPLAN-SIGACT
  Symposium on Principles of Programming Languages}} (Austin, Texas, USA)
  \emph{(\bibinfo{series}{POPL ’89})}. \bibinfo{publisher}{Association for
  Computing Machinery}, \bibinfo{address}{New York, NY, USA},
  \bibinfo{pages}{60–76}.
\newblock
\showISBNx{0897912942}
\urldef\tempurl%
\url{https://doi.org/10.1145/75277.75283}
\showDOI{\tempurl}


\bibitem[\protect\citeauthoryear{Weirich}{Weirich}{2017}]%
        {10.1145/3009837.3009923}
\bibfield{author}{\bibinfo{person}{Stephanie Weirich}.}
  \bibinfo{year}{2017}\natexlab{}.
\newblock \showarticletitle{The Influence of Dependent Types (Keynote)}. In
  \bibinfo{booktitle}{\emph{Proceedings of the 44th ACM SIGPLAN Symposium on
  Principles of Programming Languages}} (Paris, France)
  \emph{(\bibinfo{series}{POPL 2017})}. \bibinfo{publisher}{Association for
  Computing Machinery}, \bibinfo{address}{New York, NY, USA},
  \bibinfo{pages}{1}.
\newblock
\showISBNx{9781450346603}
\urldef\tempurl%
\url{https://doi.org/10.1145/3009837.3009923}
\showDOI{\tempurl}


\bibitem[\protect\citeauthoryear{Wenzel}{Wenzel}{1997}]%
        {10.1007/BFb0028402}
\bibfield{author}{\bibinfo{person}{Markus Wenzel}.}
  \bibinfo{year}{1997}\natexlab{}.
\newblock \showarticletitle{Type classes and overloading in higher-order
  logic}. In \bibinfo{booktitle}{\emph{Theorem Proving in Higher Order
  Logics}}, \bibfield{editor}{\bibinfo{person}{Elsa~L. Gunter} {and}
  \bibinfo{person}{Amy Felty}} (Eds.). \bibinfo{publisher}{Springer Berlin
  Heidelberg}, \bibinfo{address}{Berlin, Heidelberg},
  \bibinfo{pages}{307--322}.
\newblock
\showISBNx{978-3-540-69526-4}


\bibitem[\protect\citeauthoryear{Wright and Felleisen}{Wright and
  Felleisen}{1994}]%
        {Wright94syntactic}
\bibfield{author}{\bibinfo{person}{Andrew~K. Wright} {and}
  \bibinfo{person}{Matthias Felleisen}.} \bibinfo{year}{1994}\natexlab{}.
\newblock \showarticletitle{A Syntactic Approach to Type Soundness}.
\newblock \bibinfo{journal}{\emph{Information and {C}omputation}}
  \bibinfo{volume}{115} (\bibinfo{year}{1994}), \bibinfo{pages}{38--94}.
\newblock
Issue 1.


\bibitem[\protect\citeauthoryear{Xi, Chen, and Chen}{Xi et~al\mbox{.}}{2003}]%
        {10.1145/604131.604150}
\bibfield{author}{\bibinfo{person}{Hongwei Xi}, \bibinfo{person}{Chiyan Chen},
  {and} \bibinfo{person}{Gang Chen}.} \bibinfo{year}{2003}\natexlab{}.
\newblock \showarticletitle{Guarded Recursive Datatype Constructors}. In
  \bibinfo{booktitle}{\emph{Proceedings of the 30th ACM SIGPLAN-SIGACT
  Symposium on Principles of Programming Languages}} (New Orleans, Louisiana,
  USA) \emph{(\bibinfo{series}{POPL '03})}. \bibinfo{publisher}{Association for
  Computing Machinery}, \bibinfo{address}{New York, NY, USA},
  \bibinfo{pages}{224–235}.
\newblock
\showISBNx{1581136285}
\urldef\tempurl%
\url{https://doi.org/10.1145/604131.604150}
\showDOI{\tempurl}


\bibitem[\protect\citeauthoryear{Xi and Pfenning}{Xi and Pfenning}{1998}]%
        {xi1998eliminating}
\bibfield{author}{\bibinfo{person}{Hongwei Xi} {and} \bibinfo{person}{Frank
  Pfenning}.} \bibinfo{year}{1998}\natexlab{}.
\newblock \showarticletitle{Eliminating array bound checking through dependent
  types}. In \bibinfo{booktitle}{\emph{Proceedings of the ACM SIGPLAN 1998
  conference on Programming language design and implementation}}.
  \bibinfo{pages}{249--257}.
\newblock


\bibitem[\protect\citeauthoryear{Zalewski, McKinna, Morris, and
  Wadler}{Zalewski et~al\mbox{.}}{2020}]%
        {Zalewski20WGT}
\bibfield{author}{\bibinfo{person}{Jakub Zalewski}, \bibinfo{person}{James
  McKinna}, \bibinfo{person}{J.~Garrett Morris}, {and} \bibinfo{person}{Philip
  Wadler}.} \bibinfo{year}{2020}\natexlab{}.
\newblock \showarticletitle{Blame tracking at higher fidelity}. In
  \bibinfo{booktitle}{\emph{{Workshop on Gradual Typing (WGT)}}}.
\newblock


\end{thebibliography}

\appendix
\begin{landscape}
\section{Complete Type Checking of Extensionality Example}
\label{sec:ext-complete-typing}
\begin{figure}[h!]
\begin{adjustbox}{width=1.5\textwidth}
$$
\inference{
    \inference{
        \inference{
            \inference{
                \inference{
                   \inference{
                        \inference{
                            \color{red}\envLookUp{\env}{\extName}{\typExt}
                        }{
                        \color{orange}\hastype{\env}{\extName}{\typExt}
                    }}{
                        \color{blue}\hastype{\env}{\extName\ \tyApp{\tya}}{\typExtA{\tya}}
                    }
                }{
                    \color{green}\hastype{\env}{\extName\ \tyApp{\tya}\ \tyApp{\tyb}}{\typExtB{\tya}{\tyb}}
                }
                \\ 
                {\inference{
                    \color{red}\envLookUp{\env}{\dName}{\dTy}
                }{
                    \color{green}\hastype{\env}{\dName}{\dTy}
                }} &&
                {\inference{
                }{
                    \color{green}\issubtype{\env}{\dTy}{\typExtDict{\tya}{\tyb}}
                }[\subd]}
                }{\color{olive}\hastype{\env}{\extName\ \tyApp{\tya}\ \tyApp{\tyb}\ \dName }{\typExtC{\tya}{\tyb}}
            }
            \\ 
            {\inference{
                \color{red}\envLookUp{\env}{\hName}{\hTy}
            }{
                \color{olive}\hastype{\env}{\hName}{\hTy}
            }} &&
            {\inference{
                \color{red}\dots
            }{
                \color{olive}\issubtype{\env}{\hTy}{\typExtF{\tya}{\tyb}}
            }[\subh]}
            }{
        \color{teal}\hastype{\env}{\extName\ \tyApp{\tya}\ \tyApp{\tyb}\ \dName \ \hName }{\typExtD{\tya}{\tyb}{\hName}}
        }\\
        {\inference{
            \color{red}\envLookUp{\env}{\kName}{\kTy}
        }{
            \color{teal}\hastype{\env}{\kName}{\kTy}
        }} &&
        {\inference{
            \color{red}\dots
        }{
            \color{teal}\issubtype{\env}{\kTy}{\typExtG{\tya}{\tyb}}
        }[\subk]}
        }{
    \color{blue}\hastype{\env}{\extName\ \tyApp{\tya}\ \tyApp{\tyb}\ \dName \ \hName \ \kName }{\typExtE{\tya}{tyb}{\hName}{\kName}}
    }
    \\
    {\inference{
    \color{red}\envLookUp{\env}{\lemmaName}{\lemmaTy}
    }{
    \color{blue}\hastype{\env}{\lemmaName}{\lemmaTy}
    }} &&
    {\inference{
    \color{red}\dots
    }{
    \color{blue}\issubtype{\env}{\lemmaTy}{\typeExtProp{\tya}{\tyb}{\hName}{\kName}}
    }[\sublemma]}
}{
    \color{violet}\hastype{\env}{\extName\ \tyApp{\tya}\ \tyApp{\tyb}\ \dName \ \hName \ \kName \ \lemmaName }{\propExpr{\fEq{h}{k}}}
}
$$
\end{adjustbox}

    $$
    \inference{
        \inference{
            \apred \Rightarrow \hdom
        }{
        \issubtype{\env}{\tya}{\{v:\alpha \mid \hdom\}}
         } && 
         \inference{
            \apred \Rightarrow \hrng \Rightarrow \bpred
         }{
         \issubtype{\env,\envBind{x}{\tya}}{\{v:\beta \mid \hrng\}}{\tyb}
         }
    }{
      \issubtype{\env}{\hTy}{\typExtF{\tya}{\tyb}}
    }[\subh]
    $$

    $$
    \inference{
        \inference{
            \apred \Rightarrow \kdom
        }{
        \issubtype{\env}{\tya}{\{v:\alpha \mid \kdom\}}
         } && 
         \inference{
            \apred \Rightarrow \krng \Rightarrow \bpred
         }{
         \issubtype{\env,\envBind{x}{\tya}}{\{v:\beta \mid \krng\}}{\tyb}
         }
    }{
      \issubtype{\env}{\kTy}{\typExtF{\tya}{\tyb}}
    }[\subk]
    $$

    $$
    \inference{
        \inference{
            \apred \Rightarrow \texttt{true}
        }{
        \issubtype{\env}{\tya}{\alpha}
         } && 
         \inference{
            \apred \Rightarrow p \Rightarrow \bEq{h\ x}{k\ x}
         }{
         \issubtype{\env,\envBind{x}{\tya}}{\propExpr{p}}{\propExpr{\bEq{h\ x}{k\ x}}}
         }
    }{
      \issubtype{\env}{\lemmaTy}{\typeExtProp{\tya}{\tyb}{\hName}{\kName}}
    }[\sublemma]
    $$
    
    \caption{Complete type checking of naive extensionality in \texttt{theoremEq}.}
    \label{fig:extensionality-checking-full}
    \end{figure}    
\end{landscape}

\section{Proofs and Definitions for Metatheory}
\label{sec:proofs-metatheory}
In this section we provide proofs and definitions ommitted from \S\ref{sec:eqrt}.

\subsection{Base Type Checking}
\label{subsec:appendix:base-type-checking}

For completeness, we defined \corelanueq, the unrefined version of \corelaneq, 
that ignores the refinements 
on basic types and the expression indices from the typed equality. 

The function \unrefine{\cdot} is defined to turn \corelaneq types to their unrefined 
counterparts. 
$$
\begin{array}{rcl}
  \unrefine{\tbool} & \doteq & \tbool\\
  \unrefine{\tunit} & \doteq & \tunit\\
  \unrefine{\eqrt{\typ}{\expr_1}{\expr_2}} & \doteq & \eqt{\unrefine{\typ}}{}{} \\[2mm]
  \unrefine{\tref{v}{\tbase}{\refa}} & \doteq & \tbase \\ 
  \unrefine{\tfun{x}{\typ_x}{\typ}} & \doteq & \tfuns{\unrefine{\typ_x}}{\unrefine{\typ}}\\  
\end{array}
$$

Figure~\ref{fig:basictyping} defines the syntax and typing of \corelanueq
that we use to define type denotations of \corelaneq.

\begin{figure}

$$
\begin{array}{rrcl}
    \textit{Expressions} & \expr & ::= & \text{as in } \corelaneq \\[1mm]
    \textit{Types} & \utyp & ::= & \tbool \mid \tunit \mid \eqt{\utyp}{\expr}{\expr} \mid  \tfuns{\utyp}{\utyp} \\[1mm]
    \textit{Typing Environment} & \uenv & ::= & \emptyset \mid \uenv, \envBind{x}{\utyp} \\[1mm]
\end{array}
$$

\textit{Basic Type checking}\hfill\fbox{\hasbtype{\uenv}{\expr}{\utyp}}
$$
\inference{
}{
\hasbtype{\uenv}{\con}{\unrefine{\tycon{\con}}}
}[\tbcon]
\quad
\inference{
\envBind{x}{\utyp} \in \uenv
}{
\hasbtype{\uenv}{x}{\utyp}
}[\tbvar]
$$

$$
\inference{
  \hasbtype{\uenv}{\expr}{\tfuns{\utyp_x}{\utyp}} && 
  \hasbtype{\uenv}{\expr_x}{\utyp_x}  
}{
\hasbtype{\uenv}{\expr\ \expr_x}{\utyp}
}[\tbapp]
\quad
\inference{
\hasbtype{\uenv,\envBind{x}{\unrefine{\typ_x}}}{\expr}{\utyp}
}{
\hasbtype{\uenv}{\elam{x}{\typ_x}{\expr}}{\tfuns{\unrefine{\typ_x}}{\utyp}}
}[\tblam]
$$

$$
\inference{
  \hasbtype{\uenv}{\expr}{\tunit} \\
  \hasbtype{\uenv}{\expr_1}{\tbase} && 
  \hasbtype{\uenv}{\expr_2}{\tbase} 
}{
  \hasbtype{\uenv}{\ebeq{\tbase}\ \expr_1\ \expr_2\ \expr}{\eqt{\tbase}{\expr_1}{\expr_2}}
}[\tbeqbase]
\quad
\inference{
  \hasbtype{\uenv}{\expr}{\tunit} \\
  \hasbtype{\uenv}{\expr_1}{\unrefine{\tfuns{\typ_x}{\typ}}} && 
  \hasbtype{\uenv}{\expr_2}{\unrefine{\tfuns{\typ_x}{\typ}}}  
}{
  \hasbtype{\uenv}{\exeq{x}{\typ_x}{\typ}\ \expr_1\ \expr_2\ \expr}{\eqt{\unrefine{\tfuns{\typ_x}{\typ}}}{\expr_1}{\expr_2}}
}[\tbeqfun]
$$
\caption{Syntax and Typing of \corelanueq.}
\label{fig:basictyping}
\end{figure}

\subsection{Constant Property}

\begin{theorem}\label{theorem:constant-property}
For the constants $ \con = \etrue, \efalse, \eunit,$ and $\bbbeqsym{\tbase}$, \iffull Property~\ref{property:constants} holds\fi, \ie $\con \in \interp{\tycon{\con}}$. 
\end{theorem}

\begin{proof}
Below are the proofs for each of the four constants.
\begin{itemize}
    \item $\expr \equiv \etrue$ and $\expr \in \interp{\tref{\rbind}{\tbool}{\bbbeq{\rbind}{\etrue}{\tbool}}}$. 
    We need to prove the below three requirements of membership in the interpretation of basic types: 
    \begin{itemize}
        \item $\goesto{\expr}{\val}$, which holds because $\etrue$ is a value, thus $\val = \etrue$;
        \item $\hasbtype{}{\expr}{\tbool}$, which holds by the typing rule \tbcon; and 
        \item $\goesto{(\bbbeq{\rbind}{\etrue}{\tbool})\subst{\rbind}{\expr}}{\etrue}$, which holds because 
            $$
            \begin{array}{rcl}
            (\bbbeq{\rbind}{\etrue}{\tbool})\subst{\rbind}{\expr} & = & \bbbeq{\etrue}{\etrue}{\tbool}\\
            & \evals{}{} & \bbbeq{(}{)}{(\etrue,\tbool)}\ \etrue \\
            & \evals{}{} & \etrue = \etrue  \\
            & =          & \etrue
            \end{array}
            $$
    \end{itemize}
    \item $\expr \equiv \efalse$ and $\expr \in \interp{\tref{\rbind}{\tbool}{\bbbeq{\rbind}{\efalse}{\tbool}}}$. 
    We need to prove the below three requirements of membership in the interpretation of basic types: 
    \begin{itemize}
        \item $\goesto{\expr}{\val}$, which holds because $\efalse$ is a value, thus $\val = \efalse$;
        \item $\hasbtype{}{\expr}{\tbool}$, which holds by the typing rule \tbcon; and 
        \item $\goesto{(\bbbeq{\rbind}{\efalse}{\tbool})\subst{\rbind}{\expr}}{\etrue}$, which holds because 
            $$
            \begin{array}{rcl}
            (\bbbeq{\rbind}{\efalse}{\tbool})\subst{\rbind}{\expr} & = & \bbbeq{\efalse}{\efalse}{\tbool}\\
            & \evals{}{} & \bbbeq{(}{)}{(\efalse,\tbool)}\ \efalse \\
            & \evals{}{} & \efalse = \efalse  \\
            & =          & \etrue
            \end{array}
            $$
    \end{itemize}
    \item $\expr \equiv \eunit$ and $\expr \in \interp{\tref{\rbind}{\tunit}{\bbbeq{\rbind}{\eunit}{\tunit}}}$. 
    We need to prove the below three requirements of membership in the interpretation of basic types: 
    \begin{itemize}
        \item $\goesto{\expr}{\val}$, which holds because $\eunit$ is a value, thus $\val = \eunit$;
        \item $\hasbtype{}{\expr}{\tunit}$, which holds by the typing rule \tbcon; and 
        \item $\goesto{(\bbbeq{\rbind}{\eunit}{\tunit})\subst{\rbind}{\expr}}{\etrue}$, which holds because 
            $$
            \begin{array}{rcl}
            (\bbbeq{\rbind}{\eunit}{\tunit})\subst{\rbind}{\expr} & = & \bbbeq{\eunit}{\eunit}{\tunit}\\
            & \evals{}{} & \bbbeq{(}{)}{(\eunit,\tunit)}\ \eunit \\
            & \evals{}{} & \eunit = \eunit  \\
            & =          & \etrue
            \end{array}
            $$
    \end{itemize}
    \item $\bbbeqsym{\tbase} \in \interp{\tfun{x}{\tbase}{\tfun{y}{\tbase}{\tref{\vv}{\tbool}{\bbbeq{\vv}{(\bbbeq{x}{y}{\tbase})}{\tbool}}}}}$. 
    By the definition of interpretation of function types, we fix 
    $\expr_x, \expr_y \in \interp{\tbase}$ and we need to prove that 
    $\expr \equiv \bbbeq{\expr_x}{\expr_y}{\tbase} \in 
    \interp{(\tref{\vv}{\tbool}{\bbbeq{\vv}{(\bbbeq{x}{y}{\tbase})}{\tbool}})\subst{x}{\expr_x}\subst{y}{\expr_y}}$. 
   We  prove the below three requirements of membership in the interpretation of basic types: 
    \begin{itemize}
        \item $\goesto{\expr}{\val}$, which holds because 
        $$
        \begin{array}{rcll}
            \expr & =           & \bbbeq{\expr_x}{\expr_y}{\tbase} &  \\
                  & \goesto{}{} & \bbbeq{\val_x}{\expr_y}{\tbase}  & \text{because } \expr_x \in \interp{\tbase}\\
                  & \goesto{}{} & \bbbeq{\val_x}{\val_y}{\tbase}   & \text{because } \expr_y \in \interp{\tbase}\\
                  & \evals{}{}  & \bbbeq{(}{)}{(\val_x,\tbase)}\ \val_y & \\
                  & \evals{}{}  & \val_x = \val_y  & \\
                  & =           & \val & \text{with } \val = \etrue \text{ or } \val = \efalse 
        \end{array}
        $$
        \item $\hasbtype{}{\expr}{\tbool}$, which holds by the typing rule \tbcon and because 
        $\expr_x, \expr_y \in \interp{\tbase}$ thus $\hasbtype{}{\expr_x}{\tbase}$ and $\hasbtype{}{\expr_y}{\tbase}$; and 
        \item $\goesto{(\bbbeq{\vv}{(\bbbeq{x}{y}{\tbase})}{\tbool})
                        \subst{\vv}{\expr}
                        \subst{x}{\expr_x}
                        \subst{y}{\expr_y}
                      }{\etrue}$. Since $\expr_x, \expr_y \in \interp{\tbase}$
                      both expressions evaluate to values, 
                      say $\goesto{\expr_x}{\val_x}$ and $\goesto{\expr_y}{\val_y}$
            which holds because 
            $$
            \begin{array}{rcll}
                (\bbbeq{\vv}{(\bbbeq{x}{y}{\tbase})}{\tbool})
                \subst{\vv}{\expr}
                \subst{x}{\expr_x}
                \subst{y}{\expr_y} & = & \bbbeq{\expr}{(\bbbeq{\expr_x}{\expr_y}{\tbase})}{\tbool} & \\
                & =           & \bbbeq{(\bbbeq{\expr_x}{\expr_y}{\tbase})}{(\bbbeq{\expr_x}{\expr_y}{\tbase})}{\tbool} & \\
                & \goesto{}{} & \bbbeq{(\bbbeq{\val_x}{\expr_y}{\tbase})}{(\bbbeq{\expr_x}{\expr_y}{\tbase})}{\tbool} & \text{since } \goesto{\expr_x}{\val_x}\\
                & \goesto{}{} & \bbbeq{(\bbbeq{\val_x}{\val_y}{\tbase})}{(\bbbeq{\expr_x}{\expr_y}{\tbase})}{\tbool} & \text{since } \goesto{\expr_y}{\val_y}\\
                & \evals{}{}  & \bbbeq{(\bbbeq{(}{)}{(\val_x,\tbase)}\ \val_y)}{(\bbbeq{\expr_x}{\expr_y}{\tbase})}{\tbool} & \\
                & \evals{}{}  & \bbbeq{(\val_x =  \val_y)}{(\bbbeq{\expr_x}{\expr_y}{\tbase})}{\tbool} & \\
                & \goesto{}{}  & \bbbeq{(\val_x =  \val_y)}{(\bbbeq{\val_x}{\expr_y}{\tbase})}{\tbool} &  \text{since } \goesto{\expr_x}{\val_x} \\
                & \goesto{}{}  & \bbbeq{(\val_x =  \val_y)}{(\bbbeq{\val_x}{\val_y}{\tbase})}{\tbool} &  \text{since } \goesto{\expr_y}{\val_y} \\
                & \evals{}{}  & \bbbeq{(\val_x =  \val_y)}{(\bbbeq{(}{)}{(\val_x,\tbase)}\ \val_y)}{\tbool} &   \\
                & \evals{}{} & \bbbeq{(\val_x =  \val_y)}{(\val_x =  \val_y)}{\tbool} & \\
                & \evals{}{} & \bbbeq{(\val_x =  \val_y)}{(\val_x =  \val_y)}{\tbool} & \\
                & \evals{}{} & (\bbbeq{(}{)}{((\val_x =  \val_y),\tbool)}\ (\val_x =  \val_y) & \\
                & \evals{}{} & (\val_x =  \val_y) = (\val_x =  \val_y)  & \\
                & = & \etrue  & \\
            \end{array}
            $$
        \end{itemize}    
    \end{itemize}    
\end{proof}

\subsection{Type Soundness}\label{sec:proofs-soundness}

\begin{theorem}[Semantic soundness]\label{proofs-thm:soundness}
  If \hastype{\env}{\expr}{\typ} then \hassemtype{\env}{\expr}{\typ}.
\begin{proof}
By induction on the typing derivation.
\begin{itemize}
  \item[\tsub]
  By inversion of the rule we have 
  \begin{enumerate}
    \item \hastype{\env}{\expr}{\typ'} 
    \item \issubtype{\env}{\typ'}{\typ}
  \end{enumerate} 
  By IH on (1) we have 
  \begin{enumerate}
    \setcounter{enumi}{2}
    \item \hassemtype{\env}{\expr}{\typ'}
  \end{enumerate}  
  By Theorem~\ref{proofs-thm:subsoundness} and (2) we have 
  \begin{enumerate}
    \setcounter{enumi}{3}
    \item \issemsubtype{\env}{\typ'}{\typ}
  \end{enumerate}  
  By (3), (4), and the definition of subsets we directly get \hassemtype{\env}{\expr}{\typ}.
  \item[\tself]
  Assume \hastype{\env}{\expr}{\tref{\vv}{\tbase}{\bbbeq{\vv}{\expr}{\tbase}}}. 
  By inversion we have 
  \begin{enumerate}
    \item \hastype{\env}{\expr}{\tref{\vv}{\tbase}{\refa}}
  \end{enumerate} 
    By IH we have 
    \begin{enumerate}
      \setcounter{enumi}{1}
      \item \hassemtype{\env}{\expr}{\tref{\vv}{\tbase}{\refa}}
    \end{enumerate} 
    We fix $\model \in \interp{\env}$. By the definition of semantic typing we get 
    \begin{enumerate}
      \setcounter{enumi}{2}
      \item $\modelapp{\model}{\expr} \in \interp{\modelapp{\model}{\tref{\vv}{\tbase}{\refa}}}$
    \end{enumerate} 
    By the definition of denotations on basic types we have 
    \begin{enumerate}
      \setcounter{enumi}{3}
      \item \goesto{\modelapp{\model}{\expr}}{\val}
      \item \hasbtype{}{\modelapp{\model}{\expr}}{\tbase} 
      \item \goesto{\modelapp{\model}{\refa}\subst{\vv}{\modelapp{\model}{\expr}}}{\etrue}
    \end{enumerate} 
    Since \model contains values, by the definition of \bbbeqsym{\tbase} we have 
    \begin{enumerate}
      \setcounter{enumi}{6}
      \item \goesto{\bbbeq{\modelapp{\model}{\expr}}{\modelapp{\model}{\expr}}{\tbase}}{\etrue}
    \end{enumerate} 
    Thus 
    \begin{enumerate}
      \setcounter{enumi}{7}
      \item \goesto{\modelapp{\model}{(\bbbeq{\vv}{\expr}{\tbase})}\subst{\vv}{\modelapp{\model}{\expr}}}{\etrue}
    \end{enumerate} 
    By (4), (5), and (8) we have 
    \begin{enumerate}
      \setcounter{enumi}{8}
      \item $\modelapp{\model}{\expr} \in \interp{\modelapp{\model}{\tref{\vv}{\tbase}{\bbbeq{\vv}{\expr}{\tbase}}}}$
    \end{enumerate} 
    Thus, \hassemtype{\env}{\expr}{\tref{\vv}{\tbase}{\bbbeq{\vv}{\expr}{\tbase}}}.
  \item[\tcon]
    This case holds exactly because of Property~\ref{theorem:constant-property}. 
  \item[\tvar]
    This case holds by the definition of closing substitutions.  
  \item[\tlam]
    Assume \hastype{\env}{\elam{x}{\typ_x}{\expr}}{\tfun{x}{\typ_x}{\typ}}. 
    By inversion of the rule we have 
    \hastype{\env, \envBind{x}{\typ_x}}{\expr}{\typ}.
    By IH we get
    \hassemtype{\env, \envBind{x}{\typ_x}}{\expr}{\typ}.
    
    We need to show that \hassemtype{\env}{\elam{x}{\typ_x}{\expr}}{\tfun{x}{\typ_x}{\typ}}.
    Which, for some $\model\in\interp{\env}$ is equivalent to 
    \elam{x}{\modelapp{\model}{\typ_x}}{\modelapp{\model}{\expr}} $\in$ \interp{\tfun{x}{\modelapp{\model}{\typ_x}}{\modelapp{\model}{\typ}}}.

    \nv{all this mess is needed because theta has values and not expressions but the denotations are defined over expressions}
    We pick a random $\expr_x \in \interp{\modelapp{\model}{\typ_x}}$ thus 
    we need to show that 
    \modelapp{\model}{\expr\subst{x}{\expr_x}} $\in$ \interp{\modelapp{\model}{\typ\subst{x}{\expr_x}}}.
    By Lemma~\ref{proofs-lemma:semantic-preservation}, there exists $\val_x$ so that 
    \goesto{\expr_x}{\val_x} and $\val_x \in \interp{\typ_x}$.
    By the inductive hypothesis, 
    \modelapp{\model}{\expr\subst{x}{\val_x}} $\in$ \interp{\modelapp{\model}{\typ\subst{x}{\val_x}}}.
    By Lemma~\ref{proofs-lemma:semantic-contextual-preservation}, 
    \modelapp{\model}{\expr\subst{x}{\expr_x}} $\in$ \interp{\modelapp{\model}{\typ\subst{x}{\expr_x}}}, 
    which concludes our proof. 
  \item[\tapp]
  Assume  \hastype{\env}{\expr\ \expr_x}{\typ\subst{x}{\expr_x}}. 
  By inversion we have 
  \begin{enumerate}
    \item \hastype{\env}{\expr}{\tfun{x}{\typ_x}{\typ}} 
    \item \hastype{\env}{\expr_x}{\typ_x}
  \end{enumerate}
  By IH we get 
  \begin{enumerate}
    \setcounter{enumi}{2}
    \item \hassemtype{\env}{\expr}{\tfun{x}{\typ_x}{\typ}} 
    \item \hassemtype{\env}{\expr_x}{\typ_x}
  \end{enumerate}
  We fix $\model \in \interp{\env}$. 
  By the definition of semantic types 
  \begin{enumerate}
    \setcounter{enumi}{4}
    \item \modelapp{\model}{\expr} $\in$ \interp{\modelapp{\model}{\tfun{x}{\typ_x}{\typ}}} 
    \item \modelapp{\model}{\expr_x} $\in$ \interp{\modelapp{\model}{\typ_x}}
  \end{enumerate}
  By (5), (6), and the definition of semantic typing on functions: 
  \begin{enumerate}
    \setcounter{enumi}{6}
    \item \modelapp{\model}{\expr\ \expr_x} $\in$ \interp{\modelapp{\model}{\typ\subst{x}{\expr_x}}}
  \end{enumerate}
  Which directly leads to the required
  \hassemtype{\env}{\expr\ \expr_x}{\typ\subst{x}{\expr_x}}
  \item[\teqbase]
  Assume \hastype{\env}{\ebeq{\tbase}\ \expr_l\ \expr_r\ \expr}{\eqrt{\tbase}{\expr_l}{\expr_r}}.
  By inversion we get: 
  \begin{enumerate}
    \item \hastype{\env}{\expr_l}{\typ_l} 
    \item \hastype{\env}{\expr_r}{\typ_r}
    \item \issubtype{\env}{\typ_l}{\tref{x}{\tbase}{\etrue}}   
    \item \issubtype{\env}{\typ_r}{\tref{x}{\tbase}{\etrue}} 
    \item \hastype{\env, \envBind{r}{\typ_r}, \envBind{l}{\typ_l}}{e}{\tref{\rbind}{\tunit}{\bbbeq{l}{r}{\tbase}}}
  \end{enumerate}
  By IH we get 

  \begin{enumerate}
    \setcounter{enumi}{3}
    \item \hassemtype{\env}{\expr_l}{\typ_l} 
    \item \hassemtype{\env}{\expr_r}{\typ_r}
    \item \hassemtype{\env, \envBind{r}{\typ_r}, \envBind{l}{\typ_l}}{e}{\tref{\rbind}{\tunit}{\bbbeq{l}{r}{\tbase}}}
  \end{enumerate}

  We fix $\model \in \interp{\env}$. Then (4) and (5) become 
  \begin{enumerate}
    \setcounter{enumi}{6}
    \item \modelapp{\model}{\expr_l} $\in$ \interp{\modelapp{\model}{\typ_l}} 
    \item \modelapp{\model}{\expr_r} $\in$ \interp{\modelapp{\model}{\typ_r}} 
    \item \hassemtype{\env}{\expr_r}{\typ_r}
    \item \hassemtype{\env, \envBind{r}{\typ_r}, \envBind{l}{\typ_l}}{e}{\tref{\rbind}{\tunit}{\bbbeq{l}{r}{\tbase}}}
  \end{enumerate}
   Assume 
   \begin{enumerate}
    \setcounter{enumi}{10}
    \item \goesto{\modelapp{\model}{\expr_l}}{\val_l}
    \item \goesto{\modelapp{\model}{\expr_r}}{\val_r}
  \end{enumerate}
  By (7), (8), (11), (12), and Lemma~\ref{proofs-lemma:semantic-preservation} we get
  \begin{enumerate}
    \setcounter{enumi}{12}
    \item $\val_l \in \interp{\modelapp{\model}{\typ_l}}$
    \item $\val_r \in \interp{\modelapp{\model}{\typ_r}}$
  \end{enumerate}

  By (10), (11), and (12) we get 
  \begin{enumerate}
    \setcounter{enumi}{14}
    \item \goesto{\bbbeq{\val_l}{\val_r}{\tbase}}{\etrue}
  \end{enumerate}
By (11), (12), (15), ane Lemma~\ref{proofs-lemma:operational-pres} we have 
\begin{enumerate}
  \setcounter{enumi}{15}
  \item \goesto{\bbbeq{\modelapp{\model}{\expr_l}}{\modelapp{\model}{\expr_r}}{\tbase}}{\etrue}
\end{enumerate}
By (1-5) we get:
\begin{enumerate}
  \setcounter{enumi}{16}
  \item \hasbtype{}{\modelapp{\model}{\ebeq{\tbase}\ \expr_l\ \expr_r\ \expr}}{\eqt{\tbase}}  
\end{enumerate}

Trivially, with zero evaluation steps we have: 
  \begin{enumerate}
    \setcounter{enumi}{17}
    \item \goesto{\modelapp{\model}{\ebeq{\tbase}\ \expr_l\ \expr_r\ \expr}}{\ebeq{\tbase}\ ({\modelapp{\model}{\expr_l}})\ ({\modelapp{\model}{\expr_l}})\ ({\modelapp{\model}{\expr}})}
  \end{enumerate}
By (16), (17), (18) and the definition of semantic types on basic equality types we have 
\begin{enumerate}
  \setcounter{enumi}{18}
  \item \modelapp{\model}{\ebeq{\tbase}\ \expr_l\ \expr_r\ \expr} $\in$ \interp{\modelapp{\model}{\eqrt{\tbase}{\expr_l}{\expr_r}}}
\end{enumerate}
Which leads to the required 
\hassemtype{\env}{\ebeq{\tbase}\ \expr_l\ \expr_r\ \expr}{\eqrt{\tbase}{\expr_l}{\expr_r}}.

  \item[\teqfun]        
  Assume \hastype{\env}{\exeq{x}{\typ_x}{\typ}\ \expr_l\ \expr_r\ \expr}{\eqrt{\tfun{x}{\typ_x}{\typ}}{\expr_l}{\expr_r}}. 
  By inversion we have 
  \begin{enumerate}
    \item \hastype{\env}{\expr_l}{\typ_l} 
    \item \hastype{\env}{\expr_r}{\typ_r}
    \item \issubtype{\env}{\typ_l}{\tfun{x}{\typ_x}{\typ}}
    \item \issubtype{\env}{\typ_r}{\tfun{x}{\typ_x}{\typ}} 
    \item \hastype{\env, \envBind{r}{\typ_r}, \envBind{l}{\typ_l}}{e}{(\tfun{x}{\typ_x}{\eqrt{\typ}{l\ x}{r\ x}})}
    \item \iswellformed{\env}{\tfun{x}{\typ_x}{\typ}}
  \end{enumerate}
  By IH and Theorem~\ref{proofs-thm:subsoundness} we get 
  \begin{enumerate}
    \setcounter{enumi}{6}
    \item \hassemtype{\env}{\expr_l}{\typ_l} 
    \item \hassemtype{\env}{\expr_r}{\typ_r}
    \item \issemsubtype{\env}{\typ_l}{\tfun{x}{\typ_x}{\typ}}
    \item \issemsubtype{\env}{\typ_r}{\tfun{x}{\typ_x}{\typ}} 
    \item \hassemtype{\env, \envBind{r}{\typ_r}, \envBind{l}{\typ_l}}{e}{(\tfun{x}{\typ_x}{\eqrt{\typ}{l\ x}{r\ x}})}
  \end{enumerate}

  By (1-5) we get 
  \begin{enumerate}
    \setcounter{enumi}{11}
    \item \hasbtype{}{\modelapp{\model}{\exeq{x}{\typ_x}{\typ}\ \expr_l\ \expr_r\ \expr}}{\eqt{\unrefine{\modelapp{\model}{(\tfun{x}{\typ_x}{\typ})}}}} 
  \end{enumerate}

  Trivially, by zero evaluation steps, we get 
  \begin{enumerate}
    \setcounter{enumi}{12}
    \item \goesto{\modelapp{\model}{\exeq{x}{\typ_x}{\typ}\ \expr_l\ \expr_r\ \expr}}{\exeq{x}{\modelapp{\model}{\typ_x}}{\modelapp{\model}{\typ}}\ ({\modelapp{\model}{\expr_l}})\ (\modelapp{\model}{\expr_r})\ (\modelapp{\model}{\expr})}
  \end{enumerate}
  
  By (7-10) we get 
  \begin{enumerate}
    \setcounter{enumi}{13}
    \item $\modelapp{\model}{\expr_l}, \modelapp{\model}{\expr_r} \in \interp{\modelapp{\model}{\tfun{x}{\typ_x}{\typ}}} $
  \end{enumerate}
  
  By (7), (8), (11), the definition of semantic types on functions, and Lemmata~\ref{proofs-lemma:semantic-preservation} and~\ref{proofs-lemma:semantic-contextual-preservation} 
  (similar to the previous case) we have 
  \begin{itemize}
    \setcounter{enumi}{14}
    \item $\forall \expr_x\in\interp{\typ_x}. \expr\ \expr_x \in \interp{\eqrt{\typ\subst{x}{\expr_x}}{\expr_l\ \expr_x}{\expr_r\ \expr_x}} $
  \end{itemize}

  By (12), (13), (14), and (15) we get 
  \begin{enumerate}
    \setcounter{enumi}{18}
    \item \modelapp{\model}{\exeq{x}{\typ_x}{\typ}\ \expr_l\ \expr_r\ \expr} $\in$ \interp{\modelapp{\model}{\eqrt{\tfun{x}{\typ_x}{\typ}}{\expr_l}{\expr_r}}}
  \end{enumerate}
  Which leads to the required 
  \hassemtype{\env}{\exeq{x}{\typ_x}{\typ}\ \expr_l\ \expr_r\ \expr}{\eqrt{\tfun{x}{\typ_x}{\typ}}{\expr_l}{\expr_r}}. 
\end{itemize}
\end{proof}
\end{theorem}

\begin{lemma}\label{proofs-lemma:semantic-preservation}
If $\expr \in \interp{\typ}$, then \goesto{\expr}{\val} and $\val \in \interp{\typ}$.
\end{lemma}
\begin{proof}
By structural induction of the type $\typ$. 
\end{proof}

\begin{lemma}\label{proofs-lemma:semantic-contextual-preservation}
If \goesto{\expr_x}{\val_x} and $\expr\subst{x}{\val_x} \in \interp{\typ\subst{x}{\val_x}}$, 
then $\expr\subst{x}{\expr_x} \in \interp{\typ\subst{x}{\expr_x}}$.
\end{lemma}
\begin{proof}
  We can use parallel reductions (of~\S\ref{app:parred}) to prove that if 
  $\parreds{\expr_1}{\expr_2}$, 
  then (1) $\interp{\typ\subst{x}{\expr_1}} = \interp{\typ\subst{x}{\expr_2}}$ and 
       (2) $\expr_1 \in \interp{\typ}$ \textit{iff} $\expr_2 \in \interp{\typ}$. 
  The proof directly follows by these two properties. \todo{actually do the proof}
\end{proof}

\begin{lemma}\label{proofs-lemma:operational-pres}
  If \goesto{\expr_x}{\expr_x'} and \goesto{\expr\subst{x}{\expr_x'}}{\con}, 
  then \goesto{\expr\subst{x}{\expr_x}}{\con}. 
\end{lemma}
\begin{proof}
  As an instance of Corollary~\ref{cor:cotermination-multi}. \todo{actually do the proof}
\end{proof}

We define semantic subtyping as follows: 
\issemsubtype{\env}{\typ}{\typ'} \text{iff} $\forall \model \in \interp{\env}. \interp{\modelapp{\model}{\typ}} \subseteq \interp{\modelapp{\model}{\typ'}}$.
\begin{theorem}[Subtyping semantic soundness]\label{proofs-thm:subsoundness}
  If \issubtype{\env}{\typ}{\typ'} then \issemsubtype{\env}{\typ}{\typ'}.
\begin{proof}
  By induction on the derivation tree: 
  \begin{itemize}
    \item[\subBase]
    Assume \issubtype{\env}{\tref{\rbind}{\tbase}{\refa}}{\tref{\rbind'}{\tbase}{\refa'}}. 
    By inversion 
    $\forall \model\in\interp{\env},~ \interp{\modelapp{\model}{\tref{\rbind}{\tbase}{\refa}}}
      \subseteq \interp{\modelapp{\model}{\tref{\rbind'}{\tbase}{\refa'}}}$, which exactly leads to the required. 
    \item[\subFun]
    Assume \issubtype{\env}{\tfun{x}{\typ_x}{\typ}}{\tfun{x}{\typ'_x}{\typ'}}. 
    By inversion 
    \begin{enumerate}
      \item \issubtype{\env}{\typ'_x}{\typ_x}
      \item \issubtype{\env, \envBind{x}{\typ'_x}}{\typ}{\typ'}
    \end{enumerate}
    By IH 
    \begin{enumerate}
      \setcounter{enumi}{2}
      \item \issemsubtype{\env}{\typ'_x}{\typ_x}
      \item \issemsubtype{\env, \envBind{x}{\typ'_x}}{\typ}{\typ'}
    \end{enumerate}
    We fix $\model\in \env$. 
    We pick $\expr$. We assume $\expr \in \interp{\modelapp{\model}{\tfun{x}{\typ_x}{\typ}}}$ 
    and we will show that $\expr \in \interp{\modelapp{\model}{\tfun{x}{\typ'_x}{\typ'}}}$.
    By assumption 
    \begin{enumerate}
      \setcounter{enumi}{4}
      \item $\forall \expr_x \in \interp{\modelapp{\model}{\typ_x}}\!.\ \expr\ \expr_x \in \interp{\modelapp{\model}{\typ\subst{x}{\expr_x}}}$
    \end{enumerate}
    We need to show 
    $\forall \expr_x \in \interp{\modelapp{\model}{\typ'_x}}\!.\ \expr\ \expr_x \in \interp{\modelapp{\model}{\typ'\subst{x}{\expr_x}}}$.
    We fix $\expr_x$. 
    By (3), if $\expr_x \in \interp{\modelapp{\model}{\typ'_x}}$, then 
    $\expr_x \in \interp{\modelapp{\model}{\typ_x}}$ and (5) applies, so 
    $\expr\ \expr_x \in \interp{\modelapp{\model}{\typ\subst{x}{\expr_x}}}$, which by (4) gives
    $\expr\ \expr_x \in \interp{\modelapp{\model}{\typ'\subst{x}{\expr_x}}}$.
    Thus, $\expr \in \interp{\modelapp{\model}{\tfun{x}{\typ'_x}{\typ'}}}$.
    This leads to $\interp{\modelapp{\model}{\tfun{x}{\typ_x}{\typ}}} \subseteq \interp{\modelapp{\model}{\tfun{x}{\typ'_x}{\typ'}}}$, 
    which by definition gives semantic subtyping: 
    \issemsubtype{\env}{\tfun{x}{\typ_x}{\typ}}{\tfun{x}{\typ'_x}{\typ'}}.
    \item[\subEq]  
    Assume \issubtype{\env}{\eqrt{\typ_i}{\expr_l}{\expr_r}}{\eqrt{\typ_i'}{\expr_l}{\expr_r}}. 
    We split cases on the structure of $\typ_i$. 
    \begin{itemize}
      \item \textit{If } $\typ_i$ { is a basic type}, then $\typ_i$ is trivially refined to true. 
            Thus, $\typ_i = \typ_i' = \tbase$ and  
            for each $\model\in\env$, $\interp{\modelapp{\model}{\eqrt{\typ}{\expr_l}{\expr_r}}} = \interp{\modelapp{\model}{\eqrt{\typ'}{\expr_l}{\expr_r}}}$, 
            thus set inclusion reduces to equal sets. 
      \item \textit{If } $\typ_i$ { is a function type}, thus 
            \issubtype{\env}{\eqrt{\tfun{x}{\typ_x}{\typ}}{\expr_l}{\expr_r}}{\eqrt{\tfun{x}{\typ'_x}{\typ'}}{\expr_l}{\expr_r}}
    \end{itemize}
    By inversion 
    \begin{enumerate}
      \item \issubtype{\env}{\tfun{x}{\typ_x}{\typ}}{\tfun{x}{\typ'_x}{\typ'}}
      \item \issubtype{\env}{\tfun{x}{\typ'_x}{\typ'}}{\tfun{x}{\typ_x}{\typ}}
    \end{enumerate}
    By inversion on (1) and (2) we get 
    \begin{enumerate}
      \setcounter{enumi}{2}
      \item \issubtype{\env}{\typ'_x}{\typ_x}
      \item \issubtype{\env, \envBind{x}{\typ'_x}}{\typ}{\typ'}
      \item \issubtype{\env, \envBind{x}{\typ_x}}{\typ'}{\typ}
    \end{enumerate}
    By IH on (1) and (3) we get 
    \begin{enumerate}
      \setcounter{enumi}{5}
      \item \issemsubtype{\env}{\tfun{x}{\typ_x}{\typ}}{\tfun{x}{\typ'_x}{\typ'}}
      \item \issemsubtype{\env}{\typ'_x}{\typ_x}
    \end{enumerate}
    We fix $\model\in\env$ and some $\expr$.
    If $\expr \in \interp{\modelapp{\model}{\eqrt{\tfun{x}{\typ_x}{\typ}}{\expr_l}{\expr_r}}}$
    we need to show that 
    $\expr \in \interp{\modelapp{\model}{\eqrt{\tfun{x}{\typ'_x}{\typ'}}{\expr_l}{\expr_r}}}$. 
    By the assumption we have 
    \begin{enumerate}
      \setcounter{enumi}{7}
      \item $\hasbtype{}{\expr}{\eqt{\unrefine{\modelapp{\model}{(\tfun{x}{\typ_x}{\typ})}}}}$
      \item $\goesto{\expr}{\exeqName_{\_}\ (\modelapp{\model}{\expr_l})\ (\modelapp{\model}{\expr_r})\ {\expr_{pf}}}$
      \item $(\modelapp{\model}{\expr_l}), (\modelapp{\model}{\expr_r}) \in \interp{\modelapp{\model}(\tfun{x}{\typ_x}{\typ})}$
      \item $\forall \expr_x\in\interp{\modelapp{\model}{\typ_x}}. \expr_{pf}\ \expr_x \in \interp{\eqrt{\modelapp{\model}(\typ\subst{x}{\expr_x})}{(\modelapp{\model}{\expr_l})\ \expr_x}{(\modelapp{\model}{\expr_r})\ \expr_x}}$            
    \end{enumerate}
    Since (8) only depends on the structure of the type index, we get 
    \begin{enumerate}
      \setcounter{enumi}{11}
      \item $\hasbtype{}{\expr}{\eqt{\unrefine{\modelapp{\model}{(\tfun{x}{\typ'_x}{\typ'})}}}}$
    \end{enumerate}
    By (6) and (10) we get 
    \begin{enumerate}
      \setcounter{enumi}{12}
      \item $(\modelapp{\model}{\expr_l}), (\modelapp{\model}{\expr_r}) \in \interp{\modelapp{\model}(\tfun{x}{\typ'_x}{\typ'})}$
    \end{enumerate}
    By (4), (5), Lemma~\ref{lem:sub-strengthening}, the rule \subEq and the IH, we get that 
    $\interp{\eqrt{\modelapp{\model}(\typ\subst{x}{\expr_x})}{(\modelapp{\model}{\expr_l})\ \expr_x}{(\modelapp{\model}{\expr_r})\ \expr_x}}
    \subseteq
    \interp{\eqrt{\modelapp{\model}(\typ'\subst{x}{\expr_x})}{(\modelapp{\model}{\expr_l})\ \expr_x}{(\modelapp{\model}{\expr_r})\ \expr_x}}$. 
    By which, (11), (7), and reasoning similar to the \subFun case, we get 
    \begin{enumerate}
      \setcounter{enumi}{13}
      \item $\forall \expr_x\in\interp{\modelapp{\model}{\typ'_x}}. \expr_{pf}\ \expr_x \in \interp{\eqrt{\modelapp{\model}(\typ'\subst{x}{\expr_x})}{(\modelapp{\model}{\expr_l})\ \expr_x}{(\modelapp{\model}{\expr_r})\ \expr_x}}$            
    \end{enumerate}
    By (12), (9), (13), and (14) we conclude that 
    $\expr \in \interp{\modelapp{\model}{\eqrt{\tfun{x}{\typ'_x}{\typ'}}{\expr_l}{\expr_r}}}$, 
    thus 
    \issemsubtype{\env}{\eqrt{\tfun{x}{\typ_x}{\typ}}{\expr_l}{\expr_r}}{\eqrt{\tfun{x}{\typ'_x}{\typ'}}{\expr_l}{\expr_r}}.
  \end{itemize}
\end{proof}
\end{theorem}

\begin{lemma}[Strengthening]\label{lem:sub-strengthening}
  If \issubtype{\env_1}{\typ_1}{\typ_2}, then:
  \begin{enumerate}
  \item If   \hastype{\env_1, \envBind{x}{\typ_2}, \env_2}{\expr}{\typ}
        then \hastype{\env_1, \envBind{x}{\typ_1}, \env_2}{\expr}{\typ}.
  \item If   \issubtype{\env_1, \envBind{x}{\typ_2}, \env_2}{\typ}{\typ'}
        then \issubtype{\env_1, \envBind{x}{\typ_1}, \env_2}{\typ}{\typ'}.
  \item If   \iswellformed{\env_1, \envBind{x}{\typ_2}, \env_2}{\typ}
        then \iswellformed{\env_1, \envBind{x}{\typ_1}, \env_2}{\typ}.
  \item If   \envwellformed{\env_1, \envBind{x}{\typ_2}, \env_2}
        then \envwellformed{\env_1, \envBind{x}{\typ_1}, \env_2}.
  \end{enumerate}
\begin{proof}
  The proofs go by induction. 
  Only the \tvar case is insteresting; we use \tsub and our assumption.
  \todo{By just the biggest, most mutual induction you ever seen. Only
    the \tvar case matters; we use \tsub and our assumption. }
\end{proof}
\end{lemma}

\begin{lemma}[Semantic typing is closed under parallel reduction in expressions]\label{lem:sem-parred-expr}
  If \parredsto{\expr_1}{\expr_2}, then $\expr_1 \in \interp{\typ}$ iff $\expr_2 \in \interp{\typ}$.
\begin{proof}
  By induction on $\typ$, using parallel reduction as a bisimulation
  (Lemma~\ref{lem:parred-forward-simulation} and
  Corollary~\ref{cor:parred-backward-simulation}).
\end{proof}
\end{lemma}

\begin{lemma}[Semantic typing is closed under parallel reduction in types]\label{lem:sem-parred-type}
  If \parredsto{\typ_1}{\typ_2} then $\interp{\typ_1} = \interp{\typ_2}$.
\begin{proof}
  By induction on $\typ_1$ (which necessarily has the same shape as
  $\typ_2$). We use parallel reduction as a bisimulation
  (Lemma~\ref{lem:parred-forward-simulation} and
  Corollary~\ref{cor:parred-backward-simulation}).
\end{proof}
\end{lemma}

\begin{lemma}[Parallel reducing types are equal]\label{lem:subtype-parred}
  If \iswellformed{\env}{\typ_1} and \iswellformed{\env}{\typ_2} and
  \parredsto{\typ_1}{\typ_2} then \issubtype{\env}{\typ_1}{\typ_2} and
  \issubtype{\env}{\typ_1}{\typ_2}.
\begin{proof}
  By induction on the parallel reduction sequence; for a single step,
  by induction on $\typ_1$ (which must have the same structure as
  $\typ_2$).
   We use parallel reduction as a bisimulation
   (Lemma~\ref{lem:parred-forward-simulation} and
   Corollary~\ref{cor:parred-backward-simulation}).
\end{proof}
\end{lemma}

\begin{lemma}[Regularity]\label{lem:regularity}
\begin{enumerate}
\item If \hastype{\env}{\expr}{\typ} then \envwellformed{\env} and \iswellformed{\env}{\typ}.
\item If \iswellformed{\env}{\typ} then \envwellformed{\env}.
\item If \issubtype{\env}{\typ_1}{\typ_2} then \envwellformed{\env} and \iswellformed{\env}{\typ_1} and \iswellformed{\env}{\typ_2}.
\end{enumerate}
\begin{proof}
  By a big ol' induction.
\todo{the \subBase case will need extra wf assumptions on the rule}
\end{proof}
\end{lemma}

\begin{lemma}[Canonical forms]\label{lem:canonical-forms}
  If \hastype{\env}{\val}{\typ}, then:
  \begin{itemize}
  \item If $\typ = \tref{x}{\tbase}{\expr}$, then $\val = \con$ such
    that $\tycon{\con} = \tbase$ and
    \issubtype{\env}{\tycon{\con}}{\tref{x}{\tbase}{\expr}}.

  \item If $\typ = \tfun{x}{\typ_x}{\typ'}$, then $\val =
    \tlam{x}{\typ_x'}{\expr}$ such that
    \issubtype{\env}{\typ_x}{\typ_x'} and \hastype{\env, \envBind{x}{\typ_x'}}{\expr}{\typ''} such that \issubtype{\typ''}{\typ'}.

  \item If $\typ = \eqrt{\tbase}{\expr_l}{\expr_r}$ then $\val =
    \ebeq{\tbase}\ \expr_l\ \expr_r\ \val_p$ such that
    \hastype{\env}{\expr_l}{\typ_l} and
    \hastype{\env}{\expr_r}{\typ_r} (for some $\typ_l$ and $\typ_r$
    that are refinements of $\tbase$) and \hastype{\env,
      \envBind{r}{\typ_r},
      \envBind{l}{\typ_l}}{\val_p}{\tref{x}{\tunit}{\bbbeq{l}{r}{\tbase}}}.

  \item If $\typ = \eqrt{\tfun{x}{\typ_x}{\typ'}}{\expr_l}{\expr_r}$ then $\val =
    \exeq{x}{\typ_x'}{\typ''}\ \expr_l\ \expr_r\ \val_p$ such that
    \issubtype{\env}{\typ_x}{\typ_x'} and \issubtype{\env, \envBind{x}{\typ_x}}{\typ''}{\typ'} and
    \hastype{\env}{\expr_l}{\typ_l} and
    \hastype{\env}{\expr_r}{\typ_r} (for some $\typ_l$ and $\typ_r$
    that are subtypes of \tfun{x}{\typ_x'}{\typ''}) and \hastype{\env,
      \envBind{r}{\typ_r},
      \envBind{l}{\typ_l}}{\val_p}{\tfun{x}{\typ_x'}{\eqrt{\typ''}{\expr_l\ x}{\expr_r\ x}}}.
  \end{itemize}
\end{lemma}

\subsection{The Binary Logical Relation}\label{proofs-logic-rel}

\begin{theorem}[$\mathsf{EqRT}$ soundness]\label{proofs-thm:eq-relation}
  If \hastype{\env}{\expr}{\eqrt{\typ}{\expr_1}{\expr_2}}, 
  then \relatesEnv{\env}{\expr_1}{\expr_2}{\typ}.
\end{theorem}

\begin{proof}
  By \hastype{\env}{\expr}{\eqrt{\typ}{\expr_1}{\expr_2}} and 
  the Fundamental Property~\ref{thm:fundamental} we have 
  \relatesEnv{\env}{\expr}{\expr}{\eqrt{\typ}{\expr_1}{\expr_2}}.
  Thus, for a fixed $\rmodel \in \env$, 
  \relatesModel{\rmodel}{\expr}{\expr}{\eqrt{\typ}{\expr_1}{\expr_2}}.
  By the definition of the logical relation for $\texttt{EqRT}$, we have 
  \relatesModel{\rmodel}{\expr_1}{\expr_2}{\typ}.
  So, \relatesEnv{\env}{\expr_1}{\expr_2}{\typ}.
\end{proof}

\begin{lemma}[LR respects subtyping]\label{proofs-lemma:sub}
  If \relatesEnv{\env}{\expr_1}{\expr_2}{\typ} and \issubtype{\env}{\typ}{\typ'}, 
  then \relatesEnv{\env}{\expr_1}{\expr_2}{\typ'}.
\end{lemma}
\begin{proof}
By induction on the derivation of the subtyping tree.
\begin{itemize}
  \item[\subBase]
  By assumption we have 
  \begin{enumerate}
    \item \relatesEnv{\env}{\expr_1}{\expr_2}{\tref{\rbind}{\tbase}{\refa}}
    \item \issubtype{\env}{\tref{\rbind}{\tbase}{\refa}}{\tref{\rbind'}{\tbase}{\refa'}}
  \end{enumerate}

  By inversion on (2) we get 
  \begin{enumerate}
    \setcounter{enumi}{2}
    \item $  \forall \model\in\interp{\env},~ \interp{\modelapp{\model}{\tref{\rbind}{\tbase}{\refa}}}
    \subseteq \interp{\modelapp{\model}{\tref{\rbind'}{\tbase}{\refa'}}}$
  \end{enumerate}
  
  We fix $\rmodel\in\env$.
  By (1) we get 
  \begin{enumerate}
    \setcounter{enumi}{3}
    \item \relatesModel{\rmodel}{\expr_1}{\expr_2}{\tref{\rbind}{\tbase}{\refa}}
  \end{enumerate}
  By the definition of logical relations: 
  \begin{enumerate}
    \setcounter{enumi}{4}
    \item \goesto{\modelapp{\rmodel_1}{\expr_1}}{\val_1}
    \item \goesto{\modelapp{\rmodel_2}{\expr_2}}{\val_2}
    \item \relates{\rmodel}{\val_1}{\val_2}{\tref{\rbind}{\tbase}{\refa}}
  \end{enumerate}
  By (7) and the definition of the logical relation on basic types we have 
  \begin{enumerate}
    \setcounter{enumi}{7}
  \item $\val_1 = \val_2 = c$
  \item \hasbtype{}{\con}{\tbase}  
  \item \goesto{\modelapp{\rmodel_1}{\refa\subst{\rbind}{\con}}}{\etrue} 
  \item \goesto{\modelapp{\rmodel_2}{\refa\subst{\rbind}{\con}}}{\etrue}  
\end{enumerate}

By (3), (10) and (11) become 
\begin{enumerate}
  \setcounter{enumi}{11}
\item \goesto{\modelapp{\rmodel_1}{\refa'\subst{\rbind'}{\con}}}{\etrue} 
\item \goesto{\modelapp{\rmodel_2}{\refa'\subst{\rbind'}{\con}}}{\etrue}  
\end{enumerate}
By (8), (9), (12), and (13) we get 

\begin{enumerate}
  \setcounter{enumi}{13}
  \item \relates{\rmodel}{\val_1}{\val_2}{\tref{\rbind'}{\tbase}{\refa'}}
\end{enumerate}

By (5), (6), and (14) we have 
  \begin{enumerate}
    \setcounter{enumi}{14}
    \item \relatesModel{\rmodel}{\expr_1}{\expr_2}{\tref{\rbind'}{\tbase}{\refa'}}
  \end{enumerate}

  Thus, \relatesEnv{\env}{\expr_1}{\expr_2}{\tref{\rbind'}{\tbase}{\refa'}}.

  \item[\subFun]
  By assumption: 
  \begin{enumerate}
    \item \relatesEnv{\env}{\expr_1}{\expr_2}{\tfun{x}{\typ_x}{\typ}}
    \item \issubtype{\env}{\tfun{x}{\typ_x}{\typ}}{\tfun{x}{\typ'_x}{\typ'}}
  \end{enumerate}
  By inversion of the rule (2)
  \begin{enumerate}
    \setcounter{enumi}{2}
    \item \issubtype{\env}{\typ'_x}{\typ_x}  
    \item \issubtype{\env, \envBind{x}{\typ'_x}}{\typ}{\typ'}
\end{enumerate}
  We fix $\rmodel\in\env$. 
  By (1) and the definition of logical relation 
 \begin{enumerate}
  \setcounter{enumi}{4}
  \item \goesto{\modelapp{\rmodel_1}{\expr_1}}{\val_1}
  \item \goesto{\modelapp{\rmodel_2}{\expr_2}}{\val_2}
  \item \relates{\rmodel}{\val_1}{\val_2}{\tfun{x}{\typ_x}{\typ}}
\end{enumerate}
We fix $\val'_1$ and $\val'_2$ so that 
\begin{enumerate}
  \setcounter{enumi}{7}
  \item \relates{\rmodel}{\val'_1}{\val'_2}{\typ'_x}
\end{enumerate}
By (8) and the definition of logical relations, since the values are idempotent under substitution, 
we have 
\begin{enumerate}
  \setcounter{enumi}{8}
  \item \relatesEnv{\env}{\val'_1}{\val'_2}{\typ'_x}
\end{enumerate}

By (9) and inductive hypothesis on (3) we have 
\begin{enumerate}
  \setcounter{enumi}{9}
  \item \relatesEnv{\env}{\val'_1}{\val'_2}{\typ_x}
\end{enumerate}

By (10), idempotence of values under substitution, 
and the definition of logical relations, we have 
\begin{enumerate}
  \setcounter{enumi}{10}
  \item \relates{\rmodel}{\val'_1}{\val'_2}{\typ_x}
\end{enumerate}

By (7), (11), and the definition of logical relations on function values:
\begin{enumerate}
  \setcounter{enumi}{11}
  \item \relates{\rmodel,\rmodelBind{\val'_1}{\val'_2}{x}}{\val_1\ \val'_1}{\val_2\ \val'_2}{\typ}
\end{enumerate}

By (9), (12), and the definition of logical relations we have 
\begin{enumerate}
  \setcounter{enumi}{11}
  \item \relatesEnv{\env,\envBind{x}{\typ'_x}}{\val_1\ \val'_1}{\val_2\ \val'_2}{\typ}
\end{enumerate}

By (12) and inductive hypothesis on (4) we have 
\begin{enumerate}
  \setcounter{enumi}{12}
  \item \relatesEnv{\env,\envBind{x}{\typ'_x}}{\val_1\ \val'_1}{\val_2\ \val'_2}{\typ'}
\end{enumerate}

By (8), (13), and the definition of logical relations, we have 

\begin{enumerate}
  \setcounter{enumi}{13}
  \item \relates{\rmodel,\rmodelBind{\val'_1}{\val'_2}{x}}{\val_1\ \val'_1}{\val_2\ \val'_2}{\typ'}
\end{enumerate}

By (8), (14), and the definition of logical relations, we have 

\begin{enumerate}
  \setcounter{enumi}{14}
  \item \relates{\rmodel}{\val_1}{\val_2}{\tfun{x}{\typ'_x}{\typ'}}
\end{enumerate}

By (5), (6), and (15), we get 

\begin{enumerate}
  \setcounter{enumi}{15}
  \item \relatesModel{\rmodel}{\expr_1}{\expr_2}{\tfun{x}{\typ'_x}{\typ'}}
\end{enumerate}

So, \relatesEnv{\env}{\expr_1}{\expr_2}{\tfun{x}{\typ'_x}{\typ'}}.

\item[\subEq]  
By hypothesis:
\begin{enumerate}
  \item  \relatesEnv{\env}{\expr_1}{\expr_2}{\eqrt{\typ}{\expr_l}{\expr_r}} 
  \item  \issubtype{\env}{\eqrt{\typ}{\expr_l}{\expr_r}}{\eqrt{\typ'}{\expr_l}{\expr_r}}
\end{enumerate}
We fix $\rmodel\in\env$.
By (1)
\begin{enumerate}
  \setcounter{enumi}{2}
  \item  \relatesModel{\rmodel}{\expr_1}{\expr_2}{\eqrt{\typ}{\expr_l}{\expr_r}} 
\end{enumerate}

By (3) and the definition of logical relations. 
\begin{enumerate}
  \setcounter{enumi}{3}
  \item \goesto{\modelapp{\rmodel_1}{\expr_1}}{\val_1}
  \item \goesto{\modelapp{\rmodel_2}{\expr_2}}{\val_2}
  \item \relates{\rmodel}{\val_1}{\val_2}{\eqrt{\typ}{\expr_l}{\expr_r}}
\end{enumerate}
By (6) and the definition of logical relations 
\begin{enumerate}
  \setcounter{enumi}{6}
  \item \relates{\rmodel}{\modelapp{\rmodel_1}{\expr_l}}{\modelapp{\rmodel_2}{\expr_r}}{\typ}
\end{enumerate}

By (7) and the definition of logical relations. 
\begin{enumerate}
  \setcounter{enumi}{7}
  \item \relatesEnv{\env}{{\expr_l}}{{\expr_r}}{\typ}
\end{enumerate}

By inversion on (2)
\begin{enumerate}
  \setcounter{enumi}{8}
  \item \issubtype{\env}{\typ}{\typ'}
  \item \issubtype{\env}{\typ'}{\typ}
\end{enumerate}

By (8) and inductive hypothesis on (9) 

\begin{enumerate}
  \setcounter{enumi}{10}
  \item \relatesEnv{\env}{{\expr_l}}{{\expr_r}}{\typ'}
\end{enumerate}

Thus,

\begin{enumerate}
  \setcounter{enumi}{11}
  \item \relatesModel{\delta}{{\expr_l}}{{\expr_r}}{\typ'}
\end{enumerate}

By (12), (4), (5), and determinism of operational semantics: 
\begin{enumerate}
  \setcounter{enumi}{11}
  \item \relates{\rmodel}{\val_1}{\val_2}{\eqrt{\typ'}{\expr_l}{\expr_r}}
\end{enumerate}

By (4), (5), and (13)
\begin{enumerate}
  \setcounter{enumi}{13}
  \item  \relatesModel{\rmodel}{\expr_1}{\expr_2}{\eqrt{\typ'}{\expr_l}{\expr_r}} 
\end{enumerate}

So, by definition of logical relations, 
\relatesEnv{\env}{\expr_1}{\expr_2}{\eqrt{\typ'}{\expr_l}{\expr_r}}. 
\end{itemize}
\end{proof}

\begin{lemma}[Constant soundness]\label{proofs-lemma:const}
  \relatesEnv{\env}{\con}{\con}{\tycon{\con}}
\end{lemma}
\begin{proof}
  The proof follows the same steps as Theorem~\ref{theorem:constant-property}.
\end{proof}

\begin{lemma}[Selfification of constants]\label{proofs-lemma:self}
  If \relatesEnv{\env}{\expr}{\expr}{\tref{\vv}{\tbase}{\refa}}
  then \relatesEnv{\env}{x}{x}{\tref{\vv}{\tbase}{\bbbeq{\vv}{x}{\tbase}}}.
\end{lemma}
\begin{proof}
We fix $\rmodel \in \env$. 
By hypothesis $\rmodelBind{\val_1}{\val_2}{x} \in \rmodel$ with 
\relates{\rmodel}{\val_1}{\val_2}{\tref{\vv}{\tbase}{\refa}}.
We need to show that 
\relatesModel{\rmodel}{x}{x}{\tref{\vv}{\tbase}{\bbbeq{\vv}{x}{\tbase}}}.
Which reduces to 
\relates{\rmodel}{\val_1}{\val_2}{\tref{\vv}{\tbase}{\bbbeq{\vv}{x}{\tbase}}}.
By the definition on the logical relation on basic values, we know 
$\val_1 = \val_2 = c$ and \hasbtype{}{\con}{\tbase}. 
Thus, we are left to prove that   
\goesto{\modelapp{\rmodel_1}{((\bbbeq{\vv}{x}{\tbase})\subst{\vv}{\con})}}{\etrue} and  
\goesto{\modelapp{\rmodel_2}{((\bbbeq{\vv}{x}{\tbase})\subst{\vv}{\con})}}{\etrue} 
which, both, trivially hold by the definition of \bbbeqsym{\tbase}.
\end{proof}

\begin{lemma}[Variable soundness]\label{proofs-lemma:var}
  If $\envBind{x}{\typ} \in \env$, 
  then \relatesEnv{\env}{x}{x}{\typ}.
\end{lemma}
\begin{proof}
By the definition of the logical relation it suffices to show that 
$\forall \rmodel \in \env. \relates{\rmodel}{\rmodel_1(x)}{\rmodel_2(x)}{\typ}$; 
which is trivially true by the definition of $\rmodel\in\env$.
\end{proof}

\begin{lemma}[Transitivity of Evaluation]\label{proofs-lemma:op-transitivity}
  If \goesto{\expr}{\expr'}, then 
  \goesto{\expr}{\val} \textit{iff} \goesto{\expr'}{\val}.
\end{lemma}
\begin{proof}
  Assume \goesto{\expr}{\val}. 
  Since the \evals{}{} is by definition deterministic, 
  there exists a unique sequence 
  $\expr \evals{}{} \expr_1 \evals{}{} \dots \evals{}{} \expr_i \evals{}{} \dots \evals{}{} \val$.
  By assumption, \goesto{\expr}{\expr'}, so there exists a $j$, so $\expr' \equiv \expr_j$, 
  and \goesto{\expr'}{\val} following the same sequence. 

  Assume \goesto{\expr'}{\val}. 
  Then \goesto{\expr}{\goesto{\expr'}{\val}} uniquely evaluates $\expr$ to $\val$. 
\end{proof}

\begin{lemma}[LR closed under evaluation]\label{proofs-lemma:app}
  If \goesto{\expr_1}{\expr'_1}, 
  \goesto{\expr_2}{\expr'_2}, then 
  \relates{\rmodel}{\expr'_1}{\expr'_2}{\typ}
  \textit{iff} \relates{\rmodel}{\expr_1}{\expr_2}{\typ}. 
\end{lemma}
\begin{proof}
Assume \relates{\rmodel}{\expr'_1}{\expr'_2}{\typ}, by the definition of 
the logical relation on closed terms we have 
\goesto{\expr'_1}{\val_1}, \goesto{\expr'_2}{\val_2}, 
and \relates{\rmodel}{\val_1}{\val_2}{\typ}.
By Lemma~\ref{proofs-lemma:op-transitivity} and 
by assumption, \goesto{\expr_1}{\expr'_1} and \goesto{\expr_2}{\expr'_2}, 
we have \goesto{\expr_1}{\val_1} and \goesto{\expr_2}{\val_2}.
By which and \relates{\rmodel}{\val_1}{\val_2}{\typ} we get that 
\relates{\rmodel}{\expr_1}{\expr_2}{\typ}.
The other direction is identical.
\end{proof}

\begin{lemma}[LR closed under parallel reduction]\label{proofs-lemma:lr-parred}
  If \parredsto{\expr_1}{\expr'_1}, 
  \parredsto{\expr_2}{\expr'_2}, and 
  \relates{\rmodel}{\expr'_1}{\expr'_2}{\typ},
  then \relates{\rmodel}{\expr_1}{\expr_2}{\typ}. 
\mmg{another iff, if we want it}
\begin{proof}
By induction on \typ, using parallel reduction as a backward
simulation (Corollary~\ref{cor:parred-backward-simulation}).
\todo{fill out cases}
\end{proof}
\end{lemma}

\begin{lemma}[LR Compositionality]\label{proofs-lemma:apptyp}
  If \goesto{\modelapp{\delta_1}{\expr_x}}{\val_{x_1}}, 
  \goesto{\modelapp{\delta_2}{\expr_x}}{\val_{x_2}},  
  \relates{\rmodel, \rmodelBind{\val_{x_1}}{\val_{x_2}}{x}}{\expr_1}{\expr_2}{\typ},
  then \relates{\rmodel}{\expr_1}{\expr_2}{\typ\subst{x}{\expr_x}}. 
\end{lemma}
\begin{proof}
By the assumption we have that 
\begin{enumerate}
  \item \goesto{\modelapp{\delta_1}{\expr_x}}{\val_{x_1}}
  \item \goesto{\modelapp{\delta_2}{\expr_x}}{\val_{x_2}}
  \item \goesto{\expr_1}{\val_1}
  \item \goesto{\expr_2}{\val_2}
  \item \relates{\rmodel, \rmodelBind{\val_{x1}}{\val_{x_2}}{x}}{\val_1}{\val_2}{\typ}
\end{enumerate}
and we need to prove that 
\relates{\rmodel}{\val_1}{\val_2}{\typ\subst{x}{\expr_x}}.
The proof goes by structural induction on the type $\typ$.
\begin{itemize}
  \item $\typ \doteq \tref{\vv}{\tbase}{\refa}$. 
  For $i = 1,2$ we need to show that 
  if \goesto{\modelapp{\rmodel_i,\subst{x}{\val_{x_i}}}{\refa\subst{\vv}{\val_i}}}{\etrue} 
  then \goesto{\modelapp{\rmodel_i}{\refa\subst{\vv}{\val_i}\subst{x}{\expr_i}}}{\etrue}.
  We have \parredsto%
  {\modelapp{\rmodel_i,\subst{x}{\val_{x_i}}}{\refa\subst{\vv}{\val_i}}}%
  {\modelapp{\rmodel_i}{\refa\subst{\vv}{\val_i}\subst{x}{\expr_i}}}
  because substituting parallel reducing terms parallel reduces
  (Corollary~\ref{cor:parred-subst-eval-multi}) and parallel reduction
  subsumes reduction (Lemma~\ref{lem:parred-eval}).
  By cotermination at constants (Corollary~\ref{cor:cotermination-multi}), we have
  \goesto{\modelapp{\rmodel_i}{\refa\subst{\vv}{\val_i}\subst{x}{\expr_i}}}{\etrue}.

  \item $\typ \doteq \tfun{y}{\typ'_y}{\typ'}$.
  We need to show that
  if \relates{\rmodel, \rmodelBind{\val_{x_1}}{\val_{x_2}}{x}}{\val_1}{\val_2}{\tfun{y}{\typ'_y}{\typ'}},
  then \relates{\rmodel}{\val_1}{\val_2}{\tfun{y}{\typ'_y}{\typ'}\subst{x}{\expr_x}}.

  We fix $\val_{y_1}$ and $\val_{y_2}$ so that 
  $\relates{\rmodel, \rmodelBind{\val_{x_1}}{\val_{x_2}}{x}}{\val_{y_1}}{\val_{y_2}}{\typ'_y}$.
  
  Then, we have that 
  $\relates{\rmodel, \rmodelBind{\val_{x_1}}{\val_{x_2}}{x}, \rmodelBind{\val_{y_1}}{\val_{y_2}}{y}}{\val_1\ \val_{y_1}}{\val_2\ \val_{y_2}}{\typ'}$.

  By inductive hypothesis, we have that 
  $\relates{\rmodel, \rmodelBind{\val_{y_1}}{\val_{y_2}}{y}}{\val_1\ \val_{y_1}}{\val_2\ \val_{y_2}}{\typ'\subst{x}{\expr_x}}$.

  By inductive hypothesis on the fixed arguments, we also get 
  $\relates{\rmodel}{\val_{y_1}}{\val_{y_2}}{\typ'_y\subst{x}{\expr_x}}$.

  Combined, we get \relates{\rmodel}{\val_1}{\val_2}{\tfun{y}{\typ'_y}{\typ'}\subst{x}{\expr_x}}.

  \item $\typ \doteq \eqrt{\typ'}{\expr_l}{\expr_r}$.
  We need to show that
  if \relates{\rmodel, \rmodelBind{\val_{x_1}}{\val_{x_2}}{x}}{\val_1}{\val_2}{\eqrt{\typ'}{\expr_l}{\expr_r}},
  then \relates{\rmodel}{\val_1}{\val_2}{\eqrt{\typ'}{\expr_l}{\expr_r}\subst{x}{\expr_x}}.
  
  This reduces to showing that 
  if   \relates{\rmodel}{\modelapp{\delta_1, \subst{x}{\val_{x_1}}}{\expr_l}}{\modelapp{\delta_2, \subst{x}{\val_{x_2}}}{\expr_r}}{\typ'}, 
  then \relates{\rmodel}{\modelapp{\delta_1}{\expr_l\subst{x}{\expr_x}}}{\modelapp{\delta_2}{\expr_r\subst{x}{\expr_x}}}{\typ'}; 
  we find
  \parredsto{\modelapp{\delta_1}{\expr_l\subst{x}{\expr_x}}}{\modelapp{\delta_1, \subst{x}{\val_{x_1}}}{\expr_l}}
  and
  \parredsto{\modelapp{\delta_2}{\expr_r\subst{x}{\expr_x}}}{\modelapp{\delta_2, \subst{x}{\val_{x_2}}}{\expr_r}}
  because substituting multiple parallel reduction is parallel reduction (Corollary~\ref{cor:parred-subst-eval-multi}).
  The logical relation is closed under parallel reduction (Lemma~\ref{proofs-lemma:lr-parred}), and so
  \relates{\rmodel}{\modelapp{\delta_1}{\expr_l\subst{x}{\expr_x}}}{\modelapp{\delta_2}{\expr_r\subst{x}{\expr_x}}}{\typ'}.
\end{itemize}
\end{proof}

\begin{theorem}[LR Fundamental Property]\label{thm:fundamental}
    If \hastype{\env}{\expr}{\typ}, then 
    \relatesEnv{\env}{\expr}{\expr}{\typ}.  
  \end{theorem}
  
  \begin{proof}
  The proof goes by induction on the derivation tree: 
  \begin{itemize}
    \item[\tsub]
    By inversion of the rule we have 
    \begin{enumerate}
      \item \hastype{\env}{\expr}{\typ'} 
      \item \issubtype{\env}{\typ'}{\typ}
    \end{enumerate} 
    By IH on (1) we have 
    \begin{enumerate}
      \setcounter{enumi}{2}
      \item \relatesEnv{\env}{\expr}{\expr}{\typ'}
    \end{enumerate} 
  
    By (3), (4), and Lemma~\ref{proofs-lemma:sub} we have 
    \relatesEnv{\env}{\expr}{\expr}{\typ}.
  
    \item[\tcon] By Lemma~\ref{proofs-lemma:const}.
      
    \item[\tself] By inversion of the rule, we have:
      \begin{enumerate}
      \item \label{fund:self-inv} \subitem \hastype{\env}{\expr}{\tref{\vv}{\tbase}{\refa}}.
      \item By the IH on (\ref{fund:self-inv}), we have:
        \subitem \relatesEnv{\env}{\expr}{\expr}{\tref{\vv}{\tbase}{\refa}}.
      \item We fix a $\rmodel$ such that:
        \subitem $\rmodel \in \env$ and
        \subitem \relates{\rmodel}{\modelapp{\rmodel_1}{\expr}}{\modelapp{\rmodel_2}{\expr}}{\tref{\vv}{\tbase}{\refa}}
    \item \label{fund:self-red-val} There must exist $\val_1$ and $\val_2$ such that:
      \subitem \goesto{\modelapp{\rmodel_1}{\expr}}{\val_1}
      \subitem \goesto{\modelapp{\rmodel_2}{\expr}}{\val_2}
      \subitem \relates{\rmodel}{\val_1}{\val_2}{\tref{\vv}{\tbase}{\refa}}
    \item \label{fund:self-red-con} By definition, $\val_1 = \val_2 = c$ such that:
      \subitem \hasbtype{}{\con}{\tbase}
      \subitem \goesto{\modelapp{\rmodel_1}{\refa\subst{\rbind}{\con}}}{\etrue} 
      \subitem \goesto{\modelapp{\rmodel_2}{\refa\subst{\rbind}{\con}}}{\etrue}
    \item We find \relates{\rmodel}{\val_1}{\val_2}{\tref{\vv}{\tbase}{\bbbeq{\vv}{\expr}{\tbase}}}, because:
      \subitem \hasbtype{}{\con}{\tbase} by (\ref{fund:self-red-con})
      \subitem \goesto{\modelapp{\rmodel_1}{(\bbbeq{\vv}{\expr}{\tbase})\subst{\vv}{\con}}}{\etrue} because \goesto{\modelapp{\rmodel_1}{\expr}}{\val_1 = \con} by (\ref{fund:self-red-val})
      \subitem \goesto{\modelapp{\rmodel_2}{(\bbbeq{\vv}{\expr}{\tbase})\subst{\vv}{\con}}}{\etrue} because \goesto{\modelapp{\rmodel_2}{\expr}}{\val_2 = \con} by (\ref{fund:self-red-val})
    \end{enumerate}

    \item[\tvar] By inversion of the rule and Lemma~\ref{proofs-lemma:var}.
    \item[\tlam] By hypothesis: 
    \begin{enumerate}
      \item   \hastype{\env}{\elam{x}{\typ_x}{\expr}}{\tfun{x}{\typ_x}{\typ}} 
    \end{enumerate}
    By inversion of the rule we have 
    \begin{enumerate}
      \setcounter{enumi}{1}
      \item \hastype{\env, \envBind{x}{\typ_x}}{\expr}{\typ} 
      \item \iswellformed{\env}{\typ_x}
    \end{enumerate} 
    By inductive hypothesis on (2) we have 
    \begin{enumerate}
      \setcounter{enumi}{3}
      \item \relatesEnv{\env, \envBind{x}{\typ_x}}{\expr}{\expr}{\typ} 
    \end{enumerate} 
  
    We fix a $\rmodel$, $\val_{x_1}$, and $\val_{x_2}$ so that  
    \begin{enumerate}
      \setcounter{enumi}{4}
      \item $\rmodel \in \env$
      \item \relates{\rmodel}{\val_{x_1}}{\val_{x_2}}{\typ_x} 
    \end{enumerate} 
    Let $\rmodel' \doteq \rmodel,\rmodelBind{\val_{x_1}}{\val_{x_2}}{x}$.
  
    By the definition of the logical relation on open terms, 
    (4), (5), and (6) we have 
      
    \begin{enumerate}
      \setcounter{enumi}{6}
      \item \relatesModel{\rmodel'}{\expr}{\expr}{\typ} 
    \end{enumerate} 
    By the definition of substitution
    \begin{enumerate}
      \setcounter{enumi}{7}
      \item \relates{\rmodel'}{\modelapp{\rmodel_1}{\expr\subst{x}{\val_{x_1}}}}{\modelapp{\rmodel_2}{\expr\subst{x}{\val_{x_2}}}}{\typ} 
    \end{enumerate} 
    By the definition of the logical relation on closed expressions
    \begin{enumerate}
      \setcounter{enumi}{8}
      \item 
      \goesto{\modelapp{\rmodel_1}{\expr\subst{x}{\val_{x_1}}}}{\val_1}, 
      \goesto{\modelapp{\rmodel_2}{\expr\subst{x}{\val_{x_2}}}}{\val_2},
      and \relates{\rmodel'}{\val_1}{\val_2}{\typ} 
    \end{enumerate}
    By the definition and determinism of operational semantics  
    \begin{enumerate}
      \setcounter{enumi}{9}
      \item 
      \goesto{\modelapp{\rmodel_1}{(\elam{x}{\typ_x}{\expr})\ \val_{x_1}}}{\val_1}, 
      \goesto{\modelapp{\rmodel_2}{(\elam{x}{\typ_x}{\expr})\ \val_{x_2}}}{\val_2}, 
      and \relates{\rmodel'}{\val_1}{\val_2}{\typ} 
    \end{enumerate}
  
    By (6) and the definition of logical relation on function values,
    \begin{enumerate}
      \setcounter{enumi}{10}
      \item \relatesModel{\rmodel}{\elam{x}{\typ_x}{\expr}}{\elam{x}{\typ_x}{\expr}}{\tfun{x}{\typ_x}{\typ}} 
    \end{enumerate} 
    Thus, by the definition of the logical relation,
    \relatesEnv{\env}{\elam{x}{\typ_x}{\expr}}{\elam{x}{\typ_x}{\expr}}{\tfun{x}{\typ_x}{\typ}} 
  
    \item[\tapp]
    By hypothesis: 
    
    \begin{enumerate}
      \item \hastype{\env}{\expr\ \expr_x}{\typ\subst{x}{\expr_x}}
    \end{enumerate}
    
    By inversion we get 
    \begin{enumerate}
      \setcounter{enumi}{1}
      \item \hastype{\env}{\expr}{\tfun{x}{\typ_x}{\typ}}
      \item \hastype{\env}{\expr_x}{\typ_x}
    \end{enumerate} 
          
    By inductive hypothesis
    \begin{enumerate}
      \setcounter{enumi}{2}
      \item \relatesEnv{\env}{\expr}{\expr}{\tfun{x}{\typ_x}{\typ}}
      \item \relatesEnv{\env}{\expr_x}{\expr_x}{\typ_x}
    \end{enumerate} 
    We fix a $\rmodel \in \env$. Then, by the definition of the logical relation on open terms
    \begin{enumerate}
      \setcounter{enumi}{4}
      \item \relatesModel{\rmodel}{\expr}{\expr}{(\tfun{x}{\typ_x}{\typ})}
      \item \relatesModel{\rmodel}{\expr_x}{\expr_x}{\typ_x}
    \end{enumerate} 
    By the definition of the logical relation on open terms:
    \begin{enumerate}
      \setcounter{enumi}{6}
      \item \goesto{\modelapp{\rmodel_1}{\expr}}{\val_1}
      \item \goesto{\modelapp{\rmodel_2}{\expr}}{\val_2}
      \item \relates{\rmodel}{\val_1}{\val_2}{\tfun{x}{\typ_x}{\typ}}
      \item \goesto{\modelapp{\rmodel_1}{\expr_x}}{\val_{x_1}}
      \item \goesto{\modelapp{\rmodel_2}{\expr_x}}{\val_{x_2}}
      \item \relates{\rmodel}{\val_{x_1}}{\val_{x_2}}{\typ_x}
    \end{enumerate}
  
  By (7) and (10)  
  \begin{enumerate}
    \setcounter{enumi}{12}
    \item \goesto{\modelapp{\rmodel_1}{\expr\ \expr_x}}{\val_1\ \val_{x_1}}
  \end{enumerate}

  By (8) and (11)  
  \begin{enumerate}
    \setcounter{enumi}{13}
    \item \goesto{\modelapp{\rmodel_2}{\expr\ \expr_x}}{\val_2\ \val_{x_2}}
  \end{enumerate}

  By (9), (12), and the definition of logical relation on functions: 
  \begin{enumerate}
    \setcounter{enumi}{14}
    \item \relates{\rmodel,\rmodelBind{\val_{x_1}}{\val_{x_2}}{x}}{\val_1\ \val_{x_1}}{\val_2\ \val_{x_2}}{\typ}
  \end{enumerate}
  
  By (13), (14), (15), and Lemma~\ref{proofs-lemma:app}
  
  \begin{enumerate}
    \setcounter{enumi}{15}
    \item \relates{\rmodel,\rmodelBind{\val_{x_1}}{\val_{x_2}}{x}}{\modelapp{\rmodel_1}{\expr\ \expr_x}}{\modelapp{\rmodel_2}{\expr\ \expr_x}}{\typ}
  \end{enumerate}
  
  By (10), (11), (16), and Lemma~\ref{proofs-lemma:apptyp}
  \begin{enumerate}
    \setcounter{enumi}{16}
    \item \relates{\rmodel}{\modelapp{\rmodel_1}{\expr\ \expr_x}}{\modelapp{\rmodel_2}{\expr\ \expr_x}}{\typ\subst{x}{\expr_x}}
  \end{enumerate}
  
    So from the definition of logical relations, 
    \relatesEnv{\env}{\expr\ \expr_x}{\expr\ \expr_x}{\typ\subst{x}{\expr_x}}.
    \item[\teqbase]
    By hypothesis: 
    \begin{enumerate}
      \item \hastype{\env}{\ebeq{\tbase}\ \expr_l\ \expr_r\ \expr}{\eqrt{\tbase}{\expr_l}{\expr_r}}
    \end{enumerate}
    By inversion of the rule: 
    \begin{enumerate}
      \setcounter{enumi}{1}
      \item \hastype{\env}{\expr_l}{\typ_r}
      \item \hastype{\env}{\expr_r}{\typ_l} 
      \item \issubtype{\env}{\typ_r}{\tbase}  
      \item \issubtype{\env}{\typ_l}{\tbase} 
      \item \hastype{\env, \envBind{r}{\typ_r}, \envBind{l}{\typ_l}}{e}{\tref{\rbind}{\tunit}{\bbbeq{l}{r}{\tbase}}}
    \end{enumerate}
    By inductive hypothesis on (2), (3), and (6) we have 
  
    \begin{enumerate}
      \setcounter{enumi}{6}
      \item \relatesEnv{\env}{\expr_l}{\expr_l}{\typ_r}
      \item \relatesEnv{\env}{\expr_r}{\expr_r}{\typ_l} 
      \item \relatesEnv{\env, \envBind{r}{\typ_r}, \envBind{l}{\typ_l}}{e}{e}{\tref{\rbind}{\tunit}{\bbbeq{l}{r}{\tbase}}}
    \end{enumerate}
  
    We fix $\rmodel \in \env$. Then (7) and (8) become 
  
    \begin{enumerate}
      \setcounter{enumi}{9}
      \item \relatesModel{\rmodel}{\expr_l}{\expr_l}{\typ_r}
      \item \relatesModel{\rmodel}{\expr_r}{\expr_r}{\typ_l} 
    \end{enumerate}
    By the definition of the logical relation on closed terms:
    \begin{enumerate}
      \setcounter{enumi}{11}
      \item \goesto{\modelapp{\rmodel_1}{\expr_l}}{\val_{l_1}}
      \item \goesto{\modelapp{\rmodel_2}{\expr_l}}{\val_{l_2}}
      \item \relates{\rmodel}{\val_{l_1}}{\val_{l_2}}{\typ_l}
      \item \goesto{\modelapp{\rmodel_1}{\expr_r}}{\val_{r_1}}
      \item \goesto{\modelapp{\rmodel_2}{\expr_r}}{\val_{r_2}}
      \item \relates{\rmodel}{\val_{r_1}}{\val_{r_2}}{\typ_r}
    \end{enumerate}
    We define $\rmodel' \doteq \rmodel,\rmodelBind{\val_{r_1}}{\val_{r_2}}{r},\rmodelBind{\val_{l_1}}{\val_{l_2}}{l}$.
  
    By (9), (14), and (17) we have 
  
    \begin{enumerate}
      \setcounter{enumi}{17}
      \item \relatesModel{\rmodel'}{\expr}{\expr}{\tref{\rbind}{\tunit}{\bbbeq{l}{r}{\tbase}}}
    \end{enumerate}
  
    By the definition of the logical relation on closed terms: 
  
    \begin{enumerate}
      \setcounter{enumi}{18}
      \item \goesto{\modelapp{\rmodel'}{\expr}}{\val_1}
      \item \goesto{\modelapp{\rmodel'}{\expr}}{\val_2}
      \item \relates{\rmodel'}{\val_1}{\val_2}{\tref{\rbind}{\tunit}{\bbbeq{l}{r}{\tbase}}}
    \end{enumerate}
  
    By (21) and the definition of logical relation on basic values: 
  
  \begin{enumerate}
      \setcounter{enumi}{18}
      \item \goesto{\modelapp{\rmodel'_1}{(\bbbeq{l}{r}{\tbase}})}{\etrue}
      \item \goesto{\modelapp{\rmodel'_2}{(\bbbeq{l}{r}{\tbase}})}{\etrue}
  \end{enumerate}
  
  By the definition of $\bbbeqsym{\tbase}$
  
    \begin{enumerate}
      \setcounter{enumi}{20}
      \item $\val_{l_1} = \val_{r_1}$
      \item $\val_{l_2} = \val_{r_2}$
    \end{enumerate}
   By (14) and (17) and since $\typ_l$ and $\typ_r$ are basic types 
   \begin{enumerate}
    \setcounter{enumi}{22}
    \item $\val_{l_1} = \val_{l_2}$
    \item $\val_{r_1} = \val_{r_2}$
  \end{enumerate}
  By (21) and (24)
  \begin{enumerate}
    \setcounter{enumi}{24}
    \item $\val_{l_1} = \val_{r_2}$
  \end{enumerate}
  By the definition of the logical relation on basic types
  \begin{enumerate}
    \setcounter{enumi}{25}
    \item \relates{\delta}{\val_{l_1}}{\val_{r_2}}{\tbase}
  \end{enumerate}
  By which, (12), (16), and Lemma~\ref{proofs-lemma:app}
  \begin{enumerate}
    \setcounter{enumi}{26}
    \item \relates{\rmodel}{\modelapp{\rmodel_1}{\expr_l}}{\modelapp{\rmodel_2}{\expr_r}}{\tbase}
  \end{enumerate}
  
  By (12), (15), and (19)
  \begin{enumerate}
    \setcounter{enumi}{27}
    \item \goesto{\modelapp{\rmodel_1}{\ebeq{\tbase}\ \expr_l\ \expr_r\ \expr}}{\ebeq{\tbase}\ \val_{l_1}\ \val_{r_1}\ \val_1}
  \end{enumerate}
  
  By (13), (16), and (20)
  \begin{enumerate}
    \setcounter{enumi}{28}
    \item \goesto{\modelapp{\rmodel_2}{\ebeq{\tbase}\ \expr_l\ \expr_r\ \expr}}{\ebeq{\tbase}\ \val_{l_2}\ \val_{r_2}\ \val_2}
  \end{enumerate}
  
  By (27) and the definition of the logical relation on $\texttt{EqRT}$
    \begin{enumerate}
      \setcounter{enumi}{29}
      \item \relates{\rmodel}{\ebeq{\tbase}\ \val_{l_1}\ \val_{r_1}\ \val_1}{\ebeq{\tbase}\ \val_{l_2}\ \val_{r_2}\ \val_2}{\eqrt{\tbase}{\expr_l}{\expr_r}}.
    \end{enumerate}
  By (28), (29), and (30)
    \begin{enumerate}
      \setcounter{enumi}{30}
      \item \relatesModel{\rmodel}{\ebeq{\tbase}\ \expr_l\ \expr_r\ \expr}{\ebeq{\tbase}\ \expr_l\ \expr_r\ \expr}{\eqrt{\tbase}{\expr_l}{\expr_r}}.
    \end{enumerate}
    So, by the definition on the logical relation, 
    \relatesEnv{\env}{\ebeq{\tbase}\ \expr_l\ \expr_r\ \expr}{\ebeq{\tbase}\ \expr_l\ \expr_r\ \expr}{\eqrt{\tbase}{\expr_l}{\expr_r}}.
    \item[\teqfun]
    By hypothesis
    \begin{enumerate}
      \item \hastype{\env}{\exeq{\typ_x}{\typ}\ \expr_l\ \expr_r\ \expr}{\eqrt{\tfun{x}{\typ_x}{\typ}}{\expr_l}{\expr_r}}
    \end{enumerate}    
    By inversion of the rule 
    \begin{enumerate}
      \setcounter{enumi}{1}
      \item  \hastype{\env}{\expr_l}{\typ_r} 
      \item \hastype{\env}{\expr_r}{\typ_l}
      \item \issubtype{\env}{\typ_r}{\tfun{x}{\typ_x}{\typ}} 
      \item \issubtype{\env}{\typ_l}{\tfun{x}{\typ_x}{\typ}}  
      \item \hastype{\env, \envBind{r}{\typ_r}, \envBind{l}{\typ_l}}{e}{(\tfun{x}{\typ_x}{\eqrt{\typ}{l\ x}{r\ x}})} 
      \item  \iswellformed{\env}{\tfun{x}{\typ_x}{\typ}}
    \end{enumerate}    
  
  By inductive hypothesis on (2), (3), and (6) we have 
  
  \begin{enumerate}
    \setcounter{enumi}{7}
    \item \relatesEnv{\env}{\expr_l}{\expr_l}{\typ_r}
    \item \relatesEnv{\env}{\expr_r}{\expr_r}{\typ_l} 
    \item \relatesEnv{\env, \envBind{r}{\typ_r}, \envBind{l}{\typ_l}}{e}{e}{(\tfun{x}{\typ_x}{\eqrt{\typ}{l\ x}{r\ x}})}
  \end{enumerate}
  
  By (8), (9), and Lemma~\ref{proofs-lemma:sub}
  \begin{enumerate}
    \setcounter{enumi}{10}
    \item \relatesEnv{\env}{\expr_l}{\expr_l}{\tfun{x}{\typ_x}{\typ}}
    \item \relatesEnv{\env}{\expr_r}{\expr_r}{\tfun{x}{\typ_x}{\typ}} 
  \end{enumerate}

  We fix $\rmodel \in \env$.
  Then (11), and (12) become 
  
  \begin{enumerate}
    \setcounter{enumi}{12}
    \item \relatesModel{\rmodel}{\expr_l}{\expr_l}{\tfun{x}{\typ_x}{\typ}}
    \item \relatesModel{\rmodel}{\expr_r}{\expr_r}{\tfun{x}{\typ_x}{\typ}} 
  \end{enumerate}
  
  By the definition of the logical relation on closed terms:
  \begin{enumerate}
    \setcounter{enumi}{14}
    \item \goesto{\modelapp{\rmodel_1}{\expr_l}}{\val_{l_1}}
    \item \goesto{\modelapp{\rmodel_2}{\expr_l}}{\val_{l_2}}
    \item \relates{\rmodel}{\val_{l_1}}{\val_{l_2}}{\tfun{x}{\typ_x}{\typ}}
    \item \relates{\rmodel}{\val_{l_1}}{\val_{l_2}}{\typ_l}
    \item \goesto{\modelapp{\rmodel_1}{\expr_r}}{\val_{r_1}}
    \item \goesto{\modelapp{\rmodel_2}{\expr_r}}{\val_{r_2}}
    \item \relates{\rmodel}{\val_{r_1}}{\val_{r_2}}{\tfun{x}{\typ_x}{\typ}}
    \item \relates{\rmodel}{\val_{r_1}}{\val_{r_2}}{\typ_r}
  \end{enumerate}
  We fix $\val_{x_1}$ and $\val_{x_2}$ so that \relates{\rmodel}{\val_{x_1}}{\val_{x_2}}{\typ_x}.
  Let $\rmodel_x \doteq \rmodel,\rmodelBind{\val_{x_1}}{\val_{x_2}}{x}$.

  By the definition on the logical relation on function values, 
  (17) and (21) become 
  \begin{enumerate}
    \setcounter{enumi}{22}
    \item \relates{\rmodel_x}{\val_{l_1}\ \val_{x_1}}{\val_{l_2}\ \val_{x_2}}{\typ}
    \item \relates{\rmodel_x}{\val_{r_1}\ \val_{x_1}}{\val_{r_2}\ \val_{x_2}}{\typ}
  \end{enumerate}
  
  Let $\rmodel_{lr} \doteq \rmodel,\rmodelBind{\val_{r_1}}{\val_{r_2}}{r},\rmodelBind{\val_{l_1}}{\val_{l_2}}{l}$.

  By the definition of the logical relation on closed terms, (10) becomes: 
  
  \begin{enumerate}
    \setcounter{enumi}{24}
    \item \goesto{\modelapp{\rmodel_{lr}}{\expr}}{\val_1}
    \item \goesto{\modelapp{\rmodel_{lr}}{\expr}}{\val_2}
    \item \relates{\rmodel_{lr}}{\val_1}{\val_2}{\tfun{x}{\typ_x}{\eqrt{\typ}{l\ x}{r\ x}}}
  \end{enumerate}

  By (27) and the definition of logical relation on function values: 
  
  \begin{enumerate}
    \setcounter{enumi}{27}
    \item \relates{\rmodel_{lr},\rmodelBind{\val_{x_1}}{\val_{x_2}}{x}}{\val_1\ \val_{x_1}}{\val_2\ \val_{x_2}}{\eqrt{\typ}{l\ x}{r\ x}}
  \end{enumerate}
  
  By the definition of the logical relation on $\texttt{EqRT}$

  \begin{enumerate}
    \setcounter{enumi}{28}
    \item \relates{\rmodel_{lr},\rmodelBind{\val_{x_1}}{\val_{x_2}}{x}}{\val_{l_1}\ \val_{x_1}}{\val_{r_2}\ \val_{x_2}}{\typ}
  \end{enumerate}
  
  By the definition of logical relations on function values
  
  \begin{enumerate}
    \setcounter{enumi}{29}
    \item \relates{\rmodel_{lr}}{\val_{l_1}}{\val_{r_2}}{\tfun{x}{\typ_x}{\typ}}
  \end{enumerate}
  
  By (7), $l$ and $r$ do not appear free in the relation, so 
  
  \begin{enumerate}
    \setcounter{enumi}{30}
    \item \relates{\rmodel}{\val_{l_1}}{\val_{r_2}}{\tfun{x}{\typ_x}{\typ}}
  \end{enumerate}
  
  By which, (15), (20), and Lemma~\ref{proofs-lemma:app}
  
  \begin{enumerate}
    \setcounter{enumi}{31}
    \item \relates{\rmodel}{\modelapp{\rmodel_1}{\expr_l}}{\modelapp{\rmodel_2}{\expr_r}}{\tfun{x}{\typ_x}{\typ}}
  \end{enumerate}
  
  By (15), (19), and (25)
  \begin{enumerate}
  \setcounter{enumi}{32}
  \item \goesto{\modelapp{\rmodel_1}{\exeq{\typ_x}{\typ}\ \expr_l\ \expr_r\ \expr}}{\exeq{\typ_x}{\typ}\ \val_{l_1}\ \val_{r_1}\ \val_1}
  \end{enumerate}
  
  By (16), (20), and (26)
  \begin{enumerate}
  \setcounter{enumi}{33}
  \item \goesto{\modelapp{\rmodel_2}{\exeq{\typ_x}{\typ}\ \expr_l\ \expr_r\ \expr}}{\exeq{\typ_x}{\typ}\ \val_{l_2}\ \val_{r_2}\ \val_2}
  \end{enumerate}
  
  By (32) and the definition of the logical relation on $\texttt{EqRT}$
  \begin{enumerate}
    \setcounter{enumi}{34}
    \item \relates{\rmodel}{\exeq{\typ_x}{\typ}\ \val_{l_1}\ \val_{r_1}\ \val_1}{\exeq{\typ_x}{\typ}\ \val_{l_2}\ \val_{r_2}\ \val_2}{\eqrt{\tfun{x}{\typ_x}{\typ}}{\expr_l}{\expr_r}}.
  \end{enumerate}
  By (33), (34), and (35)
  \begin{enumerate}
    \setcounter{enumi}{35}
    \item \relatesModel{\rmodel}{\exeq{\typ_x}{\typ}\ \expr_l\ \expr_r\ \expr}{\exeq{\typ_x}{\typ}\ \expr_l\ \expr_r\ \expr}{\eqrt{\tfun{x}{\typ_x}{\typ}}{\expr_l}{\expr_r}}.
  \end{enumerate}
  So, by the definition on the logical relation, 
  \relatesEnv{\env}{\exeq{\typ_x}{\typ}\ \expr_l\ \expr_r\ \expr}{\exeq{\typ_x}{\typ}\ \expr_l\ \expr_r\ \expr}{\eqrt{\tfun{x}{\typ_x}{\typ}}{\expr_l}{\expr_r}}.
  
  \end{itemize}
  \end{proof}

\subsection{The Logical Relation and the $\mathsf{EqRT}$ Type are Equivalence Relations}

\begin{theorem}[The logical relation is an equivalence relation]
    \label{proofs:equality-logical-relation}
    $\relatesEnv{\env}{\expr_1}{\expr_2}{\typ}$ is reflexive, symmetric, and transivite.
    \begin{itemize}
      \item \textit{Reflexivity:} If \hastype{\env}{\expr}{\typ}, then $\relatesEnv{\env}{\expr}{\expr}{\typ}$.
      \item \textit{Symmetry:} If $\relatesEnv{\env}{\expr_1}{\expr_2}{\typ}$, then $\relatesEnv{\env}{\expr_2}{\expr_1}{\typ}$. 
      \item \textit{Transitivity:} If $\hastype{\env}{\expr_2}{\typ}$ and $\relatesEnv{\env}{\expr_1}{\expr_2}{\typ}$ and $\relatesEnv{\env}{\expr_2}{\expr_3}{\typ}$, 
                                   then $\relatesEnv{\env}{\expr_1}{\expr_3}{\typ}$. 
\mmg{we might be able to relax this assumption, proving a property like: if \relatesEnv{\env}{\expr_1}{\expr_2}{\typ} then \relates{\env}{\expr_2}{\expr_2}{\typ}.}
    \end{itemize}
  \end{theorem}
\begin{proof}
\textbf{Reflexivity:} This is exactly the Fundamental Property~\ref{thm:fundamental}. 

\textbf{Symmetry:}
Let $\bar{\rmodel}$ be defined such that $\bar{\rmodel}_1(x) = \rmodel_2(x)$ and $\bar{\rmodel}_2(x) = \rmodel_1(x)$.
First, we prove that
$\relates{\rmodel}{v_1}{v_2}{\typ}$ implies
$\relates{\bar{\rmodel}}{v_2}{v_1}{\typ}$, by structural induction on $\typ$.
\begin{itemize}    
    \item $\typ \doteq \tref{\vv}{\tbase}{\refa}$. 
    This case is immediate: we have to show that
    $\relates{\bar{\rmodel}}{c}{c}{\tref{\vv}{\tbase}{\refa}}$ given
    $\relates{\rmodel}{c}{c}{\tref{\vv}{\tbase}{\refa}}$.
    But the definition in this case is itself symmetric: the predicate
    goes to $\etrue$ under both substitutions.

    \item $\typ \doteq \tfun{x}{\typ'_x}{\typ'}$.
    We fix $\val_{x_1}$ and $\val_{x_2}$ so that
    \begin{enumerate}
        \item \relates{\rmodel}{\val_{x_1}}{\val_{x_2}}{\typ'_x}
    \end{enumerate}
    By the definition of logical relations on open terms and 
    inductive hypothesis
    \begin{enumerate}
        \setcounter{enumi}{1}
        \item \relates{\bar{\rmodel}}{\val_{x_2}}{\val_{x_1}}{\typ'_x}
    \end{enumerate}
    By the definition on logical relations on functions 
    \begin{enumerate}
        \setcounter{enumi}{2}
        \item \relates{\rmodel,\rmodelBind{\val_{x_1}}{\val_{x_2}}{x}}{\val_1\ \val_{x_1}}{\val_2\ \val_{x_2}}{\typ'}
    \end{enumerate}
    By the definition of logical relations on open terms 
    and since the expressions $\val_1\ \val_{x_1}$ and $\val_2\ \val_{x_2}$ are closed,
    By the inductive hypothesis on $\typ'$:
    \begin{enumerate}
        \setcounter{enumi}{3}
        \item \relates{\bar{\rmodel},\envBind{x}{\typ'_x}}{\val_2\ \val_{x_2}}{\val_1\ \val_{x_1}}{\typ'}
    \end{enumerate}
    By (2) and the definition of logical relations on open terms
    \begin{enumerate}
        \setcounter{enumi}{4}
        \item \relates{\bar{\rmodel},\rmodelBind{\val_{x_2}}{\val_{x_1}}{x}}{\val_2\ \val_{x_2}}{\val_1\ \val_{x_1}}{\typ'}
    \end{enumerate}
    By the definition of the logical relation on functions, we conclude 
    that  
    \relates{\bar{\rmodel}}{\val_2}{\val_1}{\tfun{x}{\typ'_x}{\typ'}}

    \item $\typ \doteq \eqrt{\typ'}{\expr_l}{\expr_r}$.
    By assumption,
    \begin{enumerate}
        \item \relates{\rmodel}{\val_1}{\val_2}{\eqrt{\typ'}{\expr_l}{\expr_r}}
    \end{enumerate}

    By the definition of the logical relation on $\texttt{EqRT}$ types
    \begin{enumerate}
        \setcounter{enumi}{1}
        \item \relates{\rmodel}{\modelapp{\rmodel_1}{\expr_l}}{\modelapp{\rmodel_2}{\expr_r}}{\typ'}
    \end{enumerate}
    i.e., $\goesto{\modelapp{\rmodel_1}(\expr_l)}{\val_l}$ and
    similarly for $\val_r$ such that
    $\relates{\rmodel}{\val_l}{\val_r}{\typ'}$.

    By the IH on $\typ'$, we have:    
    \begin{enumerate}
        \setcounter{enumi}{2}
      \item \relates{\bar{\rmodel}}{\val_r}{\val_l}{\typ'}
    \end{enumerate}

    And so, by the definition of the LR on equality proofs:
    
    \begin{enumerate}
        \setcounter{enumi}{3}
        \item \relates{\bar{\rmodel}}{\val_2}{\val_1}{\eqrt{\typ'}{\expr_l}{\expr_r}}
    \end{enumerate}
\end{itemize}

Next, we show that $\rmodel \in \Gamma$ implies $\bar{\rmodel} \in \Gamma$.
We go by structural induction on $\Gamma$.
\begin{itemize}
  \item $\Gamma = \cdot$. This case is trivial.
  \item $\Gamma = \Gamma',\envBind{x}{\typ}$. 
    For $\envBind{x}{\typ}$, we know that
    $\relates{\rmodel}{\rmodel_1(x)}{\rmodel_2(x)}{\typ}$. By
    the IH on $\typ$, we find
    $\relates{\bar{\rmodel}}{\rmodel_2(x)}{\rmodel_1(x)}{\typ}$,
    which is just the same as
    $\relates{\bar{\rmodel}}{\bar{\rmodel}_1(x)}{\bar{\rmodel}_2(x)}{\typ}$.
    By the IH on $\Gamma'$,
    we can use similar reasoning to find
    $\relates{\bar{\rmodel}}{\bar{\rmodel}_1(y)}{\bar{\rmodel}_2(y)}{\typ'}$
    for all $\envBind{y}{\typ'} \in \Gamma'$.
\end{itemize}

Now, suppose $\relatesEnv{\env}{\expr_1}{\expr_2}{\typ}$; we must show $\relatesEnv{\env}{\expr_2}{\expr_1}{\typ}$.
We fix $\rmodel\in\env$; we must show \relatesModel{\rmodel}{\expr_2}{\expr_1}{\typ}, i.e., there must exist $\val_1$ and $\val_2$ such that
\goesto{\modelapp{\rmodel_1}{\expr_2}}{\val_2} and
\goesto{\modelapp{\rmodel_2}{\expr_1}}{\val_1} and 
\relates{\rmodel}{\val_2}{\val_1}{\typ}. 
We have $\rmodel \in \Gamma$, and so $\bar{\rmodel} \in \Gamma$ by our second lemma. But
then, by assumption, we have $\val_1$ and $\val_2$ such that
$\goesto{\modelapp{\bar{\rmodel_1}}{\expr_1}}{\val_1}$ and
$\goesto{\modelapp{\bar{\rmodel_2}}{\expr_2}}{\val_2}$ and
$\relates{\bar{\rmodel}}{\val_1}{\val_2}{\typ}$.  Our first lemma then
yields $\relates{\rmodel}{\val_2}{\val_1} {\typ}$ as desired.

\textbf{Transitivity:} 
First, we prove an inner property: if $\rmodel\in\env$ and $\relates{\rmodel}{\val_1}{\val_2}{\typ}$ and 
$\relates{\rmodel}{\val_2}{\val_3}{\typ}$, then $\relates{\rmodel}{\val_1}{\val_3}{\typ}$.
We go by structural induction on the type index $\typ$.
\begin{itemize}
    \item $\typ \doteq \tref{\vv}{\tbase}{\refa}$.
    Here all of the values must be the fixed constant
    $c$. Furthermore, we must have
    $\goesto{\modelapp{\rmodel_1}{\refa\subst{\rbind}{\con}}}{\etrue}$
    and
    $\goesto{\modelapp{\rmodel_2}{\refa\subst{\rbind}{\con}}}{\etrue}$,
    so we can immediately find
    $\relates{\rmodel}{\val_1}{\val_3}{\typ}$.

    \item $\typ \doteq \tfun{x}{\typ'_x}{\typ'}$.
      
      Let $\relates{\rmodel}{\val_l}{\val_r}{\typ'_x}$ be given.
      We must show that $\relates{\rmodel,\rmodelBind{\val_l}{\val_r}{x}}{\val_1}{\val_3}{\tau}$.
      We know by assumption that:
$\relates{\rmodel,\rmodelBind{\val_l}{\val_r}{x}}{\val_1\ \val_l}{\val_2\ \val_r}{\tau'}$ and
$\relates{\rmodel,\rmodelBind{\val_l}{\val_r}{x}}{\val_2\ \val_l}{\val_3\ \val_r}{\tau'}$.
      By the IH on $\tau'$, we find
      $\relates{\rmodel,\rmodelBind{\val_l}{\val_r}{x}}{\val_1\ \val_l}{\val_3\ \val_r}{\tau'}$; 
      which gives 
      $\relates{\rmodel,\rmodelBind{\val_l}{\val_r}{x}}{\val_1}{\val_3}{\tau}$.

    \item $\typ \doteq \eqrt{\typ'}{\expr_l}{\expr_r}$.

      To find
      $\relates{\rmodel}{\val_1}{\val_3}{\eqrt{\typ}{\expr_l}{\expr_r}}$,
      we merely need to find that
      $\relates{\rmodel}{\modelapp{\rmodel_1}{\expr_l}}{\modelapp{\rmodel_2}{\expr_r}}{\typ}$,
      which we have by inversion on
      $\relates{\rmodel}{\val_1}{\val_2}{\eqrt{\typ}{\expr_l}{\expr_r}}$.
\end{itemize}

With our proof that the value relation is transitive in hand, we turn our attention to the open relation.
Suppose \relatesEnv{\env}{\expr_1}{\expr_2}{\typ} and
\relatesEnv{\env}{\expr_2}{\expr_3}{\typ}; we want to see
\relatesEnv{\env}{\expr_1}{\expr_3}{\typ}.
Let $\rmodel \in \env$ be given.
We have
\relatesModel{\rmodel}{\expr_1}{\expr_2}{\typ} and 
\relatesModel{\rmodel}{\expr_2}{\expr_3}{\typ}.
By the definition of the logical relations, we have 
\goesto{\modelapp{\rmodel_1}{\expr_1}}{\val_1}, 
\goesto{\modelapp{\rmodel_2}{\expr_2}}{\val_2},  
\goesto{\modelapp{\rmodel_1}{\expr_2}}{\val_2'},  
\goesto{\modelapp{\rmodel_2}{\expr_3}}{\val_3},  
\relates{\rmodel}{\val_1}{\val_2}{\typ}, and 
\relates{\rmodel}{\val_2'}{\val_3}{\typ}. 

Moreover, we know that $\expr_2$ is well typed, so by the fundamental
theorem (Theorem~\ref{thm:fundamental}), we know that
\relatesEnv{\env}{\expr_2}{\expr_2}{\typ}, and so
$\relates{\rmodel}{\val_2}{\val_2'}{\typ}$.

By our transitivity lemma on the value relation, we can find that
$\val_1$ is equivalent to $\val_2$ is equivalent to $\val_2'$ is
equivalent to $\val_3$, and so
$\relates{\rmodel}{\val_1}{\val_3}{\typ}$.

\end{proof}
  
\[ \begin{array}{rcl}
\pfty &:& \expr \rightarrow \expr \rightarrow \typ \\
\pfty(l, r, \tbase)                 &=& \tref{x}{\tunit}{\bbbeq{l}{r}{\tbase}} \\
\pfty(l, r, \tfun{x}{\typ_x}{\typ}) &=& \tfun{x}{\typ_x}{\eqrt{\typ}{l\ x}{r\ x}} \\
\end{array} \]

Our propositional equality \eqrt{\typ}{\expr_l}{\expr_r} is a
reflection of the logical relation, so it is unsurprising that it is
also an equivalence relation.
We can prove that our propositional equality is treated as an
equivalence relation by the syntactic type system.
There are some tiny wrinkles in the syntactic system: symmetry and
transitivity produce normalized proofs, but reflexivity produces
unnormalized ones in order to generate the correct invariant types
$\typ_l$ and $\typ_r$ in the base case.

\begin{theorem}[\texttt{EqRT} is an equivalence relation]
  \label{proofs:eqrt-equivalence}
  $\eqrt{\typ}{\expr_1}{\expr_2}$ is reflexive, symmetric, and
  transitive on equable types. That is, for all $\typ$ that contain
  only refinements and functions:
  \begin{itemize}
    \item \textit{Reflexivity:} If $\hastype{\env}{\expr}{\typ}$, then
      there exists $\expr_p$ such that
      $\hastype{\env}{\expr_p}{\eqrt{\typ}{\expr}{\expr}}$.

    \item \textit{Symmetry:} $\forall \env, \typ, \expr_1, \expr_2, \val_{12}$.
      if $\hastype{\env}{\val_{12}}{\eqrt{\typ}{\expr_1}{\expr_2}}$, then
      there exists $\val_{21}$ such that
      $\hastype{\env}{\val_{21}}{\eqrt{\typ}{\expr_2}{\expr_1}}$.

      \todo{we can change $\val_{12}$ to an expression when we correct issues with term indices}

    \item \textit{Transitivity:} 
      $\forall \env, \typ, \expr_1, \expr_2, \expr_3, \val_{12}, \val_{23}$.
      if  \hastype{\env}{\val_{12}}{\eqrt{\typ}{\expr_1}{\expr_2}}
      and \hastype{\env}{\val_{23}}{\eqrt{\typ}{\expr_2}{\expr_3}}, 
      then there exists $\val_{13}$ such that \hastype{\env}{\val_{13}}{\eqrt{\typ}{\expr_1}{\expr_3}}.

  \end{itemize}
\end{theorem}
\begin{proof}
\newcommand\epf{\ensuremath{\expr_{\mathsf{pf}}}}

\textbf{Reflexivity}: We strengthen the IH, simultaneously proving
that there exist $\expr_p, \epf$ and $\issubtype{\env}{\typ_l}{\typ}$ and $\issubtype{\env}{\typ_r}{\typ}$ such that
$\hastype{\env, \envBind{l}{\typ_l},
  \envBind{r}{\typ_r}}{\epf}{\pfty(\expr, \expr, \typ)}$ and
$\hastype{\env}{\expr_p}{\eqrt{\typ}{\expr}{\expr}}$ by induction on
$\typ$, leaving $\expr$ general.
\begin{itemize}
\item $\typ \doteq \tref{x}{\tbase}{\expr'}$. 
\begin{enumerate}
\item Let $\epf = ()$.
\item Let $\expr_p = \ebeq{\tbase}\ \expr\ \expr\ \epf$.
\item Let $\typ_l = \typ_r =
  \tref{x}{\tbase}{\bbbeq{x}{\expr}{\tbase}}$.
\item We have $\issubtype{\env}{\bbbeq{x}{\expr}{\tbase}}{\typ}$ by
  \subBase and semantic typing.
\item We find \hastype{\env}{\expr_p}{\eqrt{\tbase}{\expr}{\expr}} by
  \teqbase, with $\expr_l = \expr_r = \expr$. We must show:
  \begin{enumerate}
  \item\label{refl:expr} \hastype{\env}{\expr_l}{\typ_l} and
    \hastype{\env}{\expr_r}{\typ_r}, i.e.,
    \hastype{\env}{\expr}{\tref{x}{\tbase}{\bbbeq{x}{\expr}{\tbase}}};
  \item\label{refl:sub} \issubtype{\env}{\typ_r}{\tref{x}{\tbase}{\etrue}} and
    \issubtype{\env}{\typ_l}{\tref{x}{\tbase}{\etrue}}; and
  \item\label{refl:pf} \hastype{\env, \envBind{r}{\typ_r},
    \envBind{l}{\typ_l}}{\epf}{\tref{\rbind}{\tunit}{\bbbeq{l}{r}{\tbase}}}.
  \end{enumerate}
\item We find (\ref{refl:expr}) by \tself.
\item We find (\ref{refl:sub}) immediately by \subBase.
\item We find (\ref{refl:pf}) by \tvar, using \tsub to see that if
  $\envBind{l, r}{\tref{x}{\tbase}{\bbbeq{x}{\expr}{\tbase}}}$ then
  $\eunit$ will be typeable at the refinement where both $l$ and $r$ are equal to $\expr$.
\end{enumerate}

\item $\typ \doteq \tfun{x}{\typ_x}{\typ'}$.
\begin{enumerate}
\item \hastype{\env, \envBind{x}{\typ_x}}{\expr\ x}{\typ\subst{x}{x}} by \tapp and \tvar, noting that $\typ\subst{x}{x} = \typ$.
\item By the IH on $\hastype{\env, \envBind{x}{\typ_x}}{\expr\ x}{\typ'\subst{x}{x}} = \typ'$, there
  exist $\expr_p', \epf', \typ_l',$ and $\typ_r'$ such that:
  \begin{enumerate}
  \item \label{refl:ih-sub}
    \issubtype{\envBind{x}{\typ_x}}{\typ_l'}{\typ} and \issubtype{\envBind{x}{\typ_x}}{\typ_r'}{\typ};
  \item \label{refl:ih-pf}
   \hastype{\env, \envBind{x}{\typ_x}, \envBind{l}{\typ_l'}, \envBind{r}{\typ_r'}}
            {\epf'}{\pfty(\expr\ x, \expr\ x, \typ')}; and
  \item \label{refl:ih-refl}
    \hastype{\env, \envBind{x}{\typ_x}}{\expr_p'}{\eqrt{\typ'}{\expr\ x}{\expr\ x}}.
  \end{enumerate}
\item If $\typ' = \tref{x}{\tunit}{\typ'}{\expr\ x}{\expr\ x}$, then
  $\pfty(\expr\ x, \expr x, \tbase) = \tref{x}{\tunit}{\bbbeq{\expr
    x}{\expr x}{\tbase}}$; otherwise, $\pfty(l, r,
  \tfun{x}{\typ_x}{\typ}) =
  \tfun{x}{\typ_x}{\eqrt{\typ}{\expr\ x}{\expr\ x}}$.

  In the former case, let $\epf'' = \ebeq{\tbase}\ (\expr\ x) (\expr\ x) \epf'$.
  In the latter case, let $\epf'' = \epf'$.

  Either way, we have \hastype{\env, \envBind{x}{\typ_x},
    \envBind{l}{\typ_l'}, \envBind{r}{\typ_r'}}
  {\epf''}{\eqrt{\typ'}{\expr\ x}{\expr\ x}} by \teqbase or \teqfun,
  respectively.

\item Let $\epf = \tfun{x}{\typ_x}{\epf''}$.
\item Let $\expr_p = \exeq{x}{\typ_x}{\typ}\ \expr\ \expr\ \epf$.
\item Let $\expr_l = \expr_r = \expr$ and $\typ_l =
  \tfun{x}{\typ_x}{\typ_l'}$ and $\typ_r = \tfun{x}{\typ_x}{\typ_r'}$.
\item \label{refl:fun-sub-funty} We find subtyping by \subFun and (\ref{refl:ih-sub}).
\item By \teqfun. We must show:
  \begin{enumerate}
  \item \label{refl:fun-expr} \hastype{\env}{\expr_l}{\typ_l} and \hastype{\env}{\expr_r}{\typ_r};
  \item \label{refl:fun-sub} \issubtype{\env}{\typ_l}{\tfun{x}{\typ_x}{\typ}} and \issubtype{\env}{\typ_r}{\tfun{x}{\typ_x}{\typ}};
  \item \label{refl:fun-pf} \hastype{\env, \envBind{r}{\typ_r}, \envBind{l}{\typ_l}}{\epf}{(\tfun{x}{\typ_x}{\eqrt{\typ}{l\ x}{r\ x}})}
  \item \label{refl:fun-wf} \iswellformed{\env}{\tfun{x}{\typ_x}{\typ}}
  \end{enumerate}
\item We find (\ref{refl:fun-expr}) by assumption, \tsub, and (\ref{refl:fun-sub-funty}).
\item We find (\ref{refl:fun-sub}) by (\ref{refl:fun-sub-funty}).
\item We find (\ref{refl:fun-pf}) by \tlam and (\ref{refl:ih-pf}).
\end{enumerate}

\item $\typ \doteq \eqrt{\typ'}{\expr_1}{\expr_2}$. These types are not equable, so we ignore them.
\end{itemize}

\textbf{Symmetry}: 
By induction on $\typ$.
\begin{itemize}
\item $\typ \doteq \tref{x}{\tbase}{\expr}$. 
\begin{enumerate}
\item We have $\hastype{\env}{\val_{12}}{\eqrt{\tbase}{\expr_1}{\expr_2}}$.
\item \label{sym:base-inversion} By canonical forms, $\val_{12} =
  \ebeq{\tbase}\ \expr_l\ \expr_r\ \val_p$ such that
  \hastype{\env}{\expr_l}{\typ_l} and \hastype{\env}{\expr_r}{\typ_r}
  (for some $\typ_l$ and $\typ_r$ that are refinements of $\tbase$)
  and \hastype{\env, \envBind{r}{\typ_r},
    \envBind{l}{\typ_l}}{\val_p}{\tref{x}{\tunit}{\bbbeq{l}{r}{\tbase}}}
  (Lemma~\ref{lem:canonical-forms}).
\item Let $\val_{21} = \ebeq{\tbase}\ \expr_r\ \expr_l\ \val_p$.
\item By \teqbase, swapping $\typ_l$ and $\typ_r$ from
  (\ref{sym:base-inversion}). We already have appropriate typing and subtyping derivations; we only need to see
\hastype{\env, \envBind{l}{\typ_l},
    \envBind{r}{\typ_r}}{\val_p}{\tref{x}{\tunit}{\bbbeq{r}{l}{\tbase}}}.
\item We have \issubtype{\env, \envBind{l}{\typ_l},
  \envBind{r}{\typ_r}}{\tref{x}{\tunit}{\bbbeq{r}{l}{\tbase}}}{\tref{x}{\tunit}{\bbbeq{l}{r}{\tbase}}}
  by \subBase and symmetry of \bbbeq{(}{)}{\tbase}.
\end{enumerate}

\item $\typ \doteq \tfun{x}{\typ_x}{\typ'}$.
\begin{enumerate}
\item We have $\hastype{\env}{\val_{12}}{\eqrt{\tfun{x}{\typ_x}{\typ'}}{\expr_1}{\expr_2}}$.
\item \label{sym:fun-inversion} By canonical forms, $\val_{12} =
  \exeq{x}{\typ_x'}{\typ''}\ \expr_l\ \expr_r\ \val_p$ such that
  \issubtype{\typ_x}{\typ_x'} and \issubtype{\typ''}{\typ'} and
  \hastype{\env}{\expr_l}{\typ_l} and \hastype{\env}{\expr_r}{\typ_r}
  (for some $\typ_l$ and $\typ_r$ that are subtypes of
  \tfun{x}{\typ_x'}{\typ''}) and \hastype{\env, \envBind{r}{\typ_r},
    \envBind{l}{\typ_l}}{\val_p}{\tfun{x}{\typ_x'}{\eqrt{\typ''}{l\ x}{r\ x}}}.
\item \label{sym:fun-pf-inversion} By canonical forms, this time on
  $\val_p$ from (\ref{sym:fun-inversion}), $\val_p =
  \tlam{x}{\typ_x'}{\expr_p}$ such that
  \issubtype{\env}{\typ_x}{\typ_x'} and \hastype{\env,
    \envBind{r}{\typ_r}, \envBind{l}{\typ_l},
    \envBind{x}{\typ_x'}}{\expr}{\typ'''} such that
  \issubtype{\env,
    \envBind{r}{\typ_r}, \envBind{l}{\typ_l},
    \envBind{x}{\typ_x'}}{\typ'''}{\eqrt{\typ''}{l\ x}{r\ x}}.
\item \label{sym:fun-pf-ih} By \tsub, (\ref{sym:fun-pf-inversion}),
  and the IH on $\eqrt{\typ''}{l\ x}{r\ x}$, we know there
  exists some $\expr_p'$ such that \hastype{\env, \envBind{l}{\typ_l},
    \envBind{r}{\typ_r},
    \envBind{x}{\typ_x'}}{\expr_p'}{\eqrt{\typ''}{r\ x}{l\ x}}.
\item Let $\val_p' = \tfun{x}{\typ_x'}{\expr_p'}$.
\item \label{sym:fun-pf} By (\ref{sym:fun-pf-ih}) and \tlam, and \tsub (using subtyping
  from (\ref{sym:fun-pf-inversion}) and (\ref{sym:fun-inversion})),
\hastype{\env, \envBind{l}{\typ_l},
    \envBind{r}{\typ_r}}{\val_p'}{\eqrt{\tfun{x}{\typ_x}{\typ'}}{\expr_r\ x}{\expr_l\ x}}.
\item Let $\val_{21} = \exeq{x}{\typ_x}{\typ'}\ \expr_r\ \expr_l\ \val_p'$.
\item By \teqbase, with (\ref{sym:fun-pf}) for the proof and
  (\ref{sym:fun-pf-inversion}) and (\ref{sym:fun-inversion}) for the
  rest.
\end{enumerate}

\item $\typ \doteq \eqrt{\typ'}{\expr_1}{\expr_2}$. These types are not equable, so we ignore them.
\end{itemize}

\textbf{Transitivity}: By induction on $\typ$.
\begin{itemize}
\item $\typ \doteq \tref{x}{\tbase}{\expr}$. 
\begin{enumerate}
\item We have \hastype{\env}{\val_{12}}{\eqrt{\typ}{\expr_1}{\expr_2}} and \hastype{\env}{\val_{23}}{\eqrt{\typ}{\expr_2}{\expr_3}}.
\item \label{trans:base-inversion} By canonical forms, 
  $\val_{12} = \ebeq{\tbase}\ \expr_1\ \expr_2\ \val_{12}'$ such that
        \hastype{\env}{\expr_1}{\typ_1} and
        \hastype{\env}{\expr_2}{\typ_2} (for some $\typ_1$ and $\typ_2$ that are refinements of $\tbase$) and 
        \hastype{\env, \envBind{r}{\typ_2}, \envBind{l}{\typ_1}}{\val_{12}'}{\tref{x}{\tunit}{\bbbeq{l}{r}{\tbase}}}.
  and, similarly,
  $\val_{23} = \ebeq{\tbase}\ \expr_1\ \expr_2\ \val_{23}'$ such that
        \hastype{\env}{\expr_2}{\typ_2'} and
        \hastype{\env}{\expr_3}{\typ_3} (for some $\typ_2'$ and $\typ_3$ that are refinements of $\tbase$) and 
        \hastype{\env, \envBind{r}{\typ_3}, \envBind{l}{\typ_2'}}{\val_{23}'}{\tref{x}{\tunit}{\bbbeq{l}{r}{\tbase}}}.
\item \label{trans:base-pf-inversion} By canonical forms again, we
  know that $\val_{12}' = \val_{23}' = \eunit$ and we have:
  \[ \begin{array}{l}
     \issubtype{\env, \envBind{r}{\typ_2}, \envBind{l}{\typ_1}}{\tref{\rbind}{\tunit}{\bbbeq{\rbind}{\eunit}{\tunit}}}{\tref{x}{\tbase}{\tref{x}{\tunit}{\bbbeq{l}{r}{\tbase}}}} \text{, and} \\
     \issubtype{\env, \envBind{r}{\typ_3}, \envBind{l}{\typ_2'}}{\tref{\rbind}{\tunit}{\bbbeq{\rbind}{\eunit}{\tunit}}}{\tref{x}{\tbase}{\tref{x}{\tunit}{\bbbeq{l}{r}{\tbase}}}}. \\
  \end{array} \]
\item \label{trans:base-sub-inversion}
  Elaborating on (\ref{trans:base-pf-inversion}), we know that $\forall \model\in\interp{\env, \envBind{r}{\typ_2}, \envBind{l}{\typ_1}}$, we have:
  \[ \interp{\modelapp{\model}{\tref{x}{\tunit}{\bbbeq{\rbind}{\eunit}{\tunit}}}} \subseteq \interp{\modelapp{\model}{\tref{x}{\tunit}{\bbbeq{l}{r}{\tbase}}}} \]
  and $\forall \model\in\interp{\env, \envBind{r}{\typ_3}, \envBind{l}{\typ_2'}}$, we have:
  \[ \interp{\modelapp{\model}{\tref{x}{\tunit}{\bbbeq{\rbind}{\eunit}{\tunit}}}} \subseteq \interp{\modelapp{\model}{\tref{x}{\tunit}{\bbbeq{l}{r}{\tbase}}}}. \]
\item \label{trans:base-eq-interp} Since
  ${\tref{\rbind}{\tunit}{\bbbeq{\rbind}{\eunit}{\tunit}}}$ contains
  all computations that terminate with \eunit in all models
  (Theorem~\ref{theorem:constant-property}), the right-hand sides of
  the equations must also hold all unit computations. That is, all choices for
  $l$ and $r_2$ (resp. $l$ and $r$) that are semantically well
  typed are necessarily equal.
\item \label{trans:base-interp-trans} By (\ref{trans:base-eq-interp}),
  we can infer that in any given model, $\typ_1$, $\typ_2$, $\typ_2'$,
  and $\typ_3$ identify just one $\tbase$-constant.
  Why must $\typ_2$ and $\typ_2'$ agree? In particular, $\expr_2$ has
  \emph{both} of those types, but by semantic soundness
  (Theorem~\ref{proofs-thm:soundness}), we know that it will go to a
  value in the appropriate type interpretation.
  By determinism of evaluation, we know it must be the \emph{same} value.
  We can therefore conclude that $\forall \model\in\interp{\env, \envBind{r}{\typ_3}, \envBind{l}{\typ_1}}$, 
  $\interp{\modelapp{\model}{\tref{x}{\tunit}{\bbbeq{\rbind}{\eunit}{\tunit}}}} \subseteq
   \interp{\modelapp{\model}{\tref{x}{\tunit}{\bbbeq{l}{r}{\tbase}}}}$.

\item By \teqbase, using $\typ_1$ and $\typ_3$ and $\eunit$ as the
  proof. We need to show \hastype{\env, \envBind{r}{\typ_3},
    \envBind{l}{\typ_1}}{\eunit}{\tref{x}{\tunit}{\bbbeq{l}{r}{\tbase}}};
  all other premises follow from (\ref{trans:base-inversion}).
\item By \tsub and \subBase, using (\ref{trans:base-interp-trans}) for the subtyping.
\end{enumerate}

\item $\typ \doteq \tfun{x}{\typ_x}{\typ'}$.
\begin{enumerate}
\item We have \hastype{\env}{\val_{12}}{\eqrt{\typ}{\expr_1}{\expr_2}} and \hastype{\env}{\val_{23}}{\eqrt{\typ}{\expr_2}{\expr_3}}.

\item \label{trans:fun-inversion} By canonical forms, we have%
\[ \begin{array}{rcl}
  \val_{12} &=& \exeq{x}{\typ_x}{\typ'}\ \expr_1\ \expr_2\ \val_{12}' \\
  \val_{23} &=& \exeq{x}{\typ_x}{\typ'}\ \expr_2\ \expr_3\ \val_{23}' \\
\end{array} \]
where there exist types $\typ_1$, $\typ_2$, $\typ_2'$, and $\typ_3$ subtypes of $\tfun{x}{\typ_x}{\typ'}$ such that
\[ \begin{array}{lr}
\hastype{\env}{\expr_1}{\typ_1} &
\hastype{\env}{\expr_2}{\typ_2} \\
\hastype{\env}{\expr_2}{\typ_2'} &
\hastype{\env}{\expr_3}{\typ_3} \\
\end{array} \]
and there exist types $\typ_{x_{12}}$, $\typ_{x_{23}}$, $\typ'_{12}$, and $\typ'_{23}$ such that 
\[ \begin{array}{l}
\hastype{\env, \envBind{r}{\typ_2}, \envBind{l}{\typ_1}}{\val_{p_{12}}}{\tfun{x}{\typ_{x_{12}}}{\eqrt{\typ'_{12}}{l\ x}{r\ x}}}, \\
\issubtype{\env, \envBind{r}{\typ_2}, \envBind{l}{\typ_1}}{\typ_x}{\typ_{x_{12}}}, \\
\issubtype{\env, \envBind{r}{\typ_2}, \envBind{l}{\typ_1}, \envBind{x}{\typ_x}}{\typ'_{12}}{\typ'}, \\
\hastype{\env, \envBind{r}{\typ_3}, \envBind{l}{\typ_2'}}{\val_{p_{23}}}{\tfun{x}{\typ_x'}{\eqrt{\typ'_{23}}{l\ x}{r\ x}}}, \\
\issubtype{\env, \envBind{r}{\typ_3}, \envBind{l}{\typ_2'}}{\typ_x}{\typ_{x_{23}}}, \text{ and} \\
\issubtype{\env, \envBind{r}{\typ_3}, \envBind{l}{\typ_2'}, \envBind{x}{\typ_x}}{\typ'_{23}}{\typ'}.
\end{array} \]

\item \label{trans:fun-pf-inversion} By canonical forms on
  $\val_{p_{12}}$ and $\val_{p_{23}}$ from
  (\ref{trans:fun-inversion}), we know that:
\[ \val_{p_{12}} = \elam{x}{\typ_{x_{12}}}{\expr_{12}'} \qquad
   \val_{p_{23}} = \elam{x}{\typ_{x_{23}}}{\expr_{23}'}
\]
such that:
\[ \begin{array}{l}
\hastype{\env, \envBind{r}{\typ_2}, \envBind{l}{\typ_1}, \envBind{x}{\typ_{x_{12}}}}{\expr_{12}'}{\typ_{12}''}, \\
\issubtype{\env, \envBind{r}{\typ_2}, \envBind{l}{\typ_1}, \envBind{x}{\typ_{x_{12}}}}{\typ_{12}''}{\typ'_{12}}, \\
\\
\hastype{\env, \envBind{r}{\typ_3}, \envBind{l}{\typ_2'}, \envBind{x}{\typ_{x_{23}}}}{\expr_{23}'}{\typ_{23}''}, \\
\issubtype{\env, \envBind{r}{\typ_3}, \envBind{l}{\typ_2'}, \envBind{x}{\typ_{x_{23}}}}{\typ_{23}''}{\typ'_{23}}, \text{ and} \\
\end{array} \]

\item \label{trans:fun-pf-stronger} By strengthening
  (Lemma~\ref{lem:sub-strengthening}) using
  (\ref{trans:fun-inversion}), we can replace $x$'s type with $\typ_x$
  in both proofs, to find:
\[ \begin{array}{l}
\hastype{\env, \envBind{r}{\typ_2}, \envBind{l}{\typ_1}, \envBind{x}{\typ_x}}{\expr_{12}'}{\typ_{12}'}, \text{and} \\
\hastype{\env, \envBind{r}{\typ_3}, \envBind{l}{\typ_2'}, \envBind{x}{\typ_x}}{\expr_{23}'}{\typ_{23}'}. \\
\end{array} \]
Then, by \tsub, we can relax the type of the proof bodies:
\[ \begin{array}{l}
\hastype{\env, \envBind{r}{\typ_2}, \envBind{l}{\typ_1}, \envBind{x}{\typ_x}}{\expr_{12}'}{\typ'}, \text{and} \\
\hastype{\env, \envBind{r}{\typ_3}, \envBind{l}{\typ_2'}, \envBind{x}{\typ_x}}{\expr_{23}'}{\typ'}. \\
\end{array} \]

\item \label{trans:fun-pf-ih} By (\ref{trans:fun-pf-stronger},
  (\ref{trans:fun-pf-inversion}), and the IH on
  \eqrt{\typ'}{l\ x}{r\ x}, we know there exists some proof body
  $\expr_{13}'$ such that \hastype{\env, \envBind{r}{\typ_3},
    \envBind{l}{\typ_1}}{\expr_{13}'}{\eqrt{\typ'}{l\ x}{r\ x}}.
\item Let $\val_p = \tfun{x}{\typ_x}{\expr_{13}'}$.
\item \label{trans:fun-pf} By (\ref{trans:fun-pf-ih}), and \tlam.
\item Let $\val_{13} = \exeq{x}{\typ_x}{\typ'}\ \expr_1\ \expr_3\ \val_p$.
\item By \teqbase, with (\ref{trans:fun-pf}) for the proof and
  (\ref{trans:fun-inversion}) for the rest.
\end{enumerate}

\item $\typ \doteq \eqrt{\typ'}{\expr_1}{\expr_2}$. These types are not equable, so we ignore them. \qedhere
\end{itemize}

\end{proof}

\section{Parallel reduction and cotermination}\label{app:parred}

The conventional application rule for dependent types substitutes a
term into a type, finding $\expr_1 ~ \expr_2 : \tau\subst{x}{\expr_2}$
when $\expr_1 : \tfun{x}{\tau_x}{\tau}$.
We define two logical relations: a unary interpretation of types
(Figure~\ref{fig:semantic-typing})
and a binary logical relation characterizing equivalence
(Figure~\ref{fig:equivalence-def}).
Both of these logical relations are defined as fixpoints on types.
The type index poses a problem: the function case of these logical
relations quantify over values in the relation, but we sometimes need
to reason about expressions, not values.
If \goesto{\expr}{\val}, are $\typ\subst{x}{\expr}$ 
and $\typ\subst{x}{\val}$ treated the same by our logical relations?
We encounter this problem in particular in proof of logical relation
compositionality, which is precisely about exchanging expressions in
types with the values the expressions reduce to in closing
substitutions: \todo{ref somewhere for the unary, once we've done
  those proofs} for the unary logical relation and binary logical
relation (Lemma~\ref{proofs-lemma:apptyp}).

The key technical device to prove these compositionality lemmas is
\emph{parallel reduction} (Figure~\ref{fig:parred}). Parallel
reduction generalizes our call-by-value relation to allow multiple
steps at once, throughout a term---even under a lambda.
Parallel reduction is a bisimulation
(Lemma~\ref{lem:parred-forward-simulation} for forward simulation;
Corollary~\ref{cor:parred-backward-simulation} for backward simulation).
That is, expressions that parallel reduce to each other go to
identical constants or expressions that themselves parallel reduce,
and the logical relations put terms that parallel reduce in the same
equivalence class.

To prove the compositionality lemmas, we first show that (a) the
logical relations are closed under parallel reduction (\todo{cite} for
the unary relation and Lemma~\ref{proofs-lemma:lr-parred} for the
binary relation) and (b) use the backward simulation to change values
in the closing substitution to a substituted expression in the type.

Our proof comes in three steps.
First, we establish some basic properties of parallel reduction (\S\ref{sec:parred-properties}).
Next, proving the forward simulation is straightforward (\S\ref{sec:parred-forward}): if
\parreds{\expr_1}{\expr_2} and \evals{\expr_1}{\expr_1'}, then either
parallel reduction contracted the redex for us and
\parreds{\expr_1'}{\expr_2} immediately, or the redex is preserved and
\evals{\expr_2}{\expr_2'} such that \parreds{\expr_1'}{\expr_2'}.
Proving the backward simulation is more challenging
(\S\ref{sec:parred-backward}). If \parreds{\expr_1}{\expr_2} and
\evals{\expr_2}{\expr_2'}, the redex contracted in $\expr_2$ may not yet be exposed.
The trick is to show a tighter bisimulation, where the outermost
constructors are always the same, with the subparts parallel reducing.
We call this relation \emph{congruence} (Figure~\ref{fig:congruence});
it's a straightforward restriction of parallel reduction, eliminating
$\beta$, eq1, and eq2 as outermost constructors (but allowing them
deeper inside).
The key lemma shows that if \parreds{\expr_1}{\expr_2}, then there
exists $\expr_1'$ \goesto{\expr_1}{\expr_1'} such that
\congruent{\expr_1'}{\expr_2} (Lemma~\ref{lem:parred-congruent}).
Once we know that parallel reduction implies reduction to congruent
terms, proving that congruence is a backward simulation allows us to
reason ``up to congruence''. In particular, congruence is a
sub-relation of parallel reduction, so we find that parallel reduction
is a backward simulation.
Finally, we can show that \parreds{\expr_1}{\expr_2} implies
observational equivalence (\S\ref{sec:parred-cotermination});
for our purposes, it suffices to find cotermination at constants
(Corollary~\ref{cor:cotermination-multi}).

One might think, in light of Takahashi's explanation of parallel
reduction~\cite{DBLP:journals/jsc/Takahashi89}, that the simulation
techniques we use are too powerful for our needs: why not simply rely
on the Church-Rosser property and confluence, which she proves quite
simply?
Her approach works well when relating parallel reduction to full
$\beta$-reduction (and/or $\eta$-reduction):
the transitive closure of her parallel reduction relation is equal to
the transitive closure of plain $\beta$-reduction (resp. $\eta$- and
$\beta\eta$-reduction).
But we're interested in programming languages, so our underlying
reduction relation isn't full $\beta$: we use call-by-value, and we
will never reduce under lambdas. But even if we were call-by-name, we
would have the same issue.
Parallel reduction implies reduction, but not to the \emph{same}
value, as in her setting. Parallel reduction yields values that are
equivalent, \emph{up to parallel reduction and congruence} (see, e.g.,
Corollary~\ref{cor:parred-value}).

\begin{figure}
\hfill\fbox{$\parreds{\expr}{\expr}$}

\[
\inference
  {}
  {\parreds{x}{x}}[var]
\quad
\inference
  {}
  {\parreds{\con}{\con}}[const]
\quad
\inference
  {\parreds{\typ}{\typ'} \quad
   \parreds{\expr}{\expr'}}
  {\parreds{\elam{x}{\typ}{\expr}}{\elam{x}{\typ'}{\expr'}}}[lam]
\quad
\inference
  {\parreds{\expr_1}{\expr_1'} \quad 
   \parreds{\expr_2}{\expr_2'}}
  {\parreds{\expr_1\ \expr_2}{\expr_1'\ \expr_2'}}[app]
\]

\[
\inference
  {\parreds{\expr}{\expr'} \quad
   \parreds{\val}{\val'}}
  {\parreds{(\elam{x}{\typ}{\expr})\ \val}{\expr'\subst{x}{\val'}}}[$\beta$]
\quad
\inference
  {}
  {\parreds{\bbbeq{(}{)\ \con_1}{\tbase}}{\bbbeq{(}{)}{(\con_1,\tbase)}}}[eq1]
\quad
\inference
  {}
  {\parreds{\bbbeq{(}{)\ \con_2}{(\con_1,\tbase)}}{\con_1 = \con_2}}[eq2]
\]

\[
\inference
  {\parreds{\expr_l}{\expr_l'} \quad
   \parreds{\expr_r}{\expr_r'} \quad
   \parreds{\expr}{\expr'}
  }
  {\parreds{\ebeq{\tbase}\ \expr_l\ \expr_r\ \expr}{\ebeq{\tbase}\ \expr_l'\ \expr_r'\ \expr'}}[beq]
\quad
\inference
  {\parreds{\typ_x}{\typ_x'} \quad
   \parreds{\typ}{\typ'} \quad
   \parreds{\expr_l}{\expr_l'} \quad
   \parreds{\expr_r}{\expr_r'} \quad
   \parreds{\expr}{\expr'}
  }
  {\parreds
    {\exeq{x}{\typ_x}{\typ}\ \expr_l\ \expr_r\ \expr}
    {\exeq{x}{\typ_x'}{\typ'}\ \expr_l'\ \expr_r'\ \expr'}}[xeq]
\]

\hfill\fbox{$\parreds{\typ}{\typ}$}

\[
\inference
  {\parreds{\refa}{\refa'}}
  {\parreds{\tref{\rbind}{\tbase}{\refa}}{\tref{\rbind}{\tbase}{\refa'}}}[ref]
\quad
\inference
  {\parreds{\typ_x}{\typ_x'} \quad
   \parreds{\typ}{\typ'}
  }
  {\parreds{\tfun{x}{\typ_x}{\typ}}{\tfun{x}{\typ_x'}{\typ'}}}[fun]
\]

\[
\inference
  {\parreds{\typ}{\typ'} \quad
   \parreds{\expr_l}{\expr_l'} \quad
   \parreds{\expr_r}{\expr_r'}
  }
  {\parreds{\eqrt{\typ}{\expr_l}{\expr_r}}{\eqrt{\typ'}{\expr_l'}{\expr_r'}}}[eq]
\]

\caption{Parallel reduction in terms and types.}
\label{fig:parred}
\end{figure}

\subsection{Basic Properties}
\label{sec:parred-properties}

\begin{lemma}[Parallel reduction is reflexive]\label{lem:parred-refl}
  For all \expr and \typ, \parreds{\expr}{\expr} and \parreds{\typ}{\typ}.
\begin{proof}
  By mutual induction on \expr and \typ.
\paragraph*{Expressions}
\begin{itemize}
\item $\expr \doteq x$. By var.
\item $\expr \doteq c$. By const.
\item $\expr \doteq \elam{x}{\typ}{\expr'}$. By the IHs on $\typ$ and $\expr'$ and lam.
\item $\expr \doteq \expr_1 ~ \expr_2$. By the IH on $\expr_1$ and $\expr_2$ and app.
\item $\expr \doteq \ebeq{\tbase}\ \expr_l\ \expr_r\ \expr'$. By the IHs on
  $\expr_l$, $\expr_r$, and $\expr'$ and beq.
\item $\expr \doteq \exeq{x}{\typ_x}{\typ}\ \expr_l\ \expr_r\ \expr'$. By the IHs on
  $\typ_x$, $\typ$, $\expr_l$, $\expr_r$, and $\expr'$ and xeq. 
\end{itemize}

\paragraph*{Types}
\begin{itemize}
\item $\typ \doteq \tref{\rbind}{\tbase}{\refa}$. By the IH on $\refa$ (an expression) and ref.
\item $\typ \doteq \tfun{x}{\typ_x}{\typ'}$. By the IHs on $\typ_x$ and $\typ'$ and fun. 
\item $\typ \doteq \eqrt{\typ'}{\expr_l}{\expr_r}$. By the IHs on $\typ'$, $\expr_l$, and $\expr_r$ and eq.
\qedhere
\end{itemize}
\end{proof}
\end{lemma}

\begin{lemma}[Parallel reduction is substitutive]\label{lem:parred-subst}
  If  \parreds{\expr}{\expr'}, then:
  \begin{enumerate}
  \item If \parreds{\expr_1}{\expr_2}, then \parreds{\expr_1\subst{x}{\expr}}{\expr_2\subst{x}{\expr'}}.
  \item If \parreds{\typ_1}{\typ_2}, then \parreds{\typ_1\subst{x}{\expr}}{\typ_2\subst{x}{\expr'}}.
  \end{enumerate}
\begin{proof}
  By mutual induction on $\expr_1$ and $\typ_1$.
\paragraph*{Expressions}
\begin{itemize}
\item[var] \parreds{y}{y}. 
If $y \ne x$, then the substitution has no effect and the case is trivial.
If $y = x$, then $x\subst{x}{\expr} = \expr$ and we have \parreds{\expr}{\expr'} by assumption.
We have \parreds{\expr}{\expr} by reflexivity (Lemma~\ref{lem:parred-refl}).
\item[const] \parreds{\con}{\con}. This case is trivial: the substitution has no effect.
\item[lam] \parreds{\elam{y}{\typ}{\expr'}}{\elam{y}{\typ}{\expr''}}.
If $y \ne x$, then by the IH on $\expr'$ and lam.
If $y = x$, then the substitution has no effect and the case is trivial.
\item[app] \parreds{\expr_{11}\ \expr_{12}}{\expr_{21}\ \expr_{22}}, where
\parreds{\expr{1i}}{\expr{2i}} for $i = 1, 2$.
By the IHs on $\expr_{1i}$ and app.
\item[beta] \parreds{(\elam{y}{\typ}{\expr'})\ \val}{\expr'\subst{y}{\val'}}, where
\parreds{\expr'}{\expr''} and \parreds{v}{v'}.
If $y \ne x$, then 
\parreds{(\elam{y}{\typ}{\expr'\subst{x}{\expr}})\ \val\subst{x}{\expr}}
        {\expr''\subst{x}{\expr}\subst{y}{\val'\subst{x}{\expr}}}
by $\beta$.
Since $y \ne x$, $\expr''\subst{x}{\expr}\subst{y}{\val'\subst{x}{\expr}} = \expr''\subst{y}{\val'}\subst{x}{\expr}$ as desired.

If $y = x$, then the substitution in the lambda has no effect, and we find
\parreds{(\elam{x}{\typ}{\expr'})\ \val\subst{x}{\expr}}
        {\expr''\subst{x}{\val'\subst{x}{\expr}}}
by $\beta$.
We have $\expr''\subst{x}{\val'\subst{x}{\expr}} =
\expr''\subst{x}{\val'}\subst{x}{\expr}$ as desired.
\item[eq1] \parreds{\bbbeq{(}{)\ \con_1}{\tbase}}{\bbbeq{(}{)}{(\con_1,\tbase)}}. This case is trivial by eq1, as the substitution has no effect.
\item[eq2] \parreds{\bbbeq{(}{)\ \con_2}{(\con_1,\tbase)}}{\con_1 = \con_2}. This case is trivial by eq2, as the substitution has no effect.
\item[beq] \parreds{\ebeq{\tbase}\ \expr_l\ \expr_r\ \expr_p}{\ebeq{\tbase}\ \expr_l'\ \expr_r'\ \expr_p'}, where
\parreds{\expr_l}{\expr_l'} and \parreds{\expr_r}{\expr_r'} and \parreds{\expr_p}{\expr_p'}.
By the IHs on $\expr_l$, $\expr_r$, and $\expr_p$ and beq.
\item[xeq] \parreds{\exeq{x}{\typ_x}{\typ}\ \expr_l\ \expr_r\ \expr_p}{\exeq{x}{\typ_x}{\typ}\ \expr_l'\ \expr_r'\ \expr_p'}, where
\parreds{\expr_l}{\expr_l'} and \parreds{\expr_r}{\expr_r'} and \parreds{\expr_p}{\expr_p'}.
By the IHs on $\expr_l$, $\expr_r$, and $\expr_p$ and xeq. 
\end{itemize}

\paragraph*{Types}
\begin{itemize}

\item[ref] \parreds{\tref{y}{\tbase}{\refa}}{\tref{y}{\tbase}{\refa'}} where \parreds{\refa}{\refa'}.
If $y \ne x$, then \parreds{\refa\subst{x}{\expr}}{\refa'\subst{x}{\expr'}} by the IH on $\refa$; we are done by ref.

If $y = x$, then the substitution has no effect, and the case is immediate by reflexivity (Lemma~\ref{lem:parred-refl}).

\item[fun] \parreds{\tfun{y}{\typ_y}{\typ}}{\tfun{y}{\typ_y'}{\typ'}}
  where \parreds{\typ_y}{\typ_y'} and \parreds{\typ}{\typ'}.
If $y \ne x$, then by the IH on $\typ_y$ and $\typ$ and fun.

If $y = x$, then the substitution only has effect in the domain. The
IH on $\typ_y$ finds \parreds{\typ_y\subst{x}{\expr}}{\typ_y'\subst{x}{\expr'}} in
the domain; reflexivity covers the codomain
(Lemma~\ref{lem:parred-refl}), and we are done by fun.

\item[eq] \parreds{\eqrt{\typ}{\expr_l}{\expr_r}}{\eqrt{\typ'}{\expr_l'}{\expr_r'}}.
By the IHs and eq.
\qedhere
\end{itemize}
\end{proof}
\end{lemma}

\begin{corollary}[Substituting multiple parallel reduction is parallel reduction]\label{cor:parred-subst-eval-multi}
  If \parredsto{\expr_1}{\expr_2}, then \parredsto{\expr\subst{x}{\expr_1}}{\expr\subst{x}{\expr{2}}}.
\begin{proof}
  First, notice that \parreds{\expr}{\expr} by reflexivity (Lemma~\ref{lem:parred-refl}).
  By induction on \parredsto{\expr_1}{\expr_2}, using reflexivity in
  the base case (Lemma~\ref{lem:parred-refl}); the inductive step uses
  substituting parallel reduction (Lemma~\ref{lem:parred-subst})
  and the IH.
\end{proof}
\end{corollary}

\begin{lemma}[Parallel reduction subsumes reduction]\label{lem:parred-eval}
  If \evals{\expr_1}{\expr_2} then \parreds{\expr_1}{\expr_2}.
\begin{proof}
  By induction on the evaluation derivation, using reflexivity of
  parallel reduction to cover expressions and types that didn't step (Lemma~\ref{lem:parred-refl}).
\begin{itemize}
\item[ctx] \evals{\ctxapp{\ectx}{\expr}}{\ctxapp{\ectx}{\expr'}}, where \evals{\expr}{\expr'}.
By the IH, \parreds{\expr}{\expr'}.
By structural induction on $\ectx$.
\begin{itemize}
\item $\ectx \doteq \bullet$. By the outer IH.
\item $\ectx \doteq \ectx_1\ \expr_2$. By the inner IH on $\ectx_1$, reflexivity on $\expr_2$, and app.
\item $\ectx \doteq \val_1\ \ectx_2$. By reflexivity on $\val_1$, the inner IH on $\ectx_2$, and app.
\item $\ectx \doteq \ebeq{\tbase}\ \expr_l\ \expr_r\ \ectx'$.
  By reflexivity on $\expr_l$ and $\expr_r$, the inner IH on and $\ectx'$, and beq.
\item $\ectx \doteq \exeq{x}{\typ_x}{\typ}\ \expr_l\ \expr_r\ \ectx'$.
  By reflexivity on $\typ_x$, $\typ$, $\expr_l$ and $\expr_r$, the inner IH on and $\ectx'$, and xeq.
\end{itemize}
\item[$\beta$] \evals{(\elam{x}{\typ}{\expr})\ \val}{\expr\subst{x}{\val}}.
By reflexivity (Lemma~\ref{lem:parred-refl}, \parreds{\expr}{\expr} and \parreds{\val}{\val}.
By beta, \parreds{(\elam{x}{\typ}{\expr})\ \val}{\expr\subst{x}{\val}}.
\item[eq1] By eq1.
\item[eq2] By eq2. \qedhere
\end{itemize}
\end{proof}
\end{lemma}

\subsection{Forward Simulation}
\label{sec:parred-forward}

\begin{lemma}[Parallel reduction is a forward simulation]\label{lem:parred-forward-simulation}
  If \parreds{\expr_1}{\expr_2} and \evals{\expr_1}{\expr_1'}, then
  there exists $\expr_2'$ such that \goesto{\expr_2}{\expr_2'} and
  \parreds{\expr_1'}{\expr_2'}.
\begin{proof}
By induction on the derivation of \evals{\expr_1}{\expr_1'}, leaving $\expr_2$ general.
\begin{itemize}
\item[ctx]
By structural induction on $\ectx$, using reflexivity (Lemma~\ref{lem:parred-refl}) on parts where the IH doesn't apply.
\begin{itemize}
\item $\ectx \doteq \bullet$. By the outer IH on the actual step.
\item $\ectx \doteq \ectx_1\ \expr_2$. By the IH on $\ectx_1$, reflexivity on $\expr_2$, and app.
\item $\ectx \doteq \val_1\ \ectx_2$. By reflexivity on $\val_1$, the IH on $\ectx_2$, and app.
\item $\ectx \doteq \ebeq{\tbase}\ \expr_l\ \expr_r\ \ectx'$.
By reflexivity on $\expr_l$ and $\expr_r$, the IH on $\ectx'$, and beq.
\item $\ectx \doteq \exeq{x}{\typ_x}{\typ}\ \expr_l\ \expr_r\ \ectx'$.
By reflexivity on $\typ_x$, $\typ$, $\expr_l$ and $\expr_r$, the IH on $\ectx'$, and xeq.
\end{itemize}

\item[$\beta$] \evals{(\elam{x}{\typ}{\expr})\ \val}{\expr\subst{x}{\val}}.
One of two rules could have applied to find \parreds{\expr_1}{\expr_2}: app or $\beta$.

In the app case, we have $\expr_2 = (\elam{x}{\typ'}{\expr'})\ \val'$
where \parreds{\typ}{\typ'} and \parreds{\expr}{\expr'} and
\parreds{\val}{\val'}.
Let $\expr_2' = \expr'\subst{x}{\val'}$.
We find \goesto{\expr_2}{\expr_2'} in one step by $\beta$.
We find \parreds{\expr\subst{x}{\val}}{\expr'\subst{x}{\val'}} by
substitutivity of parallel reduction (Lemma~\ref{lem:parred-subst}).

In the $\beta$ case, we have $\expr_2 = \expr'\subst{x}{\val'}$ such
that \parreds{\expr}{\expr'} and \parreds{\val}{\val'}.
Let $\expr_2' = \expr_2$.
We find \goesto{\expr_2}{\expr_2'} in no steps at all;
we find \parreds{\expr_1'}{\expr_2'} by substitutivity of parallel
reduction (Lemma~\ref{lem:parred-subst}).

\item[eq1] \evals{\bbbeq{(}{)\ \con_1}{\tbase}}{\bbbeq{(}{)}{(\con_1,\tbase)}}.
One of two rules could have applied to find
\parreds{\bbbeq{(}{)\ \con_1}{\tbase}}{\expr_2}: app or eq1.

In the app case, we must have $\expr_2 = \expr_1 =
\bbbeq{(}{)\ \con_1}{\tbase}$, because there are no reductions
available in these constants.
Let $\expr_2' = \bbbeq{(}{)}{(\con_1,\tbase)}$.
We find \goesto{\expr_2}{\expr_2'} in a single step by our assumption (or eq1).
We find parallel reduction by reflexivity (Lemma~\ref{lem:parred-refl}).

In the eq2 case, we have $\expr_2 = \expr_1' = \bbbeq{(}{)}{(\con_1,\tbase)}$. Let $\expr_2' = \expr_2$.
We find \goesto{\expr_2}{\expr_2'} in no steps at all.
We find parallel reduction by reflexivity (Lemma~\ref{lem:parred-refl}).

\item[eq2] \evals{\bbbeq{(}{)\ \con_2}{(\con_1,\tbase)}}{\con_1 = \con_2}.
One of two rules could have applied to find
\parreds{\bbbeq{(}{)\ \con_2}{(\con_1,\tbase)}}{\expr_2}: app or eq2.

In the app case, we have $\expr_2 = \expr_1 =
{\bbbeq{(}{)\ \con_2}{(\con_1,\tbase)}}$, because there are no
reductions available in these constants. 
Let $\expr_2' \doteq \con_1 = \con_2$, i.e. \etrue when $\con_1 =
\con_2$ and \efalse otherwise.
We find \goesto{\expr_2}{\expr_2'} in a single step by our assumption (or eq2).
We find parallel reduction by reflexivity (Lemma~\ref{lem:parred-refl}).

In the eq2 case, we have $\expr_2 = \expr_1' \doteq \con_1 = \con_2$, i.e. \etrue
when $\con_1 = \con_2$ and \efalse otherwise.
Let $\expr_2' = \expr_2$.
We find \goesto{\expr_2}{\expr_2'} in no steps at all.
We find parallel reduction by reflexivity
(Lemma~\ref{lem:parred-refl}).
\qedhere

\end{itemize}
\end{proof}
\end{lemma}

\subsection{Backward Simulation}
\label{sec:parred-backward}

\begin{lemma}[Reduction is substitutive]\label{lem:red-subst}
If \evals{\expr_1}{\expr_2}, then \evals{\expr_1\subst{x}{\expr}}{\expr_2\subst{x}{\expr}}.
\begin{proof}
  By induction on the derivation of \evals{\expr_1}{\expr_2}.
\begin{itemize}
\item[ctx] By structural induction on $\ectx$.
\begin{itemize}
\item $\ectx \doteq \bullet$. By the outer IH.
\item $\ectx \doteq \ectx_1\ \expr_2$. By the IH on $\ectx_1$ and ctx.
\item $\ectx \doteq \val_1\ \ectx_2$. By the IH on $\ectx_2$ and ctx.
\item $\ectx \doteq \ebeq{\tbase}\ \expr_l\ \expr_r\ \ectx'$. By the IH on $\ectx'$ and ctx.
\item $\ectx \doteq \exeq{x}{\typ_x}{\typ}\ \expr_l\ \expr_r\ \ectx'$. By the IH on $\ectx'$ and ctx.
\end{itemize}
\item[$\beta$] \evals{(\elam{y}{\typ}{\expr'})\ \val}{\expr'\subst{y}{\val}}.
We must show
\evals{(\elam{y}{\typ}{\expr'})\subst{x}{\expr}\ \val\subst{x}{\expr}}{\expr'\subst{y}{\val}\subst{x}{\expr}}.

The exact result depends on whether $y = x$.
If $y \ne x$, the substitution goes through, and we have $(\elam{y}{\typ}{\expr'})\subst{x}{\expr} = \elam{y}{\typ\subst{x}{\expr}}{\expr'\subst{x}{\expr}}$.
By $\beta$, 
\evals{(\elam{y}{\typ\subst{x}{\expr}}{\expr'\subst{x}{\expr}})\ \val\subst{x}{\expr}}%
      {\expr'\subst{x}{\expr}\subst{y}{\val\subst{x}{\expr}}}.
But $\expr'\subst{x}{\expr}\subst{y}{\val\subst{x}{\expr}} =
\expr'\subst{y}{\val}\subst{x}{\expr}$, and we are done.

If, on the other hand, $y = x$, then the substitution has no effect in
the body of the lambda, and $(\elam{y}{\typ}{\expr'})\subst{x}{\expr}
= \elam{y}{\typ\subst{x}{\expr}}{\expr'}$.
By $\beta$ again, we find
\evals{(\elam{y}{\typ\subst{x}{\expr}}{\expr'})\ \val\subst{x}{\expr}}%
      {\expr'\subst{y}{\val\subst{x}{\expr}}}.
Since $y = x$, we really have
$\expr'\subst{x}{\val\subst{x}{\expr}}$ which is the same as
$\expr'\subst{x}{\val}\subst{x}{\expr} = \expr'\subst{y}{\val}\subst{x}{\expr} $, as desired.

\item[eq1] The substitution has no effect; immediate, by eq1.
\item[eq2] The substitution has no effect; immediate, by eq2. \qedhere
\end{itemize}
\end{proof}
\end{lemma}

\begin{corollary}[Multi-step reduction is substitutive]\label{cor:red-subst-multi}
If \goesto{\expr_1}{\expr_2}, then \goesto{\expr_1\subst{x}{\expr}}{\expr_2\subst{x}{\expr}}.
\begin{proof}
  By induction on the derivation of \goesto{\expr_1}{\expr_2}. The
  base case is immediate ($\expr_1 = \expr_2$, and we take no
  steps). The inductive case follows by the IH and single-step
  substitutivity (Lemma~\ref{lem:red-subst}).
\end{proof}
\end{corollary}

\begin{figure}
  \[
\inference
  {}
  {\congruent{x}{x}}[var]
\quad
\inference
  {}
  {\congruent{\con}{\con}}[const]
\quad
\inference
  {\parreds{\typ}{\typ'} \quad
   \parreds{\expr}{\expr'}}
  {\congruent{\elam{x}{\typ}{\expr}}{\elam{x}{\typ'}{\expr'}}}[lam]
\quad
\inference
  {\parreds{\expr_1}{\expr_1'} \quad 
   \parreds{\expr_2}{\expr_2'}}
  {\congruent{\expr_1\ \expr_2}{\expr_1'\ \expr_2'}}[app]
\]

\[
\inference
  {\parreds{\expr_l}{\expr_l'} \quad
   \parreds{\expr_r}{\expr_r'} \quad
   \parreds{\expr}{\expr'}
  }
  {\congruent{\ebeq{\tbase}\ \expr_l\ \expr_r\ \expr}{\ebeq{\tbase}\ \expr_l'\ \expr_r'\ \expr'}}[beq]
\quad
\inference
  {\parreds{\typ_x}{\typ_x'} \quad
   \parreds{\typ}{\typ'} \quad
   \parreds{\expr_l}{\expr_l'} \quad
   \parreds{\expr_r}{\expr_r'} \quad
   \parreds{\expr}{\expr'}
  }
  {\congruent
    {\exeq{x}{\typ_x}{\typ}\ \expr_l\ \expr_r\ \expr}
    {\exeq{x}{\typ_x'}{\typ'}\ \expr_l'\ \expr_r'\ \expr'}}[xeq]
\]

  \caption{Term congruence.}
  \label{fig:congruence}
\end{figure}

We say terms are \emph{congruent} when they (a) have the same
outermost constructor and (b) their subparts parallel reduce to each
other.\footnote{Congruent terms are related to Takahashi's $\tilde{M}$
  operator: in that they characterize parallel reductions that preserve
  structure. They are not the same, though: Takahashi's $\tilde{M}$
  will do $\beta\eta$-reductions on outermost redexes.}
That is, $\congruent{}{} \subseteq \parreds{}{}$, where the outermost
rule must be one of var, const, lam, app, beq, or xeq and cannot be a
\emph{real} reduction like $\beta$, eq1, or eq2.

Congruence is a key tool in proving that parallel reduction is a
backward simulation.
Parallel reductions under a lambda prevent us from having an
``on-the-nose'' relation, but reduction can keep up enough with
parallel reduction to maintain congruence.

\begin{lemma}[Congruence implies parallel reduction]\label{lem:congruence-parred}
  If \congruent{\expr_1}{\expr_2} then \parreds{\expr_1}{\expr_2}.
\begin{proof}
  By induction on the derivation of \congruent{\expr_1}{\expr_2}.
  \begin{itemize}
  \item[var] \congruent{x}{x}. By var.
  \item[const] \congruent{c}{c}. By const.
  \item[lam]
    \congruent{\elam{x}{\typ}{\expr}}{\elam{x}{\typ'}{\expr'}}, with
    \parreds{\typ}{\typ'} and \parreds{\expr}{\expr'}. By lam.
  \item[app] \congruent{\expr_1\ \expr_2}{\expr_1'\ \expr_2'}, with
    \parreds{\expr_1}{\expr_1'} and \parreds{\expr_2}{\expr_2'}. By
    app.
  \item[beq]
    \congruent{\ebeq{\tbase}\ \expr_l\ \expr_r\ \expr}{\ebeq{\tbase}\ \expr_l'\ \expr_r'\ \expr},
    with \parreds{\expr_l}{\expr_l'} and \parreds{\expr_r}{\expr_r'}
    and \parreds{\expr}{\expr'}.
    By beq.
  \item[xeq] By xeq.
    \congruent{\exeq{x}{\typ_x}{\typ}\ \expr_l\ \expr_r\ \expr}{\exeq{x}{\typ_x}{\typ}\ \expr_l'\ \expr_r'\ \expr},
    with \parreds{\typ_x}{\typ_x'} and \parreds{\typ}{\typ'} and
    \parreds{\expr_l}{\expr_l'} and \parreds{\expr_r}{\expr_r'} and
    \parreds{\expr}{\expr'}.
    By xeq.
    \qedhere
  \end{itemize}
\end{proof}
\end{lemma}

We need to strengthen substitutivity (Lemma~\ref{lem:parred-subst}) to
show that it preserves congruence.
\begin{corollary}[Congruence is substitutive]\label{cor:congruence-subst}
  If \congruent{\expr_1}{\expr_1'} and \congruent{\expr_2}{\expr_2'},
  then
  \congruent{\expr_1\subst{x}{\expr_2}}{\expr_2\subst{x}{\expr_2'}}.
\begin{proof}
  By cases on $\expr_1$.
\begin{itemize}
\item $\expr_1 = y$.
  It must be that $\expr_2 = y$ as well, since only var could have applied.
  If $y \ne x$, then the substitution has no effect and we have \congruent{y}{y} by assumption (or var).
  If $x = y$, then $\expr_1\subst{x}{\expr_2} = \expr_2$ and we have \congruent{\expr_2}{\expr_2'} by assumption.

\item $\expr_1 = c$.
  It must be that $\expr_2 = c$ as well.
  The substitution has no effect; immediate by var.

\item $\expr_1 = \elam{y}{\typ}{\expr}$.
  It must be that $\expr_2 = \elam{y}{\typ'}{\expr'}$ such that
  \parreds{\typ}{\typ'} and \parreds{\expr}{\expr'}.
  If $y \ne x$, then we must show
  \congruent{\elam{y}{\typ\subst{x}{\expr_2}}{\expr\subst{x}{\expr_2}}}
    {\elam{y}{\typ'\subst{x}{\expr_2'}}{\expr'\subst{x}{\expr_2'}}},
    which we have immediately by lam and Lemma~\ref{lem:parred-subst}
    on our two subparts.
  If $y = x$, then we must show
  \congruent{\elam{y}{\typ\subst{x}{\expr_2}}{\expr}}
            {\elam{y}{\typ'\subst{x}{\expr_2'}}{\expr'}},
  which we have immediately by lam, Lemma~\ref{lem:parred-subst} on
  our \parreds{\typ}{\typ'}, and the fact that
  \parreds{\expr}{\expr'}.
      
\item $\expr_1 = \expr_{11}\ \expr_{12}$. 
  It must be that $\expr_2 = \expr_{21}\ \expr_{22}$, such that
  \parreds{\expr_{11}}{\expr_{21}} and
  \parreds{\expr_{12}}{\expr_{22}}.
  By app and Lemma~\ref{lem:parred-subst} on the subparts.
\item $\expr_1 = \ebeq{\tbase}\ \expr_l\ \expr_r\ \expr$.
  It must be the case that $\expr_2 =
  \ebeq{\tbase}\ \expr_l'\ \expr_r'\ \expr'$ where
  \parreds{\expr_l}{\expr_l'} and \parreds{\expr_r}{\expr_r'}.
  By beq and Lemma~\ref{lem:parred-subst} on the subparts.
\item $\expr_1 = \exeq{x}{\typ_x}{\typ}\ \expr_l\ \expr_r\ \expr$.
  It must be the case that $\expr_2 =
  \exeq{x}{\typ_x'}{\typ'}\ \expr_l'\ \expr_r'\ \expr'$ where
  \parreds{\expr_l}{\expr_l'} (and similarly for $\typ_x$, $\typ$,
  $\expr_r$, and $\expr$).
  By xeq and Lemma~\ref{lem:parred-subst} on the subparts.
  \qedhere
\end{itemize}
\end{proof}
\end{corollary}

\begin{lemma}[Parallel reduction of values implies congruence]\label{lem:parred-value-congruent}
  If \parreds{\val_1}{\val_2} then \congruent{\val_1}{\val_2}.
\begin{proof}
  By induction on the derivation of \parreds{\val_1}{\val_2}.
  \begin{itemize}
  \item[var] Contradictory: variables aren't values.
  \item[const] Immediate, by const.
  \item[lam] Immediate, by lam.
  \item[app] Contradictory: applications aren't values.
  \item[beq] Immediate, by beq.
  \item[xeq] Immediate, by xeq.
  \item[$\beta$] Contradictory: applications aren't values.
  \item[eq1] Contradictory: applications aren't values.
  \item[eq2] Contradictory: applications aren't values. \qedhere
  \end{itemize}
\end{proof}
\end{lemma}

\begin{lemma}[Parallel reduction implies reduction to congruent forms]\label{lem:parred-congruent}
  If \parreds{\expr_1}{\expr_2}, then there exists $\expr_1'$
  \goesto{\expr_1}{\expr_1'} such that \congruent{\expr_1'}{\expr_2}.
\begin{proof}
By induction on \parreds{\expr_1}{\expr_2}.
\paragraph*{Structural rules}
\begin{itemize}
\item[var] \parreds{x}{x}. We have $\expr_1 = \expr_2 = x$ by var.
\item[const] \parreds{c}{c}. We have $\expr_1 = \expr_2 = c$ by const.
\item[lam] \parreds{\elam{x}{\typ}{\expr}}{\elam{x}{\typ'}{\expr'}},
  where \parreds{\typ}{\typ'} and \parreds{\expr}{\expr'}. Immediate, by lam.
\item[app] \parreds{\expr_{11}\ \expr_{12}}{\expr_{21}\ \expr_{22}},
  where \parreds{\expr_{11}}{\expr_{21}} and \parreds{\expr_{12}}{\expr_{22}}.
  Immediate, by app.
\item[beq] \parreds{\ebeq{\tbase}\ \expr_l\ \expr_r\ \expr}{\ebeq{\tbase}\ \expr_l'\ \expr_r'\ \expr'}
  where \parreds{\expr_l}{\expr_l'} and \parreds{\expr_r}{\expr_r'}
  and \parreds{\expr}{\expr'}. Immediate, by beq.
\item[xeq]
  \parreds{\exeq{x}{\typ_x}{\typ}\ \expr_l\ \expr_r\ \expr}{\exeq{x}{\typ_x'}{\typ'}\ \expr_l'\ \expr_r'\ \expr'}
  where \parreds{\typ_x}{\typ_x'} and \parreds{\typ}{\typ'} and
  \parreds{\expr_l}{\expr_l'} and \parreds{\expr_r}{\expr_r'} and
  \parreds{\expr}{\expr'}. Immediate, by xeq.
\end{itemize}

\paragraph*{Reduction rules}
These are the more interesting cases, where the parallel reduction
does a reduction step---ordinary reduction has to do more work to
catch up.
\begin{itemize}
\item[$\beta$]
  \parreds{(\elam{x}{\typ}{\expr})\ \val}{\expr'\subst{x}{\val'}},
  where \parreds{\expr}{\expr''} and \parreds{\val}{\val''}.

  We have \evals{(\elam{x}{\typ}{\expr})\ \val}{\expr\subst{x}{\val}} by $\beta$.
  By the IH on \parreds{\expr}{\expr''}, there exists $\expr'$ such
  that \goesto{\expr}{\expr'} such that \congruent{\expr'}{\expr''}.
  We \emph{ignore} the IH on \parreds{\val}{\val''}, noticing instead
  that parallel reducing values are congruent
  (Lemma~\ref{lem:parred-value-congruent}) and so
  \congruent{\val}{\val''}.
  Since reduction is substitutive
  (Corollary~\ref{cor:red-subst-multi}), we can find that
  \goesto{\expr\subst{x}{\val}}{\expr'\subst{x}{\val}}.
  Since congruence is substitutive (Lemma~\ref{cor:congruence-subst}),
  we have \congruent{\expr'\subst{x}{\val}}{\expr''\subst{x}{\val''}},
  as desired.

\item[eq1] \parreds{\bbbeq{(}{)\ \con_1}{\tbase}}{\bbbeq{(}{)}{(\con_1,\tbase)}}.
  We have
  \evals{\bbbeq{(}{)\ \con_1}{\tbase}}{\bbbeq{(}{)}{(\con_1,\tbase)}}
  in a single step; we find congruence by const.

\item[eq2] \parreds{\bbbeq{(}{)\ \con_2}{(\con_1,\tbase)}}{\con_1 = \con_2}.
  We have \evals{\bbbeq{(}{)\ \con_2}{(\con_1,\tbase)}}{\con_1 =
    \con_2} in a single step; we find congruence by const.
\qedhere
\end{itemize}
\end{proof}
\end{lemma}

\begin{lemma}[Congruence to a value implies reduction to a value]\label{lem:congruence-value}
  If \congruent{\expr}{\val'} then \goesto{\expr}{\val} such that
  \congruent{\val}{\val'}.
\begin{proof}
  By induction on $\val'$.
  \begin{itemize}
  \item $\val' \doteq c$. It must be the case that $\expr = c$. Let $\val = c$. By const.
  \item $\val' \doteq \elam{x}{\typ'}{\expr''}$. It must be the case
    that $\expr = \elam{x}{\typ}{\expr'}$ such that
    \parreds{\typ}{\typ'} and \parreds{\expr}{\expr''}. By lam.
  \item $\val \doteq \ebeq{\tbase}\ \expr_l'\ \expr_r'\ \val_p'$. It must
    be the case that $\expr = \ebeq{\tbase}\ \expr_l\ \expr_r\ \expr_p$ where 
    \parreds{\expr_l}{\expr_l'} and
    \parreds{\expr_r}{\expr_r'} and
    \parreds{\expr_p}{\val_p'}.
    Since parallel reduction implies reduction to congruent forms
    (Lemma~\ref{lem:parred-congruent}), we have
    \goesto{\expr_p}{\expr_p'} and \congruent{\expr_p'}{\val_p'}.
    By the IH on $\val_p'$, we know that \goesto{\expr_p'}{\val_p}
    such that \congruent{\val_p}{\val_p'}.
    By repeated use of ctx, we find
    \goesto{\ebeq{\tbase}\ \expr_l\ \expr_r\ \expr_p}{\ebeq{\tbase}\ \expr_l\ \expr_r\ \val_p}.
    Since its proof part is a value, this term is a value.
    We find
    \congruent{\ebeq{\tbase}\ \expr_l\ \expr_r\ \val_p}{\ebeq{\tbase}\ \expr_l'\ \expr_r'\ \val_p'}
    by ebeq.

  \item $\val \doteq \exeq{x}{\typ_x'}{\typ}\ \expr_l'\ \expr_r'\ \val_p'$. It must
    be the case that $\expr = \exeq{x}{\typ_x}{\typ}\ \expr_l\ \expr_r\ \expr_p$ where 
    \parreds{\typ_x}{\typ_x'} and
    \parreds{\typ}{\typ'} and
    \parreds{\expr_l}{\expr_l'} and
    \parreds{\expr_r}{\expr_r'} and
    \parreds{\expr_p}{\val_p'}.
    Since parallel reduction implies reduction to congruent forms
    (Lemma~\ref{lem:parred-congruent}), we have
    \goesto{\expr_p}{\expr_p'} and \congruent{\expr_p'}{\val_p'}.
    By the IH on $\val_p'$, we know that \goesto{\expr_p'}{\val_p}
    such that \congruent{\val_p}{\val_p'}.
    By repeated application of ctx, we find
    \goesto{\exeq{x}{\typ_x}{\typ}\ \expr_l\ \expr_r\ \expr_p}{\exeq{x}{\typ_x}{\typ}\ \expr_l\ \expr_r\ \val_p}.
    Since its proof part is a value, this term is a value.
    We find
    \congruent{\exeq{\typ_x}{\typ}\ \expr_l\ \expr_r\ \val_p}{\exeq{x}{\typ_x'}{\typ'}\ \expr_l'\ \expr_r'\ \val_p'}
    by exeq.
    \qedhere
  \end{itemize}
\end{proof} 
\end{lemma}

\begin{corollary}[Parallel reduction to a value implies reduction to a related value]\label{cor:parred-value}
  If \parreds{\expr}{\val'} then there exists $\val$ such that
  \goesto{\expr}{\val} and \congruent{\val}{\val'}.
\begin{proof}
  Since parallel reduction implies reduction to congruent forms
  (Lemma~\ref{lem:parred-congruent}), we have \goesto{\expr}{\expr'}
  such that \congruent{\expr'}{\val'}.
  But congruence to a value implies reduction to a value
  (Lemma~\ref{lem:congruence-value}), so \goesto{\expr'}{\val} such
  that \congruent{\val}{\val'}.
  By transitivity of reduction, \goesto{\expr}{\val}.
\end{proof}
\end{corollary}

\begin{lemma}[Congruence is a backward simulation]\label{lem:congruence-backward-simulation}
  If \congruent{\expr_1}{\expr_2} and \evals{\expr_2}{\expr_2'} then
  there exists $\expr_1'$ where \goesto{\expr_1}{\expr_1'} such that
  \congruent{\expr_1'}{\expr_2'}.
\begin{proof}
  By induction on the derivation of \evals{\expr_2}{\expr_2'}.
\begin{itemize}
\item[ctx] \evals{\ctxapp{\ectx}{\expr}}{\ctxapp{\ectx}{\expr'}}, where \evals{\expr}{\expr'}.
\begin{itemize}
\item $\ectx \doteq \bullet$. By the outer IH.
\item $\ectx \doteq \ectx_1\ \expr_2$. 
  It must be that $\expr_1 = \expr_{11}\ \expr_{12}$, where
  \parreds{\expr_{11}}{\ctxapp{\ectx_1}{\expr}} and
  \parreds{\expr_{12}}{\expr_2}.
  By the IH on $\ectx_1$, finding evaluation with ctx and congruence
  with app.
\item $\ectx \doteq \val_1'\ \ectx_2$.
  It must be that $\expr_1 = \expr_{11}\ \expr_{12}$, where
  \parreds{\expr_{11}}{\val_1'} and
  \parreds{\expr_{12}}{\ctxapp{\ectx_2}{\expr_2}}.
  We find that \goesto{\expr_{11}}{\val_1} such that
  \congruent{\val_1}{\val_1'} by Corollary~\ref{cor:parred-value}.
  By the IH on $\ectx_2$ and evaluation with ctx and congruence with
  app.
\item $\ectx \doteq \ebeq{\tbase}\ \expr_l'\ \expr_r'\ \ectx'$.
  It must be the case that $\expr_1 =
  \ebeq{\tbase}\ \expr_l\ \expr_r\ \expr_p$ where
  \parreds{\expr_l}{\expr_l'} and \parreds{\expr_r}{\expr_r'}. 
  By the IH on $\ectx'$; we find the evaluation with ctx and congruence
  with beq.
\item $\ectx \doteq \exeq{x}{\typ_x'}{\typ'}\ \expr_l'\ \expr_r'\ \ectx'$.
  It must be the case that $\expr_1 =
  \exeq{x}{\typ_x}{\typ}\ \expr_l\ \expr_r\ \expr_p$ such that
  \parreds{\typ_x}{\typ_x'} and \parreds{\typ}{\typ'} and
  \parreds{\expr_l}{\expr_l'} and \parreds{\expr_r}{\expr_r'}.
  By the IH on $\ectx'$; we find the evaluation with ctx and congruence
  with xeq.
\end{itemize}
\item[$\beta$] \evals{(\elam{x}{\typ'}{\expr'})\ \val'}{\expr'\subst{x}{\val'}}.
  Congruence implies that $\expr_1 = \expr_{11}\ \expr_{12}$ such that
  \parreds{\expr_{11}}{\elam{x}{\typ'}{\expr'}} and
  \parreds{\expr_{12}}{\val'}.
  Parallel reduction to a value implies reduction to a congruent value
  (Corollary~\ref{cor:parred-value}), \goesto{\expr_{11}}{\val_{11}}
  such that \congruent{\val_{11}'}{\elam{x}{\typ'}{\expr'}}, i.e.,
  $\val_{11} = \elam{x}{\typ}{\expr}$ such that \parreds{\typ}{\typ'}
  and \parreds{\expr}{\expr'}.
  Similarly, \goesto{\expr_{12}}{\val} such that \congruent{\val}{\val'}.

  By $\beta$, we find \goesto{(\elam{x}{\typ}{\expr})\ \val}{\expr'\subst{x}{\val}};
  by transitivity of reduction, we have \goesto{\expr_1 =
    \expr_{11}\ \expr_{12}}{\expr'\subst{x}{\val}}.
  Since congruence is substitutive (Corollary~\ref{cor:congruence-subst}), we have
  \congruent{\expr\subst{x}{\val}}{\expr'\subst{x}{\val'}}.

\item[eq1] \evals{\bbbeq{(}{)\ \con_1}{\tbase}}{\bbbeq{(}{)}{(\con_1,\tbase)}}.
  Congruence implies that $\expr_1 = \expr_{11}\ \expr_{12}$ such that
  \parreds{\expr_{11}}{\bbbeq{(}{)}{\tbase}} and
  \parreds{\expr_{12}}{\con_1}.
  Parallel reduction to a value implies reduction to a related value
  (Corollary~\ref{cor:parred-value}), \goesto{\expr_{11}}{\val_{11}}
  such that \congruent{\val_{11}}{\bbbeq{(}{)}{\tbase}} (and similarly
  for $\expr_{12}$ and $\con_1$).
  But the each constant is congruent only to itself, so
  $\val_{11} = \bbbeq{(}{)}{\tbase}$ and $\val_{12} = \con_1$.
  We have \evals{\bbbeq{(}{)\ \con_1}{\tbase}}{\bbbeq{(}{)}{(\con_1,\tbase)}} by assumption.
  So \goesto{\expr_1 =
    \expr_{11}\ \expr_{12}}{\bbbeq{(}{)}{(\con_1,\tbase)}} by
  transitivity, and we have congruence by const.

\item[eq2] \evals{\bbbeq{(}{)}{(\con_1,\tbase)}\ \con_2}{\con_1 = \con_2}.
  Congruence implies that $\expr_1 = \expr_{11}\ \expr_{12}$ such that
  \parreds{\expr_{11}}{\bbbeq{(}{)}{(\con_1,\tbase)}\ \con_2} and
  \parreds{\expr_{12}}{\con_2}.
  Parallel reduction to a value implies reduction to a related value
  (Corollary~\ref{cor:parred-value}), \goesto{\expr_{11}}{\val_{11}}
  such that \parreds{\val_{11}}{\bbbeq{(}{)}{(\con_1,\tbase)}\ \con_2} (and similarly
  for $\expr_{12}$ and $\con_2$).
  But the each constant is congruent only to itself, so
  $\val_{11} = \bbbeq{(}{)}{(\con_1,\tbase)}\ \con_2$ and $\val_{12} = \con_2$.
  We have \evals{\bbbeq{(}{)}{(\con_1,\tbase)}\ \con_2}{\con_1 = \con_2} already, by assumption.
  So \goesto{\expr_1 =
    \expr_{11}\ \expr_{12}}{\con_1 = \con_2} by
  transitivity, and we have congruence by const.
\qedhere

\end{itemize}
\end{proof}
  
\end{lemma}

\begin{corollary}[Parallel reduction is a backward simulation]\label{cor:parred-backward-simulation}
  If \parreds{\expr_1}{\expr_2} and \evals{\expr_2}{\expr_2'}, then
  there exists $\expr_1'$ such that \goesto{\expr_1}{\expr_1'} and
  \parreds{\expr_1'}{\expr_2'}.
\begin{proof}
Parallel reduction implies reduction to congruent forms, so
\goesto{\expr_1}{\expr_1'} such that \congruent{\expr_1'}{\expr_2}.
But congruence is a backward simulation (Lemma~\ref{lem:congruence-backward-simulation}), so \goesto{\expr_1'}{\expr_1''} such that \congruent{\expr_1''}{\expr_2'}.
By transitivity of evaluation, \goesto{\expr_1}{\expr_1''}.
Finally, congruence implies parallel reduction
(Lemma~\ref{lem:congruence-parred}), so \parreds{\expr_1''}{\expr_2'},
as desired.
\end{proof}
\end{corollary}

\subsection{Cotermination}
\label{sec:parred-cotermination}

\begin{theorem}[Cotermination at constants]\label{thm:cotermination}
If \parreds{\expr_1}{\expr_2} then \goesto{\expr_1}{c} iff \goesto{\expr_2}{c}.
\begin{proof}
By induction on the evaluation steps taken, using direct reduction in
the base case (Corollary~\ref{cor:parred-value}) and using parallel
reduction as a forward and backward simulation
(Lemmas~\ref{lem:parred-forward-simulation} and Corollary~\ref{cor:parred-backward-simulation}) in the inductive case.
\end{proof}
\end{theorem}

\begin{corollary}[Cotermination at constants (multiple parallel steps)]\label{cor:cotermination-multi}
If \parredsto{\expr_1}{\expr_2} then \goesto{\expr_1}{c} iff \goesto{\expr_2}{c}.
\begin{proof}
  By induction on the parallel reduction derivation. The base case is
  immediate ($\expr_1 = \expr_2$); the inductive case follows from
  cotermination at constants (Theorem~\ref{thm:cotermination}) and the
  IH.
\end{proof}
\end{corollary}

\end{document}